%% file: Roversi-SBVQ-computational-interpretation.tex
\documentclass{article}
\newtheorem{definition}{Definition}
\newtheorem{example}[definition]{Example}

\newtheorem{theorem}[definition]{Theorem}

\newtheorem{lemma}[definition]{Lemma}

\newtheorem{proposition}[definition]{Proposition}

\newtheorem{remark}[definition]{Remark}

\newtheorem{conjecture}[definition]{Conjecture}

\newtheorem{corollary}[definition]{Corollary}
\makeatletter
\renewcommand{\@begintheorem}[2]{ % not in italics
\trivlist\item[\hskip\labelsep{\bf #1\ #2}]}
\renewcommand{\@opargbegintheorem}[3]{
\trivlist\item[\hskip \labelsep{\bf #1\ #2\ (#3)}]}
\makeatother
\newtheorem{proof}{Proof}
\newcommand{\qed}{\hfill$\blacksquare$} %
\usepackage[T1]{fontenc} % correct output rendering as first packages
\usepackage{ae,aecompl}
\usepackage[latin1]{inputenc} % correct input typing as second package
\usepackage{stmaryrd}
\usepackage[lutzsyntax]{virginialake}\vlsmallbrackets\aftrianglefalse
\usepackage{amsmath,amssymb,latexsym}
\usepackage{paralist}
\usepackage{calrsfs}
\usepackage{hyperref}
\usepackage{graphicx}
\usepackage{txfonts} % supplies \multimapdots
\usepackage{color} %usage \textcolor{nameofhecolor}{texttocolor}
\usepackage{subfigure}

\usepackage{textcomp}
\usepackage{xspace}
\usepackage{setspace}
\usepackage{xypic}
\xyoption{arrow}

\input{macros}

\newcommand{\TITLE}{Communication, and concurrency with logic-based restriction
inside a calculus of structures}

\title
{\TITLE\\
{\small Luca Roversi}\\
{\small Universit\`a di Torino --- Dipartimento di Informatica\footnote
{{\it E-mail}: luca.roversi@unito.it
}}}

\begin{document}
\date{}
\maketitle

\begin{abstract}
It is well known that we can use structural proof theory to refine, or generalize, existing paradigmatic computational primitives, or to discover new ones. 
Under such a point of view we keep developing a programme whose goal is establishing a correspondence between proof-search of a logical system and computations in a process algebra.
We give a purely logical account of a process algebra operation which strictly includes the behavior of restriction on actions we find in Milner $\CCS$. 
This is possible inside a logical system in the Calculus of Structures of Deep Inference
endowed with a self-dual quantifier. Using proof-search of \cutfree proofs of such a logical system we show how to solve reachability problems in a process algebra that subsumes a significant fragment of Milner $\CCS$.
\end{abstract}
%%%%%%%%%%%%%%%%%%%%%%%%%%%%%%%%%%%%%%%%%%%%%%%%%%%%%%
\input{introduction}
\input{SBV2}
\input{SBV2-standardization-common}
  \input{SBV2-standardization-weak}
\input{SBV2-deduction}
\input{BV2-communication-core}
\input{PPi}
\input{PPi-reduced-to-BVTCC}
  \input{PPi-LTS-to-proof-search}
\input{BV2-soundness-wrt-PPi}
\input{conclusions}
%%%%%%%%%%%%%%%%%%%%%%%%%%%%%%%%%%%%%%%%%%%%%%%%%%%%%%%%%%%%%%%%%%%%%%%%%%
\bibliographystyle{plain}
\bibliography{Roversi-SBVQ-computational-interpretation}
%%%%%%%%%%%%%%%%%%%%%%%%%%%%%%%%%%%%%%%%%%%%%%%%%%%%%%%%%%%%%%%%%%%%%%%%%%
\appendix
\input{Appendix-BV2-commuting-conversions}
\input{Appendix-invertible-structures-are-invertible}
\input{Appendix-rightcontext-preserve-external-communication}
\input{Appendix-rightcontext-preserve-internal-communication}
\input{Appendix-soundness-wrt-internal-communication}
\input{Appendix-soundness-wrt-external-communication}
\end{document}

%% file: macros.tex
\newcommand\APEX[5]{
\ifthenelse{\boolean{proofintheappendix}}

{\ifthenelse{\equal{#2}{}}{\begin{#1}}{\begin{#1}[#2]}\label{#3}
\leavevmode\marginpar{\textit{\small Proof at page
\pageref{APXZ#3}.}}{#4}\end{#1}\immediate\write\tempfile{
    \unexpanded{\subsubsection*{Proof of \uppercase{#1} \ref{#3}, page
\pageref{#3}}\label{APXZ#3}{#5}}}
    }{\ifthenelse{\boolean{bodyafterstatement}}

{\ifthenelse{\equal{#2}{}}{\begin{#1}}{\begin{#1}[#2]}\label{#3}{#4}\end{#1}{#5}}

{\ifthenelse{\equal{#2}{}}{\begin{#1}}{\begin{#1}[#2]}\label{#3}{#4}\end{#1}}
 }}
\newcommand{\mapPincToDi}[1]
           {\llparenthesis\, {#1}\, \rrparenthesis} %
\newcommand{\mapDiToPinc}[2]
           {\llbracket\,\! #1\, \rrbracket_{#2}} %

\newcommand{\pincLTSJud}[3]{\xymatrix@1{#1\, \ar[r]^-{\ #3\ } & \,#2}} % with xypic
\newcommand{\pincLTSJudShort}[3]{#1\, \stackrel{#3}{\longrightarrow} \,#2} %
\newcommand{\pincLTSEXTJud}[5]
{\xymatrix@1{#1\, \ar[r#4]^-{\ #3\ } &#5 \,#2}} % with xypic
\newcommand{\pincRSJud}[3]{\xymatrix@1{#1\, \ar@{>>}[r]^-{\ #3\ } &\, #2}} %
\newcommand{\pincRSEXTJud}[5]
{\xymatrix@1{#1\, \ar@{>>}[r#4% 0 <= n occurrences of r
]^{\ #3\ } &#5% 0 <= n occurrences of &
\, #2}} %
\newcommand{\pincCong}{\approxeq} %
\newcommand{\pincact}{\mathsf{a}} %
\newcommand{\pinccom}{\mathsf{c}} %
\newcommand{\pinccntxp}{\mathsf{ctx}} %
\newcommand{\pincpi}{\mathsf{p_i}} %
\newcommand{\pincpe}{\mathsf{p_e}} %
\newcommand{\pinctran}{\mathsf{trn}} %
\newcommand{\pincrefl}{\mathsf{rfl}} %
\newcommand{\pincZer}{\mathbf{0}} %
\newcommand{\pincSec}[2]{{#1}.{#2}} %
\newcommand{\pincPar}[2]{#1\mid#2} %
\newcommand{\pincNu}[2]{#2|_{#1}} %
\newcommand{\pincNuPi}[2]{(\nu #1){#2}} %
\newcommand{\pincPreT}{\epsilon} % synchro prefix
\newcommand{\pincLabT}{\epsilon} % structural interaction
\newcommand{\pincLabL}{\mathfrak{l}} % Generic label
\newcommand{\pincLabM}{\mathfrak{m}}
\newcommand{\pincLabN}{\mathfrak{n}}
\newcommand{\pincLabnL}{\vlne{\mathfrak{l}}}

\newcommand{\pincE}{E} % Pi terms names
\newcommand{\pincF}{F}
\newcommand{\pincG}{G}
\newcommand{\pincH}{H}

\newcommand{\pinca}{a} % Pi channel names
\newcommand{\pincaRed}{\textcolor{Red}{a}}
\newcommand{\pincb}{b}
\newcommand{\pincbBlu}{\textcolor{Blue}{b}}
\newcommand{\pincc}{c}
\newcommand{\pincd}{d}
\newcommand{\pincna}{{\overline{a}}} % Pi channel co-names
\newcommand{\pincnaRed}{\textcolor{Red}{\pincna}}
\newcommand{\pincnb}{{\overline{b}}}
\newcommand{\pincnbBlu}{\textcolor{Blue}{\pincnb}}
\newcommand{\pincnc}{{\overline{c}}}

\newcommand{\pincalpha}{\alpha} % Pi meta names prefixes
\newcommand{\pincbeta}{\beta}
\newcommand{\pincgamma}{\gamma}
\newcommand{\pincAct}{s}

\newcommand{\pincFN}[1]{\operatorname{fn}(#1)} % Free names set
\newcommand{\cutfree}{cut-free\xspace}

\newcommand{\cutelimination}{cut-elimination\xspace}

\newcommand{\Deepinference}{Deep Inference\xspace}

\newcommand{\rightcontext}{right-context}
\newcommand{\Rightcontext}{Right-context}
\newcommand{\invertible}{invertible}
\newcommand{\Invertible}{Invertible}
\newcommand{\invertiblestructure}{\invertible\ structure}

\newcommand{\coinvertible}{co-invertible}

\newcommand{\coinvertiblestructure}{\coinvertible\ structure}

\newcommand{\canonical}{canonical}
\newcommand{\Canonical}{Canonical}
\newcommand{\canonicalstructure}{\canonical\ structure}
\newcommand{\Canonicalstructure}{\Canonical\ structure}
\newcommand{\trivial}{trivial}
\newcommand{\Trivial}{Trivial}
\newcommand{\trivialderivation}{\trivial\ derivation}
\newcommand{\Trivialderivation}{\Trivial\ derivation}
\newcommand{\nontrivial}{non-\trivial\xspace}

\newcommand{\nontrivialderivation}{\nontrivial\ derivation}

\newcommand{\standard}{standard\xspace}
\newcommand{\Standard}{Standard\xspace}
\newcommand{\standardderivation}{\standard\ derivation}
\newcommand{\Standardderivation}{\Standard\ derivation}

\newcommand{\simple}{simple\xspace}
\newcommand{\Simple}{Simple\xspace}
\newcommand{\simplestructure}{\simple\ structure}
\newcommand{\Simplestructure}{\Simple\ structure}
\newcommand{\simpleprocess}{simple process}
\newcommand{\Simpleprocess}{Simple process}
\newcommand{\processstructure}{process structure}
\newcommand{\Processstructure}{Process structure}
\newcommand{\environmentstructure}{environment structure}
\newcommand{\Environmentstructure}{Environment structure}

\newcommand{\lts}{labeled transition system}
\newcommand{\Lts}{Labeled transition system}

\newcommand{\Set}[1]{ \{ #1 \}}
\newcommand{\Size}[1]{|#1|}
\newcommand{\smalltitle}[1]{\small #1}
\newcommand{\ST}{such that}

\newcommand{\ie}{i.e.\xspace}
\newcommand{\wrt}{w.r.t.\xspace}
\newcommand{\dfn}[1]{\emph{#1}}
\newcommand{\grammareq}{::=\ }
\newcommand{\DI}{DI\xspace}

\newcommand{\BV}{\mathsf{BV}\xspace}
\newcommand{\SBV}{\mathsf{SBV}\xspace}
\newcommand{\BVsub}{\mathsf{B}\xspace}
\newcommand{\SBVsub}{\mathsf{B}\xspace}
\newcommand{\BVT}{\BV\mathsf{Q}\xspace}
\newcommand{\BVTL}{{\BV\mathsf{Q}\llcorner}\xspace}
\newcommand{\BVTMin}{{\BV\mathsf{Q}^{-}}\xspace}

\newcommand{\CCSRM}{\mathsf{CCS_{spr}}\xspace} % Milner CCS with restriction
\newcommand{\CCSR}{\mathsf{CCS_{spq}}\xspace} % CCS with Sdq-related restriction
\newcommand{\SBVT}{\SBV\mathsf{Q}\xspace}

\newcommand{\CCS}{\mathsf{CCS}\xspace}

\definecolor{Blue}{rgb}{0,0,1}

\definecolor{Red}{rgb}{1,0.2,0}
\newtheorem{lthm}{Theorem}[section]{\bfseries}{\itshape}
\newtheorem{lprop}[lthm]{Proposition}{\bfseries}{\itshape}
\newtheorem{llem}[lthm]{Lemma}{\bfseries}{\itshape}
\newtheorem{lcor}[lthm]{Corollary}{\bfseries}{\itshape}
\newtheorem{lcon}[lthm]{Conjecture}{\bfseries}{\itshape}
\newtheorem{lfac}[lthm]{Fact}{\bfseries}{\itshape}
\newtheorem{ldef}[lthm]{Definition}{\bfseries}{}
\newtheorem{lrem}[lthm]{Remark}{\bfseries}{}
\newtheorem{lexa}[lthm]{Example}{\bfseries}{}

\catcode`\@=11
\def\@thmcountersep{.}
\def\@thmcounterend{\;}
\catcode`@=12

\vllineartrue
\renewcommand{\vlscn}[1]{
\!\ifvirginialakesmallbrackets\{\else\left\{\fi #1
  \ifvirginialakesmallbrackets\}\else\right\}\fi   } % context
\newcommand{\vlone}{\circ} % multiplicative unit
\newcommand{\vloneRed}{\textcolor{red}{\circ}} % multiplicative unit
\newcommand{\vlqu}[2]{
  \ifvirginialakesmallbrackets\lceil\else\left\lceil\fi #2
  \ifvirginialakesmallbrackets\rfloor\else\right\rfloor\fi_{#1}
}
\newcommand{\vlfo}[2]{\vlqu{#1}{#2}}
\newcommand{\vlex}[2]{\vlqu{#1}{#2}}

\newcommand{\vlholer}[1]{#1^{\llcorner}} % right context
\newcommand{\OpNameSeq}{\textnormal{\textsf{Seq}}\xspace}
\newcommand{\OpNamePar}{\textnormal{\textsf{Par}}\xspace}
\newcommand{\OpNameCop}{\textnormal{\textsf{CoPar}}\xspace}
\newcommand{\OpNameRen}{\textnormal{\textsf{Sdq}}\xspace}
\newcommand{\OpNameTen}{\textnormal{\textsf{Tensor}}\xspace}
\newcommand{\OpNameNot}{\textnormal{\textsf{Not}}\xspace}
\renewcommand{\vlne}[1]{\overline{#1}}

\newcommand{\strFN}[1]{\operatorname{fn}(#1)} % free names in structures
\newcommand{\strBN}[1]{\operatorname{bn}(#1)} % bound names in structures
\newcommand{\strP}{P}
\newcommand{\strR}{R}
\newcommand{\strS}{S}
\newcommand{\strSb}{\breve{S}}
\newcommand{\strSc}{\tilde{S}}

\newcommand{\strT}{T}
\newcommand{\strU}{U}
\newcommand{\strV}{V}

\newcommand{\atma}{a} % atoms
\newcommand{\atmaColor}[1]{\textcolor{#1}{\atma}}
\newcommand{\atmaRed}{\atmaColor{Red}}

\newcommand{\atmb}{b}
\newcommand{\atmbColor}[1]{\textcolor{#1}{\atmb}}
\newcommand{\atmbRed}{\atmbColor{Red}}
\newcommand{\atmbBlu}{\atmbColor{Blue}}

\newcommand{\atmc}{c}
\newcommand{\atmcColor}[1]{\textcolor{#1}{\atmc}}
\newcommand{\atmcRed}{\atmcColor{Red}}

\newcommand{\atmd}{d}
\newcommand{\natma}{\vlne\atma} % negated atoms
\newcommand{\natmaColor}[1]{\textcolor{#1}{\natma}}
\newcommand{\natmaRed}{\natmaColor{Red}}
\newcommand{\natmaBlu}{\natmaColor{Blue}}
\newcommand{\natmb}{\vlne\atmb}
\newcommand{\natmbColor}[1]{\textcolor{#1}{\natmb}}

\newcommand{\natmbBlu}{\natmbColor{Blue}}
\newcommand{\natmc}{\vlne\atmc}
\newcommand{\natmcColor}[1]{\textcolor{#1}{\natmc}}

\newcommand{\natmcBlu}{\natmcColor{Blue}}

\newcommand{\atmLabL}{\mathfrak{l}} %% generic names
\newcommand{\atmLabM}{\mathfrak{m}}
\newcommand{\atmLabN}{\mathfrak{n}}
\newcommand{\atmalpha}{\alpha} % Pi meta names prefixes
\newcommand{\atmbeta}{\beta}

\newcommand{\subst}[2]{\{^{#1}\!/\!_{#2}\}} % substitutes #1 for #2
\newcommand{\bvtrdrule}{\mathsf{u}\downarrow}
\newcommand{\bvtrurule}{\mathsf{u}\uparrow}

\newcommand{\bvtrhorule}{\rho}
\newcommand{\bvtatrdrule}{\mathsf{at}\downarrow\llcorner}

\newcommand{\bvtintdrule}{\mathsf{i}\downarrow}

\newcommand{\bvtseqdrule}{\mathsf{q}\downarrow}
\newcommand{\bvtsequrule}{\mathsf{q}\uparrow}
\newcommand{\bvtatidrule}{\mathsf{ai}\downarrow}
\newcommand{\bvtatiurule}{\mathsf{ai}\uparrow}
\newcommand{\bvtswirule }{\mathsf{s}}

\newcommand{\bvtrdrulein}{\mbox{$\mathsf{u}\!\!\downarrow$}}
\newcommand{\bvtrurulein}{\mbox{$\mathsf{u}\!\!\uparrow$}}

\newcommand{\bvtrhorulein}{\rho}
\newcommand{\bvtatrdrulein}{\mbox{$\mathsf{at}\!\!\downarrow\!\!\llcorner$}}

\newcommand{\bvtintdrulein}{\mbox{$\mathsf{i}\!\!\downarrow$}}

\newcommand{\bvtseqdrulein}{\mbox{$\mathsf{q}\!\!\downarrow$}}

\newcommand{\bvtsequrulein}{\mbox{$\mathsf{q}\!\!\uparrow$}}
\newcommand{\bvtatidrulein}{\mbox{$\mathsf{ai}\!\!\downarrow$}}
\newcommand{\bvtatiurulein}{\mbox{$\mathsf{ai}\!\!\uparrow$}}
\newcommand{\bvtswirulein }{\mathsf{s}}

\newcommand{\bvtDder}{\mathcal{D}} % derivation names
\newcommand{\bvtEder}{\mathcal{E}}
\newcommand{\bvtPder}{\mathcal{P}} % proof names
\newcommand{\bvtQder}{\mathcal{Q}}
\newcommand{\bvtInfer}[2]{#1: #2}
\newcommand{\bvtJudGen}[2]{\vdash_{#1 % system of reference
}^{#2 % subset of rules
}}

%% file: introduction.tex
\section{Introduction}
\label{section:Introduction}
This is a work in structural proof-theory which builds on 
\cite{Brus:02:A-Purely:wd,Roversi:2010-LLCexDI,Roversi:TLCA11,Roversi:unpub2012-I}.
Broadly speaking we aim at using structural proof theory to study primitives of paradigmatic 
programming languages, and to give evidence that some are the natural ones, while others, 
which we might be used to think of as ``given once for all'', can, in fact, be refined or 
generalized. In our case this means to keep developing the programme in 
\cite{Brus:02:A-Purely:wd} whose goal is establishing a correspondence between proof-search of 
a logical system, and computations in a process algebra.
From \cite{Brus:02:A-Purely:wd}, we already know that both
(i) sequential composition of Milner $\CCS$ \cite{Miln:89:Communic:qo} gets modeled by the non commutative logical operator \OpNameSeq of $\BV$ \cite{Gugl:06:A-System:kl}, which is the paradigmatic calculus of structures in \Deepinference, and
(ii) parallel composition of Milner $\CCS$ gets modeled by the commutative logical operator \OpNamePar of $\BV$ so that communication becomes logical annihilation.
This is done under a logic-programming analogy. It says that the terms of a calculus 
$\mathcal{C}$ --- which is a fragment of Milner $ \CCS $ in the case of \cite{Brus:02:A-Purely:wd} --- correspond to formulas of a logical system $ \mathcal{L} $ --- which is $ \BV $ in the case of \cite{Brus:02:A-Purely:wd} ---, and that computations inside 
$\mathcal{C}$ recast to searching \cutfree proofs in $\mathcal{L}$, as summarized in~\eqref{equation:introduction-non-curry-howard-correspondence} here below.
%%%%%%%%%%%
\par\vspace{\baselineskip}\noindent
{\small
  \fbox{
    \begin{minipage}{.974\linewidth}
      \begin{equation}
        \label{equation:introduction-non-curry-howard-correspondence}
		\input{./introduction-non-curry-howard-correspondence}
      \end{equation}
    \end{minipage}
  }%fbox
}%\small
\vspace{\baselineskip}\par\noindent
%%%%%%%%%%%%%%%
\paragraph{Contributions.}
We show that in~\eqref{equation:introduction-non-curry-howard-correspondence} we can 
take $\BVT$ \cite{Roversi:2010-LLCexDI,Roversi:TLCA11,Roversi:unpub2012-I} for $\mathcal L$, 
and $ \CCSR $ for $\mathcal C$. The system $\BVT$ extends $\BV$ with a self-dual quantifier, 
while $\CCSR$ is introduced by this work (Section~\ref{section:Communicating, and concurrent processes with restriction}). The distinguishing aspect of $\CCSR$ is its operational 
semantics which subsumes the one of the fragment of Milner $\CCS$ that contains sequential, 
parallel, and restriction operators, and which we identify as $\CCSRM$.
Specifically, the self-dual quantifier of $\CCSR$ allows to relax the operational semantics of the restriction operator in $\CCSRM$ without getting to an inconsistent calculus of 
processes. This is a direct consequence of (the analogous of) the a \cutelimination property for $\BVT$ \cite{Roversi:2010-LLCexDI,Roversi:TLCA11,Roversi:unpub2012-I}.
\par
The main step that allows to take $\BVT$ for $\mathcal L$, and $\CCSR$ for $\mathcal C$ is proving Soundness of $ \BVT $ with respect to $ \CCSR $ (Section~\ref{section:Soundness of BVTCC}). The following example helps explaining what Soundness amounts to.
Let us suppose we want to observe what the following judgment describes:
%%%%%%%%%%%
\par\vspace{\baselineskip}\noindent
{\small
  \fbox{
    \begin{minipage}{.974\linewidth}
    \vspace{-.1cm}
      \begin{equation}
      \label{equation:intro-example-00-term-to-compute}
      \input{./intro-example-00-term-to-compute}
      \end{equation}
    \end{minipage}
  }%fbox
}%\small
\vspace{\baselineskip}\par\noindent
%%%%%%%%%%%%
The process $\pincSec{\pinca}{\pincSec{\pincb}{\pincE}} $ can perform actions $\pinca$, and $\pincb$, in this order, before entering $ \pincE $.
The other process can perform $\pincna$ before entering $\pincF$. In particular,
$ \pincSec{\pinca}{\pincSec{\pincb}{\pincE}} $, and $ \pincSec{\pincna}{\pincF} $ internally communicate when simultaneously firing $\pinca$, and $\pincna$. In any case, firing on 
$\pinca$, or $\pincna$, would remain private because of the outermost restriction $\pincNu{\pinca}{\,\cdot\,} $ which hides both $\pinca$, and $\pincna$ to the environment\footnote{ 
We write something related to Milner $\CCS$. Indeed,  hiding both $\pinca$, and $\pincna$ in Milner $\CCS$ is $\pincNu{\Set{\pinca,\pincna}}{\,\cdot\,}$.}.
The action $\pincb$ is always observable because $\pincb$ differs from $\pinca$.
Of course, we might describe one of the possible dynamic evolutions of~\eqref{equation:intro-example-00-term-to-compute} thanks to a suitable \lts\ able to develop a derivation like~\eqref{equation:intro-example-00}:
%%%%%%%%%%%
\par\vspace{\baselineskip}\noindent
\fbox{
 \begin{minipage}{.974\linewidth}
    {\small
      \vspace{-.3cm}
      \begin{equation}
        \label{equation:intro-example-00}
        \input{./intro-example-00}
      \end{equation}
    }%\small
    \end{minipage}
  }%fbox
%%%%%%%%%%%%%%%
\vspace{\baselineskip}\par\noindent
Soundness says that instead of rewriting
$\pincSec{\pinca}{\pincSec{\pincb}{\pincE}}$ to $\pincSec{\pincna}{\pincF}$, as in~\eqref{equation:intro-example-00}, we can 
(i) compile the whole judgment
$\pincLTSJud{\pincNu{\pinca}
		            {(\pincPar{(\pincSec{\pinca}
		                                {\pincSec{\pincb}
		                                         {\pincE}
		                                })
		                      }
		                      {\pincSec{\pincna}
		                               {\pincF}
		                      })
		            }
            }
			{\pincNu{\pinca}{(\pincPar{\pincE}{\pincF})}}
			{\pincb}$
to a structure, say $\strR$, of $\BVT$, and
(ii) search for a \cutfree proof, say $\bvtPder$ of $\strR$, and 
(iii) if $\bvtPder$ exists, then Soundness assures that~\eqref{equation:intro-example-00-term-to-compute} holds.
So, in general, Soundness recasts the reachability problem ``Is it true that $\pincLTSJudShort{\pincE}{\pincF}{\pincalpha}$'' to a problem of proof search. 
Noticeably, the Soundness we prove poses weaker constraints on the form of $\pincF$ than those ones we find in Soundness of~\cite{Brus:02:A-Purely:wd}. Specifically, 
only the silent process $\pincZer$ can be the target of the reachability problem 
in~\cite{Brus:02:A-Purely:wd}. Here, $\pincF$ can belong to the set of \emph{\simpleprocess es} which contains $\pincZer$. Intuitively, every \simpleprocess\ different from $\pincZer$ is normal with respect to internal communication, but is alive if we consider the external ones. Finally, from a technical standing point, our proof of Soundness in neatly decomposed in steps that makes it reusable for further extensions of both $\BVT$, and $\CCSRM$.
%%%%
\paragraph{Road map.}
Section~\ref{section:Systems SBVT and BVT} recalls $\BVT$ and its symmetric version $ \SBVT $ mainly from \cite{Roversi:unpub2012-I}.
Section~\ref{section:Standardization inside BVT} is about two proof-theoretical properties of $\BVT$ which were not proved in 
\cite{Roversi:2010-LLCexDI,Roversi:TLCA11,Roversi:unpub2012-I} but which Soundness relies on. The first one says that every \OpNameTen-free derivations of $\BVT$ has at least corresponding  \emph{\standard} one. The second one supplies sufficient conditions for a structure of $\BVT$ to be invertible, somewhat internalizing derivability of $\BVT$.
Section~\ref{section:Communication core of BVT} has the pedagogical aim of showing, with many examples, why the derivations of $ \BVT $ embody a computational meaning.
Section~\ref{section:Communicating, and concurrent processes with restriction} introduces $\CCSR$, namely the process calculus that $\BVT$ embodies.
Section~\ref{section:How computing  in CCSR by means of BVTCC} first formalizes the connections between $\BVT$, and $\CCSR$. Then it shows how computations inside the \lts\ of $\CCSR$ recast to proof-search inside $\BVT$, justifying the need to prove Soundness.
Section~\ref{section:Soundness of BVTCC} proves Soundness, starting with a pedagogical overview of what proving it means. Section~\ref{section:Final discussion, and future work} points to future work, mainly focused on $\CCSR$.

%% file: introduction-non-curry-howard-correspondence.tex
\begin{tabular}{c|c}
\textbf{Paradigmatic calculus} $\mathcal{C}$ & 
     \textbf{Logical system} $\mathcal{L} $
 \\\hline\hline
term             & formula\\\hline
step of computation  & logical rule\\\hline
computation  & searching a \cutfree proof
\end{tabular}

%% file: intro-example-00-term-to-compute.tex
\pincLTSJud{\pincNu{\pinca}
                {(\pincPar{(\pincSec{\pinca}
                                    {\pincSec{\pincb}
                                             {\pincE}
                                    })
                          }
                             {(\pincSec{\pincna}
                                       {\pincF})
                             })
                  }
           }
     {\pincNu{\pinca}{(\pincPar{\pincE}{\pincF})}
  }
     {\pincb}

%% file: intro-example-00.tex
\vlderivation{
 \vliin{ %\pinctran
       }{}
	   {
        \pincLTSJud{\pincNu{\pinca}
	                       {(\pincPar{(\pincSec{\pinca}
	                                           {\pincSec{\pincb}
	                                                    {\pincE}
	                                           })
	                                 }
                                     {(\pincSec{\pincna}
                                              {\pincF})
                                     })
                          }
                   }
		           {\pincNu{\pinca}{(\pincPar{\pincE}{\pincF})}
			       }
		       {\pincPreT;\pincb\,\pincCong\,\pincb}
       }
       {%%A
        \vlin{ %\pincpe
             }
             {\!\pincLabT\!\not\equiv\!\pinca}
	          {
               \pincLTSJud{\pincNu{\pinca}
	                              {(\pincPar{(\pincSec{\pinca}
					                                  {\pincSec{\pincb}{\pincE}})
				                            }
				                            {(\pincSec{\pincna}{\pincF})})
	                              }
                         }
		                  {\pincNu{\pinca}{(\pincPar{(\pincSec{\pincb}{\pincE})}
		                                     		{\pincF})}
		                  }
		                  {\pincPreT}
             } 
             {
	            \vliin{ %\pinccom
	                  }{}
    		          {
                       \pincLTSJud{
    		                       \pincPar{(\pincSec{\pinca}
		                                             {\pincSec{\pincb}{\pincE}})}
				                           {(\pincSec{\pincna}{\pincF})}
			                      }
			                      {\pincPar{(\pincSec{\pincb}{\pincE})}{\pincF}
			                      }
		                          {\pincPreT}
  		              } 
                      {%%%A.1-------
		                 \vlin{%\pincact
		                      }{}
		                      {\pincLTSJud{\pincSec{\pinca}
		                                           {\pincSec{\pincb}{\pincE}}
			                              }
			                              {\pincSec{\pincb}{\pincE}}
		                                  {\pinca}
		                      } {
		                 \vlhy{}}
		              }
		              {%%%A.2-------
		                 \vlin{%\pincact
		                      }{}
		                      {\pincLTSJud{\pincSec{\pincna}{\pincF}}{\pincF}{\pincna}
		                      }{
		                 \vlhy{}}
		              }
              }
   }%%A-stop
   {%%%B
    \vlin{%\pincpe
         }{\!\pincb\!\not\equiv\!\pinca}
   		 {\pincLTSJud{\pincNu{\pinca}{(\pincPar{(\pincSec{\pincb}{\pincE})}{\pincF})}
		             }
   			         {
   			          \pincNu{\pinca}{(\pincPar{\pincE}{\pincF})}
   			         }
   			         {\pincb}
   	      }   {
    \vlin{%\pinccntxp
         }{}
   		 {\pincLTSJud{
   		              \pincPar{(\pincSec{\pincb}{\pincE})}{\pincF}
   		             }
   			         {\pincPar{\pincE}{\pincF}
   			         }
   			         {\pincb}
   	      }   {
    \vlin{ %\pincact
         }{}
   		 {\pincLTSJud{\pincSec{\pincb}{\pincE}}
   			         {\pincE}
   			         {\pincb}
   	      }   {
    \vlhy{}   }}}
   }%%%B-stop
}%%%vlderivation---stop

%% file: SBV2.tex
\section{Recalling the systems $\SBVT$ and $\BVT$}
\label{section:Systems SBVT and BVT}
We briefly recall $ \SBVT $, and $\BVT $ from \cite{Roversi:unpub2012-I}.
\paragraph{Structures.}
Let $\atma, \atmb, \atmc, \ldots$ denote the elements of a countable set of
\dfn{positive propositional variables}. Let $\natma, \natmb, \natmc, \ldots$
denote the elements of a countable set of \dfn{negative propositional variables}.
The set of \dfn{names}, which we range over by $\atmLabL, \atmLabM$, and
$\atmLabN$, contains both positive, and negative propositional variables,
and nothing else. Let $\vlone$ be a constant, different from any name, which we call \dfn{unit}. The set of \dfn{atoms} contains both names and the unit, while the set of
\dfn{structures} identifies formulas of $\SBV$. Structures belong to
the language of the grammar in~\eqref{fig:BVT-structures}.
%%%%%%%%%%%%%%%%
\par\vspace{\baselineskip}\noindent
{\small
  \fbox{
   \begin{minipage}{.974\linewidth}
      \begin{equation}
       \label{fig:BVT-structures}
       \input{./BV2-structures}
      \end{equation}
%       \vspace{.2cm}
    \end{minipage}
  }%fbox
}%\small
\par\vspace{\baselineskip}\noindent
%%%%%%%%%%%%%%%%
We use 
$
\strR, \strT, \strU, \strV$ to range over structures, in which 
$\vlne{\strR}$ is a \OpNameNot, 
$\vlrobrl\strR\vlte\strT\vlrobrr$ is a \OpNameCop, 
$\vlsbr<\strR;\strT>$ is a \OpNameSeq, 
$\vlsqbrl\strR\vlpa\strT\vlsqbrr$ is a \OpNamePar, and 
$\vlfo{\atma}{\strR}$ is a self-dual quantifier \OpNameRen, which
comes with the proviso that $\atma$ must be a positive atom. Namely, 
$\vlfo{\natma}{\strR}$ is not in the syntax. \OpNameRen\
induces obvious notions of \dfn{free}, and \dfn{bound names} \cite{Roversi:unpub2012-I}.
%%%%%%%%%%%%%
\paragraph{Size of the structures.}
The \dfn{size} $\Size{\strR}$ of $\strR$ is the number of occurrences
of atoms in $\strR$ plus the number of occurrences of \OpNameRen\ that effectively
bind an atom. For example, $\vlstore{\vlsbr[\atma;\natma]}\Size{\vlread} =
\vlstore{\vlfo{\atmb}{\vlsbr[\atma;\natma]}}\Size{\vlread}=2$, while 
$\vlstore{\vlfo{\atma}{\vlsbr[\atma;\natma]}} \Size{\vlread}=3$.
%%%%
\paragraph{(Structure) Contexts.}
We denote them by $\strS\vlhole$. A context is a structure with a single hole
$\vlhole$ in it. If $\strS\vlscn{\strR}$, then $\strR$ is a \dfn{substructure} of
$\strS$. We shall tend to shorten
$\vlstore{\vlsbr[\strR;\strU]}\strS\vlscn{\vlread}$ as $\strS{\vlsbr[\strR;\strU]}$
when ${\vlsbr[\strR;\strU]}$ fills the hole $\vlhole$ of $\strS\vlhole$ exactly.
%%%%%%%%%%%%%
\paragraph{Congruence $\approx$ on structures.}
Structures are partitioned by the smallest congruence $\approx$ we obtain as
reflexive, symmetric, transitive and contextual closure of the relation $\sim$ whose
defining clauses are \eqref{align:negation-atom}, through \eqref{align:alpha-symm} here
below.
%%%%%%%%%%%%%
\par\vspace{\baselineskip}\noindent
  \fbox{
    \begin{minipage}{.974\linewidth}
%      \vspace{-.3cm}
     {\small
       \input{./BV2-structure-equivalences}
     }%\small
     \vspace{-.5cm}
    \end{minipage}
  }%fbox
\par\vspace{\baselineskip}\noindent
%%%%%%%%%%%%%
\emph{Contextual closure}
means that ${\strS{\vlscn{\strR}}} \approx {\strS{\vlscn{\strT}}}$ whenever $\strR
\approx \strT$.
Thanks to \eqref{align:alpha-symm}, we abbreviate 
$ \vlfo{\atma_n}{\vldots\vlfo{\atma_1}{\strR}\vldots} $ as  $ \vlfo{\vec{\atma}}{\strR}$, where we may also interpret $ \vec{\atma} $ as one of the permutations of $ \atma_1, \ldots, \atma_n $.
%%%%%
\paragraph{\Canonicalstructure s.}
\label{paragraph:Structures in canonical form}
We inspire to the normal forms of \cite{Gugl:06:A-System:kl} to define structures in
\dfn{\canonical form} inside $\SBVT$. \Canonicalstructure s will be used to define \environmentstructure s (Section~\ref{section:How computing  in CCSR by means of BVTCC}, page~\pageref{paragraph:Environmentstructure s}.)
A structure $\strR$ is \emph{\canonical} when either it is the unit $\vlone$, or the
following four conditions hold:
(i) the only negated structures appearing in $\strR$ are negative propositional variables,
(ii) no unit $\vlone$ appears in $\strR$, but at least one name occurs in
it,
(iii) the nesting of occurrences of \OpNamePar, \OpNameTen, \OpNameSeq, and \OpNameRen  build a right-recursive syntax tree of $\strR$, and
(iv) no occurrences of \OpNameRen can be eliminated from $\strR$, while maintaining
the equivalence.
%%%%%
\begin{example}[\textbf{\textit{\Canonicalstructure s}}]
\label{example:Canonical structures}
The structure $\vlstore{\vlsbr[(\natma;\natmb);\vlfo{\atmc}{\natmc}]}\vlread$ is not \canonical,
but it is equivalent to the \canonical\
one $\vlstore{\vlsbr[\natma;(\natmb;\vlfo{\atmc}{\natmc})]}\vlread$ whose syntax tree is right-recursive. Other non \canonicalstructure s are
$\vlstore{\vlsbr[\atma;(\vlone;\atmb)]}\vlne{\vlread}$, and
$\vlsbr(\vlne{\vlread};<\vlone;\natmb>)$, and 
$\vlstore{\vlsbr[\natma;(\natmb;\vlfo{\atmd}{\natmc})]}\vlread$. The first two are equivalent to $\vlsbr(\natma;\natmb)$ which, instead, is \canonical. Finally, also
$\vlstore{\vlsbr[\atma;\vlone]}\vlread$ is not \canonical, equivalent to the
\canonical\ one $\atma$.
\end{example}
%%%%%%
\begin{fact}[\textbf{\textit{Normalization to \canonicalstructure s}}]
\label{fact:Structures normalize to canonical forms}
Given a structure $\strR$:
(i) negations can move inward to atoms, and, possibly, disappear,
thanks to \eqref{align:negation-atom}, \ldots, \eqref{align:negation-fo},
(ii) units can be removed thanks to \eqref{align:unit-co}, \ldots,
\eqref{align:unit-pa}, and
(iii) brackets can move rightward by \eqref{align:assoc-co}, \ldots,
\eqref{align:assoc-pa}.
\par
So, for every $\strR$ we can take the equivalent \canonicalstructure\ which
is either $\vlone$, or different from $\vlone$.
\end{fact}
%%%%%%%%%%%%%
\paragraph{The system $\SBVT$.}
It contains the set of inference rules in \eqref{fig:SBVT} here below.
Every rule has form $\vldownsmash{\vlinf{\bvtrhorule}{}{\strR}{\strT}}$, \dfn{name}
$\bvtrhorule$, \dfn{premise} $\strT$, and \dfn{conclusion} $\strR$.
\par\vspace{\baselineskip}\noindent
  \fbox{
    \begin{minipage}{.974\linewidth}
%      \vspace{-.3cm}
     {\small
      \begin{equation}
        \label{fig:SBVT}
        \input{./SBV2-system}
      \end{equation}
     }%\small
%      \vspace{-.5cm}
    \end{minipage}
  }%fbox
%%%%%%%%%%%%
\paragraph{Derivations vs. proofs.}
A \dfn{derivation} in $\SBVT$ is either a structure or an instance of the above
rules or a sequence of two derivations. Both $\bvtDder$, and $\bvtEder$ will range
over derivations. The topmost structure in a derivation is
its \dfn{premise}. The bottommost is its \dfn{conclusion}. The \dfn{length}
$\Size{\bvtDder}$ of a derivation $\bvtDder$ is the number of rule instances in
$\bvtDder$. A derivation $\bvtDder$ of a structure $\strR$ in $\SBVT$ from a
structure $\strT$ in $\SBVT$, only using a subset $\SBVsub\subseteq\SBVT$ is
$\vlderivation                  {
\vlde{\bvtDder}{\SBVsub}{\strR}{
\vlhy                   {\strT}}}
$.
The equivalent \emph{space-saving} form is
$\bvtInfer{\bvtDder}{\strT \bvtJudGen{\SBVsub}{}\strR}$.
The derivation
$\vlupsmash{
\vlderivation                  {
\vlde{\bvtDder}{\SBVsub}{\strR}{
\vlhy                   {\strT}}}}$
is a \dfn{proof} whenever $\strT\approx \vlone$. We denote it as
$\vlupsmash{
\vlderivation                  {
\vlde{\bvtPder}{\SBVsub}{\strR}{
\vlhy                   {\vlone}}}}$, or
$\vlupsmash{\vlproof{\bvtPder}{\SBVsub}{\strR}}$,
or $\bvtInfer{\bvtPder}{\ \bvtJudGen{\SBVsub}{}\strR}$. Both $\bvtPder$, and
$\bvtQder$ will range over proofs.
In general, we shall drop $\SBVsub$ when clear from the context.
In a derivation, we write
$
% \vlupsmash{
  \vliqf{\bvtrhorule_1,\ldots,\bvtrhorule_m,n_1,\ldots,n_p}{}{\strR}{\strT}
%  }%\vlupsmash
$,
whenever we use the rules $\bvtrhorule_1,\ldots,\bvtrhorule_m$ to derive $\strR$
from $\strT$ with the help of $n_1,\ldots,n_p$ instances of
\eqref{align:negation-atom}, \ldots, \eqref{align:symm-co}.
To avoid cluttering derivations, whenever possible, we shall tend to omit the use of
negation axioms \eqref{align:negation-atom}, \ldots, \eqref{align:negation-fo},
associativity axioms \eqref{align:assoc-co}, \eqref{align:assoc-se},
\eqref{align:assoc-pa}, and symmetry aximos \eqref{align:symm-pa},
\eqref{align:symm-co}. This means we avoid writing all
brackets, as in $\vlsbr[\strR;[\strT;\strU]]$, in favor of
$\vlsbr[\strR;\strT;\strU]$, for example.
Finally if, for example, $q>1$ instances of some axiom $(n)$ of
\eqref{align:negation-atom}, \ldots, \eqref{align:alpha-symm} occurs among $n_1,\ldots,n_p$, then we write $(n)^q$.
%%%%%%
\paragraph{\dfn{Up} and \dfn{down} fragments of $\SBVT$.}
The set $\Set{\bvtatidrulein, \bvtswirulein, \bvtseqdrulein,
\bvtrdrulein}$ is the \dfn{down fragment} $\BVT$ of $\SBVT$.
The \dfn{up fragment} is $\Set{\bvtatiurulein,\bvtswirulein, \bvtsequrulein,
\bvtrurulein}$. So $\bvtswirulein$ belongs to both.
%%%%%%%%%%%%%
\begin{corollary}[\cite{Roversi:TLCA11,Roversi:unpub2012-I}]
\label{theorem:Admissibility of the up fragment}
The up-fragment 
$\Set{\bvtatiurulein, \bvtsequrulein, \bvtrurulein}$ of $\SBVT$ is admissible for $\BVT$. 
This means that we can transform any proof 
$\bvtInfer{\bvtPder}{\ \bvtJudGen{\SBVT}{}\strR}$ into a proof
$\bvtInfer{\bvtQder}{\ \bvtJudGen{\BVT}{}\strR}$ free of every occurrence of rules
that belong to the up-fragment of $\SBVT$.
\end{corollary}
%%%%%%%%
\begin{remark}
Thanks to Corollary~\ref{theorem:Admissibility of the up fragment}, we shall always focus on the up-fragment $\BVT$ of $\SBVT$.
\end{remark}

%% file: BV2-structures.tex
\strR  \grammareq \vlone
              \ \ \mid \ \ \atmLabL
              \ \ \mid \ \ \vlne{\strR}
              \ \ \mid \ \ \vlrobrl\strR\vlte\strR\vlrobrr
              \ \ \mid \ \ \vlsbr<\strR;\strR>
              \ \ \mid \ \ \vlsqbrl\strR\vlpa\strR\vlsqbrr
              \ \ \mid \ \ \vlfo{\atma}{\strR}

%% file: BV2-structure-equivalences.tex
\begin{minipage}{.48\textwidth}
      \begin{center}
      %%%%%%
      \input{./BV2-structure-equivalences-negation}
      %%%%%%

      %%%%%%
      \input{./BV2-structure-equivalences-symmetry}
      %%%%%%
%
      %%%%%%
%       \input{./BV2-structure-equivalences-singleton}
      %%%%%%
      \end{center}
\end{minipage}
%%%% 2nd column
\begin{minipage}{.5\textwidth}
      \begin{center}
%         \vspace{2\baselineskip}
      %%%%%%
      \input{./BV2-structure-equivalences-associativity}
      %%%%%%

      %%%%%%
      \input{./BV2-structure-equivalences-unit}
      %%%%%

      %%%%%%
      \input{./BV2-structure-equivalences-alpha}
      %%%%%%
      \end{center}
\end{minipage}

%% file: BV2-structure-equivalences-negation.tex
\smalltitle{\textbf{Negation}}%
{%do not erase
 \vlstore{
   \label{align:negation-atom}
   \vlne{\vlone} &
   \sim & \vlone
   \\
   \label{align:negation-negation}
   \vlne{\vlne \strR} &
   \sim & \strR
   \\
   \label{align:negation-pa}
   \vlne{\vlsbr[\strR;\strT]} &
   \sim & {\vlsbr(\vlne{\strR};\vlne{\strT})}
   \\
   \label{align:negation-co}
   \vlne{\vlsbr(\strR;\strT)} &
   \sim & {\vlsbr[\vlne{\strR};\vlne{\strT}]}
   \\
   \label{align:negation-seq}
   \vlne{\vlsbr<\strR;\strT>} &
   \sim & {\vlsbr<\vlne{\strR};\vlne{\strT}>}
   \\
   \label{align:negation-fo}
   \vlne{\vlfo{\atma}{\strR}} &
   \sim & \vlex{\atma}{\vlne{\strR}}
}%\vlstore
{\setlength{\arraycolsep}{2pt}
\begin{eqnarray}
 \vlread
\end{eqnarray}}
 }%do not erase

%% file: BV2-structure-equivalences-symmetry.tex
\smalltitle{\textbf{Symmetry}}%
{%do not erase
  \vlstore{
         \label{align:symm-pa}
         \vlsbr[\strR;\strT] & \sim & \vlsbr[\strT;\strR]
         \\
         \label{align:symm-co}
         \vlsbr(\strR;\strT) & \sim & \vlsbr(\strT;\strR)
  }%\vlstore
{\setlength{\arraycolsep}{2pt}
\begin{eqnarray}
 \vlread
\end{eqnarray}}
}%do not erase

%% file: BV2-structure-equivalences-associativity.tex
\smalltitle{\textbf{Associativity}}%
{%do not erase
 \vlstore{
   \label{align:assoc-co}
   \vlsbr(\strR;(\strT;\strV))
   & \sim &
   \vlsbr((\strR;\strT);\strV)
   \\
   \label{align:assoc-se}
   \vlsbr<\strR;<\strT;\strV>>
   & \sim &
   \vlsbr<<\strR;\strT>;\strV>
   \\
   \label{align:assoc-pa}
   \vlsbr[\strR;[\strT;\strV]]
   & \sim &
   \vlsbr[[\strR;\strT];\strV]
 }%\vlstore
{\setlength{\arraycolsep}{2pt}
\begin{eqnarray}
 \vlread
\end{eqnarray}}
}%do not erase

%% file: BV2-structure-equivalences-unit.tex
\smalltitle{\textbf{Unit}}%
{%do not erase
  \vlstore{
   \label{align:unit-co}
          \vlsbr(\vlone;\strR)  & \sim  &\strR\\
   \label{align:unit-seq}
          \vlsbr<\vlone;\strR>  & \sim  &\vlsbr<\strR;\vlone>
                                           \sim  \strR\\
   \label{align:unit-pa}
          \vlsbr[\vlone;\strR]  & \sim  &\strR
  }%\vlstore
{\setlength{\arraycolsep}{2pt}
\begin{eqnarray}
 \vlread
\end{eqnarray}}
}%do not erase

%% file: BV2-structure-equivalences-alpha.tex
\smalltitle{$\alpha$-\textbf{rule}}%
{%do not erase
    \vlstore{
%         \label{align:alpha-intro-vlone}
%	  \vlfo{\atma}{\vlone} &
%	  \sim \vlone\\
	\label{align:alpha-intro}
	  \vlfo{\atma}{\strR} &
	  \sim & \strR
	  \quad\quad\ \textrm{ if } \atma\not\in\strFN{\strR}\\
        \label{align:alpha-varsub}
	  \vlfo{\atma}{\strR\subst{\atma}{\atmb}} &
	  \sim & \vlfo{\atmb}{\strR}
	  \quad \textrm{ if } \atma\not\in\strFN{\strR}
	  \\
        \label{align:alpha-symm}
	  \vlfo{\atma}{\vlfo{\atmb}{\strR}} &
	  \sim & \vlfo{\atmb}{\vlfo{\atma}{\strR}}
%         \label{align:alpha-extrusion}
%	  \vlfo{\atma}{\vlsbr[\strR;\strT]} &
%	  \sim \vlsbr[\vlfo{\atma}{\strR};\strT]
%	  & \textrm{ if } \atma\not\in\strFN{\strT}
         }%\vlstore
{\setlength{\arraycolsep}{2pt}
\begin{eqnarray}
 \vlread
\end{eqnarray}}
}%do not erase

%% file: SBV2-system.tex
{\small
\begin{tabular}{ccc}
%%%%%%%%%%
\input{./SBV2-system-SBV-rules}
%%%%%%%%%
\\ \\
%%%%%%%%%%
\input{./SBV2-system-SBVT-rules}
%%%%%%%%%
\end{tabular}
}

%% file: SBV2-system-SBV-rules.tex
{$\vlinf{\bvtatidrule}{}{\vlsbr[\atma;\natma]}
                        {\vlone}$}
&
&
\qquad
{$\vlinf{\bvtatiurule}{}{\vlone}
                       {\vlsbr(\atma;\natma)}$}
\\
\\
{$\vlinf{\bvtseqdrule}{}
        {\vlsbr[<\strR;\strT>;<\strU;\strV>]}
        {\vlsbr<[\strR;\strU];[\strT;\strV]>}$}
&\qquad
{$\vlinf{\bvtswirule}{}
        {\vlsbr[(\strR;\strU);\strT]}
        {\vlsbr([\strR;\strT];\strU)}$}
&\qquad
{$\vlinf{\bvtsequrule}{}
        {\vlsbr<(\strR;\strU);(\strT;\strV)>}
        {\vlsbr(<\strR;\strT>;<\strU;\strV>)}$}

%% file: SBV2-system-SBVT-rules.tex
{$\vlinf{\bvtrdrule}{}
       {\vlsbr[{\vlfo{\atma}{\strR}};{\vlex{\atma}{\strU}}]}
       {\vlfo{\atma}{\vlsbr[\strR;\strU]}}$}
&\qquad&
\qquad
{$\vlinf{\bvtrurule}{}
       {\vlex{\atma}{\vlsbr(\strR;\strU)}}
       {\vlsbr({\vlex{\atma}{\strR}};{\vlfo{\atma}{\strU}})}$}

%% file: SBV2-standardization-common.tex
\section{Standardization inside a fragment of $\BVT$}
\label{section:Standardization inside BVT}
Taken a derivation $ \bvtDder $ of $ \BVT $, standardization reorganizes $\bvtDder$ into another derivation $\bvtEder $ with the same premise, and conclusion, as $ \bvtDder $. The order of application of the instances of $ \bvtatidrulein $ in $ \bvtEder $ satisfies a specific, given constraint which some examples illustrate.
Standardization in $ \BVT $ is one of the properties we need to recast reachability problems in a suitable calculus of communicating, and concurrent processes, to proof-search inside (a fragment) of $ \BVT $.
%%%%%%
\begin{example}[\textbf{\textit{\Standardderivation s of $\BVT$}}]
\label{example:Standard proofs of BVT}
Both~\eqref{equation:tracing-sequential-interactions-01},
and~\eqref{equation:tracing-sequential-interactions-00} here below are
\standardderivation s of the same conclusion
$ \vlsbr[<\pincaRed;\strR>;<\pincnbBlu;\strT>;<\pincnaRed;\pincbBlu>]$ from the same premise 
$ \vlsbr[\strR;\strT] $.
%%%%%%%%%
\par\vspace{\baselineskip}\noindent
{\scriptsize
  \fbox{
    \begin{minipage}{.482\linewidth}
       \begin{equation}
        \label{equation:tracing-sequential-interactions-01}
	    \input{./SBV2-tracing-sequential-interactions-01}
       \end{equation}
    \end{minipage}
    \begin{minipage}{.482\linewidth}
      \begin{equation}
       \label{equation:tracing-sequential-interactions-00}
   	    \input{./SBV2-tracing-sequential-interactions-00}
       \end{equation}
    \end{minipage}
  }%fbox
}%\small
\par\vspace{\baselineskip}\noindent
%%%%%%%%%
They are \standard because every occurrence of $ \bvtatidrulein $ \emph{does not appear to the right-hand side} of an instance of $ \OpNameSeq $.
\end{example}
%%%%%%%%%%%%
\begin{remark}[\textbf{\textit{Proof-thoeretical meaning of standardization}}]
Standardization says that
(i) any of the structures inside $\strR$, and $\strT$ of $\vlsbr<\strR;\strT>$ will never interact, and 
(ii) all the interactions inside $\strR$ must occur before the interactions inside $\strT$.
\end{remark}
%%%%%%%%%%%%
\paragraph{Our goal} is to show that we can transform 
a \emph{sufficiently large} set of 
derivations in $ \BVT $ into standard ones. 
We start by supplying the main definitions.
%%%%%%%%
\paragraph{\Rightcontext s.} 
We rephrase, inductively, and extend to $\BVT$ the namesake definition in
\cite{Brus:02:A-Purely:wd}. The following grammar generates \emph{\rightcontext s} which we denote as $\vlholer{\strS\vlhole}$.
\par\vspace{\baselineskip}\noindent
{\small
  \fbox{
    \begin{minipage}{.974\linewidth}
      \begin{equation}
       \label{equation:SBV2-right-contexts-inductive}
       \input{./SBV2-right-contexts-inductive}
      \end{equation}
    \end{minipage}
  }%fbox
}%\small
%%%
\begin{example}[\textbf{\textit{\Rightcontext s}}]
\label{example:Rightcontexts}
A \rightcontext\ is
$\vlstore{\vlsbr[\atma;\vlfo{\atmc}{[\atmb;<\vlhole;\natmc;\atmd>]}]}\vlread$.\\
Instead,
$\vlstore{\vlsbr[\atma;\vlfo{\atmc}{[\atmb;<\natmc;\vlhole;\atmd>]}]}\vlread$ is
not.
\end{example}
%%%%%%%%%%
\paragraph{Left atomic interaction.} Recalling it from
\cite{Brus:02:A-Purely:wd}, the \emph{left atomic interaction} is:
\par\vspace{\baselineskip}\noindent
{\small
  \fbox{
    \begin{minipage}{.974\linewidth}
      \begin{equation}
       \label{equation:SBV2-atomic-interaction-left}
        \input{./SBV2-right-atomic-interaction}
      \end{equation}
    \end{minipage}
  }%fbox
}%\small
\par\vspace{\baselineskip}\noindent
%%%%%%%
\begin{example}[\textbf{\textit{Some left atomic interaction instances}}]
\label{example:}
Let three proofs of $\BVT$ be given:
\par\vspace{\baselineskip}\noindent
{\small
  \fbox{
    \begin{minipage}{.321\linewidth}
     \begin{equation}
       \label{equation:example-standardizable}
        \input{./SBV2-example-standardizable}
      \end{equation}
    \end{minipage}
    \begin{minipage}{.321\linewidth}
      \begin{equation}
       \label{equation:example-standardized}
        \input{./SBV2-example-standardized}
     \end{equation}
    \end{minipage}
    \begin{minipage}{.321\linewidth}
      \begin{equation}
      \label{equation:example-non-standardizable}
        \input{./SBV2-example-non-standardizable}
      \end{equation}
    \end{minipage}
  }%fbox
}%\small
\par\vspace{\baselineskip}\noindent
The two occurrences of $\bvtatidrulein$ in
\eqref{equation:example-standardizable} can correctly be seen as two
instances of $\bvtatrdrulein$, as outlined by \eqref{equation:example-standardized}.
Instead, the occurrence of $\bvtatidrulein$ in
\eqref{equation:example-non-standardizable} cannot be seen as an instance of
$\bvtatrdrulein$ as it occurs to the right of \OpNameSeq, namely in the context
$ \vlsbr<[\atma;\natma];\vlhole>$ which is not in~\eqref{equation:SBV2-right-contexts-inductive}.
\end{example}
%%%%
\begin{fact}
\label{fact:atil is ati}
By definition, every occurrence of $\bvtatrdrulein$ is one of $\bvtatidrulein$. The
vice versa is false.
\end{fact}
%%%%%%%%%
\paragraph{\Standardderivation s of $\BVT$.}
Let $ \strR $, and $ \strT $ be structures. 
A derivation $\bvtInfer{\bvtDder}{\strT \bvtJudGen{\BVT}{} \strR}$
is \emph{\standard} whenever all the atomic interactions that $\bvtDder$ contains
can be labeled as $\bvtatrdrulein$. We notice that nothing forbids $ \strT\approx\vlone $.
%%%%%%%%%%

%% file: SBV2-tracing-sequential-interactions-01.tex
\vlderivation                        {
  \vlin{\bvtseqdrule}{}
       {\vlsbr
        [<\atmaRed;\strR>
        ;<\natmbBlu;\strT>
        ;<\natmaRed;\atmbBlu>]}               {
  \vliq{\bvtatidrule}{}
       {\vlsbr
        [<[\atmaRed;\natmaRed]
         ;[\strR;\atmbBlu]>
        ;<\natmbBlu;\strT>]}{
  \vliq{\eqref{align:unit-seq}}{}
       {\vlsbr
        [\strR;\atmbBlu
        ;<\natmbBlu;\strT>]}{
  \vlin{\bvtseqdrule}{}
       {\vlsbr
        [\strR;<\atmbBlu;\vlone>
        ;<\natmbBlu;\strT>]}{
  \vliq{\bvtatidrule
       ,\eqref{align:unit-seq}
       ,\eqref{align:unit-pa}}{}
       {\vlsbr
        [\strR
        ;<[\atmbBlu;\natmbBlu]
        ;[\vlone;\strT]>]}  {
  \vlhy{\vlsbr[\strR;\strT]}}}}}}}

%% file: SBV2-tracing-sequential-interactions-00.tex
\vlderivation                        {
  \vlin{\bvtseqdrule}{}
       {\vlsbr
        [<\atmaRed;\strR>
        ;<\natmbBlu;\strT>
        ;<\natmaRed;\atmbBlu>]}      {
  \vliq{\eqref{align:unit-seq}}{}
       {\vlsbr
       [<[\atmaRed;\natmaRed]
        ;[\strR;\atmbBlu]>
       ;<\natmbBlu;\strT>]}{
  \vlin{\bvtseqdrule}{}
       {\vlsbr
       [<[\atmaRed;\natmaRed]
        ;[\strR;\atmbBlu]>
       ;<\vlone;<\natmbBlu;\strT>>]}{
  \vliq{\eqref{align:unit-seq}}{}
       {\vlsbr
       <[\atmaRed;\natmaRed;\vlone]
       ;[\strR;\atmbBlu;<\natmbBlu;\strT>]>}{
  \vlin{\bvtseqdrule}{}
       {\vlsbr
       <[[\atmaRed;\natmaRed]
        ;\vlone]
       ;[\strR
        ;<\atmbBlu;\vlone>
        ;<\natmbBlu;\strT>]>}{
  \vliq{\eqref{align:unit-pa}^2}{}
       {\vlsbr
        <[[\atmaRed;\natmaRed]
         ;\vlone]
        ;[\strR
         ;<[\atmbBlu
           ;\natmbBlu]
          ;[\vlone
           ;\strT]>]>}{
  \vliq{\eqref{align:unit-seq}}{}
       {\vlsbr
        <[\atmaRed;\natmaRed]
        ;[\strR
         ;<[\atmbBlu;\natmbBlu]
         ;\strT>]>}{
  \vlin{\bvtseqdrule}{}
       {\vlsbr
        <[\atmaRed;\natmaRed]
        ;[<\vlone
          ;\strR>
         ;<[\atmbBlu
           ;\natmbBlu]
          ;\strT>]>}{
  \vliq{\eqref{align:unit-pa}}{}
       {\vlsbr
        <[\atmaRed;\natmaRed]
        ;<[\vlone
          ;\atmbBlu
          ;\natmbBlu]
         ;[\strR
          ;\strT]>>}{
  \vlin{\bvtatidrule}{}
       {\vlsbr
        <[\atmaRed;\natmaRed]
        ;<[\atmbBlu;\natmbBlu]
         ;[\strR;\strT]>>}{
  \vlin{\bvtatidrule}{}
       {\vlsbr
        <\vlone
        ;<[\atmbBlu;\natmbBlu]
         ;[\strR;\strT]>>}{
  \vliq{\eqref{align:unit-seq}^2}{}
       {\vlsbr
        <\vlone
        ;<\vlone
         ;[\strR;\strT]>>}{
  \vlhy{\vlsbr[\strR;\strT]}}}}}}}}}}}}}}

%% file: SBV2-right-contexts-inductive.tex
\begin{aligned}
\vlholer{\strS\vlhole} 
   & \grammareq
           \vlhole
      \ \mid\ \vlsbr(\vlholer{\strS'\vlhole};\strR)
      \ \mid\ \vlsbr[\vlholer{\strS'\vlhole};\strR]
      \ \mid\ \vlsbr<\vlholer{\strS'\vlhole};\strR>\\
%\end{aligned}
%\qquad\qquad\qquad\qquad
%\begin{aligned}
%\phantom{\vlholer{\strS\vlhole}} & 
      & \phantom{\grammareq\vlhole}
      \,\ \mid\ \vlsbr(\strR;\vlholer{\strS'\vlhole})
      \ \mid\ \vlsbr[\strR;\vlholer{\strS'\vlhole}]
      \ \mid\ \vlfo{\atma}{\vlholer{\strS'\vlhole}}
\end{aligned}

%% file: SBV2-right-atomic-interaction.tex
\vlderivation{
\vlin{\bvtatrdrulein}{}
      {\vlholer{\strS\vlsbr[\atma;\natma]}}{
\vlhy{\vlholer{\strS\vlscn{\vlone}}}}}

%% file: SBV2-example-standardizable.tex
\vlderivation{
\vlin{\bvtatidrule}{}
     {\vlsbr<[\atma;\natma]
 	    ;[\atmb;\natmb]>}{
\vliq{\eqref{align:unit-seq}}{}
     {\vlsbr<\vlone
	    ;[\atmb;\natmb]>}{
   \vlin{\bvtatidrule}{}
     {\vlsbr[\atmb;\natmb]}{
\vlhy{\vlone}}      }}}

%% file: SBV2-example-standardized.tex
\vlderivation{
  \vlin{\bvtatrdrule}{}
      {\vlsbr<[\atma;\natma]
	      ;[\atmb;\natmb]>}{
  \vliq{\eqref{align:unit-seq}}{}
      {\vlsbr<\vlone
	      ;[\atmb;\natmb]>}{
    \vlin{\bvtatrdrule}{}
      {\vlsbr[\atmb;\natmb]}{
  \vlhy{\vlone}}      }}}

%% file: SBV2-example-non-standardizable.tex
\vlderivation{
  \vlin{\bvtatidrule}{}
      {\vlsbr<[\atma;\natma]
	      ;[\atmb;\natmb]>}{
  \vliq{\eqref{align:unit-seq}}{}
      {\vlsbr<[\atma;\natma]
	      ;\vlone>}{
    \vlin{\bvtatrdrule}{}
      {\vlsbr[\atma;\natma]}{
  \vlhy{\vlone}}      }}}

%% file: SBV2-standardization-weak.tex
\subsection{Standardization} 
\label{subsection:Standardization}
We reorganize derivations of
$\Set{\bvtatrdrulein,\bvtatidrulein,\bvtseqdrulein,\bvtrdrulein}\subset\BVT$
which operate on \dfn{\OpNameTen-free structures} only.
%%%%%%%%%
\paragraph{\OpNameTen-free structures.} By definition, $ \strR $ in $ \BVT $ is \dfn{\OpNameTen-free} whenever it does not contain $\vlsbr(\strR_1;\vldots;\strR_n)$, for any $\strR_1,\ldots,\strR_n$, and $n>1$.
%%%%%
\paragraph{Our goal} is to prove the following theorem, inspiring to the standardization in \cite{Brus:02:A-Purely:wd}:
%%%%%%%
\begin{theorem}[\textbf{\textit{Standardization in 
$\Set{\bvtatrdrulein,\bvtatidrulein,\bvtseqdrulein,\bvtrdrulein}$}}]
\label{theorem:Standardization in bvtatrdrulein...}
Let $\strT$, and $\strR$ be \OpNameTen-free. 
For every 
$\bvtInfer{\bvtDder}
 {\strT 
  \bvtJudGen{\Set{\bvtatrdrule,\bvtatidrule,\bvtseqdrule,\bvtrdrule}}
            {} 
  \strR}$,
there is a \standardderivation\ 
$\bvtInfer{\bvtEder} 
 {\strT 
  \bvtJudGen{\Set{\bvtatrdrule,\bvtseqdrule,\bvtrdrule}}{} 
  \strR}$.
\end{theorem}
\par\vspace{\baselineskip}\noindent
%%%%%%%%
It proof relies on the coming lemmas, and proposition.
%%%%%%%%%
\begin{lemma}[\textbf{\textit{Existence of $\bvtatrdrulein$}}] 
\label{lemma:Existence of bvtatrdrulein}
The topmost instance of $\bvtatidrulein$ in a proof 
$\bvtInfer{\bvtPder} 
 {\
  \bvtJudGen{\BVT}{} 
  \strR}$ is always an instance of $\bvtatrdrulein$.
\end{lemma}
%%%%%%%
\begin{proof}
Let $\bvtPder$ be
$
\vlderivation{
\vlde{}{\BVT}
     {\strR}{
\vlin{\bvtatidrulein^{\bullet}}{}
     {\strS\vlsbr[\atma;\natma]}{
\vlpr{\bvtQder}{\Set{\bvtatrdrule,\bvtseqdrule,\bvtrdrule}}
     {\strS\vlscn{\vlone}}}}}
$ with $\bvtatidrulein^{\bullet}$ its topmost instance of $\bvtatidrulein$ which cannot be relabeled as $\bvtatrdrulein$. By contraction, let us assume $\strS\vlhole$ be a non \rightcontext, namely
$\strS\vlhole
 \approx\strS'\vlsbr<\strT;\strS''\vlhole>$ for some
$\strS'\vlhole,\strS''\vlhole$, and $\strT$ \ST $\strT\not\approx\vlone$.
In this case, to let the names of $\strT$, and, may be, those ones of $\strS''\,\vlscn{\vlone}$, to disappear from
$
\vlderivation{
\vlin{\bvtatidrulein^{\bullet}}{}
     {\strS'\vlsbr<\strT;\strS''[\atma;\natma]>}{
\vlpr{\bvtQder}{\Set{\bvtatrdrule,\bvtseqdrule,\bvtrdrule}}
     {\strS'\vlsbr<\strT;\strS''\,\vlscn{\vlone}>}}}
$ we would have to apply at least one instance of $\bvtatidrulein$ which would occur in $\bvtQder$, against our assumption on the position of $\bvtatidrulein^{\bullet}$.
\end{proof}
%%%%%%%%%
\begin{lemma}[\textbf{\textit{Commuting conversions in 
$\Set{\bvtatrdrulein,\bvtatidrulein,\bvtseqdrulein,\bvtrdrulein}$}}]
\label{lemma:bvtatidrulein commuting conversions}
Let $\strR, \strT$, and $\strS\vlscn{\vlone}$ be \OpNameTen-free.
Also, let $\bvtrhorule\in\Set{\bvtatrdrulein,\bvtseqdrulein,\bvtrdrulein}$.
Finally, let $\bvtDder$ be
$\
% \vlupsmash{
  \vlderivation{
  \vlin{\bvtatidrule^{\bullet}}{}
       {\strR}{
  \vlin{\bvtrhorule}{}
       {\vlholer{\strS\vlsbr[\atma;\natma]}}{
  \vlhy{\strT}}}}
%  }
$,
where $\bvtatidrulein^{\bullet}$ is the topmost occurrence of 
$\bvtatidrulein$ which is not $\bvtatrdrulein$.
Then, there is
$
% \vldownsmash{
 \vlderivation{
 \vlde{\bvtDder}
      {\Set{\bvtatrdrule,\bvtatidrule,\bvtseqdrule,\bvtrdrule}}
      {\strR}{
 \vlin{\bvtatidrule^{*}}{}
      {\strV}{
 \vlhy{\strT}}}}
% }
$, where $\strV$, and all the structures of $\bvtDder$
are \OpNameTen-free, and $\bvtatidrulein^{*}$ may be an instance of $\bvtatrdrulein$.
\end{lemma}
\begin{proof}
The proof is, first, by cases on $\bvtrhorulein$, and, then, by cases on $\vlstore{\vlholer{\strS\vlsbr[\atma;\natma]}}\vlread$.
Fixed $\vlstore{\vlholer{\strS\vlsbr[\atma;\natma]}}\vlread$, the proof is by cases on $\strR$ which must contain a redex of
$\bvtatidrulein, \bvtseqdrulein$, or $\bvtrdrulein$, that, after
$\bvtatidrulein^{\bullet}$, leads to the chosen $\vlstore{\vlholer{\strS\vlsbr[\atma;\natma]}}\vlread$.
(Appendix~\ref{section:Proof of lemma:bvtatidrulein commuting conversions}.)
\end{proof}
%%%%%%%%%%%%%%%%%%
\begin{proposition}[\textbf{\textit{One-step standardization in $ \Set{\bvtatrdrulein,\bvtatidrulein,\bvtseqdrulein,\bvtrdrulein} $}}]
\label{proposition:One-step standardization of BVTCC}
Let
$
%  \vlupsmash{
  \vlderivation{
  \vlde{\bvtDder'}{}
       {\strR}{
  \vlin{\bvtatidrule^\bullet}{}
       {\strU}{
  \vlde{\bvtDder''}{}
       {\strV}{
  \vlhy{\strT}}}}}
%   }
$ be a derivation 
in $\Set{\bvtatrdrulein,\bvtatidrulein,\bvtseqdrulein,\bvtrdrulein}$ such that $\bvtatidrulein^\bullet$ is the topmost instance of $\bvtatidrulein$.
There exists a derivation 
$\bvtInfer{\bvtEder}
          {\strT \bvtJudGen{\Set{\bvtatrdrule,\bvtatidrule,\bvtseqdrule,\bvtrdrule}}
                 {}
           \strR}$
where $\bvtatidrulein^\bullet$ has been eventually moved upward to transform it into
an instance of $\bvtatrdrulein$.
\end{proposition}
%%%%%%
\begin{proof}
Let $n$ be the number of rules in $\bvtDder''$.
If $\strU\approx
    \vlstore{\vlsbr[\atma;\natma]}\vlholer{\strS\vlread}$, with $ \vlsbr[\atma;\natma] $
the redex of $\bvtatidrulein^\bullet$, then $\bvtatidrulein^\bullet$ is already
an instance of $\bvtatrdrulein$, and we are done.
Otherwise, we can apply Lemma ~\ref{lemma:bvtatidrulein commuting conversions}
moving $\bvtatidrulein^\bullet$ one step upward, getting to
$\bvtInfer{\bvtEder}
          {\strT \bvtJudGen{\Set{\bvtatrdrule,\bvtatidrule,\bvtseqdrule,\bvtrdrule}}
                           {}
           \strR}$, where
$\bvtatidrulein^\bullet$ is no more than $n-1$ rules far from $\strT$.
An obvious inductive argument allows to conclude thanks to Lemma~\ref{lemma:Existence of bvtatrdrulein}.
\end{proof}
%%%%%%%%%%%%%%%%%%
\paragraph{Proof of Theorem~\ref{theorem:Standardization in bvtatrdrulein...}.}
Let $X_{\bvtDder}$ be the set of all instances of $\bvtatidrulein$ in $\bvtDder$,
that can be directly seen as instances of $\bvtatrdrulein$, and $Y_{\bvtDder}$
the set of all other instances of $\bvtatidrulein$ in $\bvtDder$.
If $Y_{\bvtDder}=\emptyset$ we are done because $\bvtEder$ is $\bvtDder$ where every
instance of $\bvtatidrulein$ in $X_{\bvtDder}$, if any, can be directly
relabeled as $\bvtatrdrulein$. Otherwise, let us pick the topmost occurrence of
$\bvtatidrulein$ in $\bvtDder$ out of $Y_{\bvtDder}$, and apply
Proposition~\ref{proposition:One-step standardization of BVTCC} to it.
We get 
$\bvtInfer{\bvtEder}
          {\strT \bvtJudGen{\Set{\bvtatrdrule,\bvtatidrule,\bvtseqdrule,\bvtrdrule}}
                           {}
           \strR}$, whose set
$Y_{\bvtEder}$ is strictly smaller than $Y_{\bvtDder}$.
An obvious inductive argument allows to conclude.
%%%%%
\paragraph{Standard fragment $\BVTL$ of $\BVT$.} 
After Theorem~\ref{theorem:Standardization in bvtatrdrulein...} it is sensible defining 
$\BVTL$ as $\Set{\bvtatrdrulein,\bvtseqdrulein,\bvtrdrulein}\subset\BVT$ whose derivations contain \OpNameTen-free only structures.

%% file: SBV2-deduction.tex
\section{Internalizing derivability of $ \BVT$}
\label{section:Internalizing the derivability in BVT}
Roughly, internalizing derivability in $ \BVT $ shows when we can ``discharge assumptions''. It is another of the properties we need to recast reachability problems in a suitable calculus of communicating, and concurrent processes, to proof-search inside (a fragment) of $ \BVT $.
The internalization links to the notion of \invertiblestructure s.
%%%%%%
\paragraph{\Invertible, and \coinvertiblestructure s.}
We define them in~\eqref{equation:SBV2-invertible-structures-in-words} here below.
\par\vspace{\baselineskip}\noindent
{
%\small
  \fbox{
    \begin{minipage}{.95\linewidth}
      \begin{equation}
       \label{equation:SBV2-invertible-structures-in-words}
       \begin{minipage}{.8\linewidth}
         \begin{center}
          \input{./SBV2-invertible-structures-in-words}
         \end{center}
       \end{minipage}
      \end{equation}
    \end{minipage}
  }%fbox
}%\small
\par\vspace{\baselineskip}\noindent
If $\strT$ is invertible, then, by definition, $\vlne{\strT}$ is \dfn{\coinvertible}.
%%%%%%
\begin{remark}
Clearly, definition~\eqref{equation:SBV2-invertible-structures-in-words} here above omits
the implication
``If  $\bvtInfer{\bvtDder}
                {\vlne{\strT}
                \bvtJudGen{\BVT}{}
                \strP}$, then 
$\vlstore{\vlsbr[\strT;\strP]}
 \bvtInfer{\bvtPder}{\ \bvtJudGen{\BVT}{} \vlread}$'' on purpose.
It always holds because $ \bvtintdrulein $ is derivable in $ \BVT $.
Moreover, our \invertiblestructure s inspire to the namesake concept in \cite{Stra:03:System-N:mb}.
\end{remark}
\par\vspace{\baselineskip}\noindent
%%%%%%%%%%%%%%%
The following proposition gives sufficient conditions for a structure to be \invertible.
\begin{proposition}[\textit{\textbf{A language of \invertiblestructure s}}]
\label{proposition:Invertible structures are invertible}
The following grammar \eqref{equation:SBV2-invertible-structures} generates \invertiblestructure s.
\par\vspace{\baselineskip}\noindent
{\small
  \fbox{
    \begin{minipage}{.95\linewidth}
      \begin{equation}
       \label{equation:SBV2-invertible-structures}
       \input{./SBV2-invertible-structures}
      \end{equation}
    \end{minipage}
  }%fbox
}%\small
\end{proposition}
\begin{proof}
Let $\vlstore{\vlsbr[\vlne{\strT};\strP]}
 \bvtInfer{\bvtPder}{\ \bvtJudGen{\BVT}{} \vlread}$ be given with 
$\vlne{\strT}$ in \eqref{equation:SBV2-invertible-structures}. We reason  by induction on $\vlstore{\vlsbr[\vlne{\strT};\strP]}
 \Size{\vlread}$, and we build $\bvtDder$
of~\eqref{equation:SBV2-invertible-structures-in-words}, proceeding by cases on $\vlne{\strT}$. (Details in Appendix~\ref{section:Proof of proposition:Invertible structures are invertible}.)
\end{proof}

%% file: SBV2-invertible-structures-in-words.tex
$\strT$ is \dfn{invertible} whenever
$\vlderivation{
 \vlpd{\bvtPder}{\BVT}
      {\vlsbr[\strT;\strP]}
      {}
 }$  
\ implies\ \
$\vlderivation{
 \vldd{\bvtDder}{\BVT}
      {\strP}
      {
 \vlhy{\vlne{\strT}}}
 }$,
for every $\strT$, and $\strP$

%% file: SBV2-invertible-structures.tex
\begin{aligned}
&\strT  \grammareq 
             \vlone
              \mid \vlsqbrl
                    \atmLabL_1 \vlpa \vldots \vlpa \atmLabL_n
                    \vlsqbrr
              \mid \vlrobrl \strT \vlte \strT \vlrobrr
              \mid  \vlsbr<\strT;\strT>
              \mid   \vlfo{\atma}{\strT}
              \\
&              \textrm{ where } 
              n > 0, 
              \textrm{ and, for every } 
              1\leq i,j\leq n,
              \textrm{ if } 
              i\neq j 
              \textrm{ then } 
              \atmLabL_i \neq  \vlne{\atmLabL_j}
\end{aligned}

%% file: BV2-communication-core.tex
\section{Intermezzo}
\label{section:Communication core of BVT}
We keep the content of this section at an intuitive level.
We describe how structures of $ \BVT $ model terms in a language whose syntax is not formally identified yet, but which is related to the one of Milner $ \CCS $. 
%Consequently, derivations of $ \BVT $ model derivation in a \lts, related to the one of Milner $ \CCS $.
%%%%
\begin{example}[\textbf{\textit{Modeling internal communication inside $\BVT$}}]
\label{example:Modeling internal communication inside BVT}
Derivations of $\BVT$ model internal communication
if we look at structures of $\BVT$ as they were terms of Milner $\CCS$, as in
\cite{Brus:02:A-Purely:wd}. Let us focus on
\eqref{equation:PPi-internal-interaction-example-ll} here below.
\par\vspace{\baselineskip}\noindent
{\small
  \fbox{
    \begin{minipage}{.487\textwidth}
      \begin{equation}
       \label{equation:PPi-internal-interaction-example-ll}
  	   \input{./PPi-internal-interaction-example-ll}
      \end{equation}
    \end{minipage}
%%%%%%
    \begin{minipage}{.487\textwidth}
      \begin{equation}
      \label{equation:PPi-internal-interaction-example-rr}
	\input{./PPi-internal-interaction-example-rr}
      \end{equation}
    \end{minipage}
  }%fbox
}%\small
\par\vspace{\baselineskip}\noindent
The instance of $\bvtseqdrulein$ moves atoms $\atma$, and $\natma$, one aside the
other, and $\bvtatidrulein$ annihilates them. Annihilation can be seen as an
internal communication between the two components
$\vlsbr<\atmaRed;\pincE>$, and $\vlsbr<\natmaBlu;\pincF>$ of the structure
$\vlsbr[<\atmaRed;\pincE>;<\natmaBlu;\pincF>]$. The usual way to formalize
such an internal communication is
\eqref{equation:PPi-internal-interaction-example-rr}, derivation that belongs to the 
\lts\ of Milner $\CCS$. The sequential composition of
\eqref{equation:PPi-internal-interaction-example-rr} stands for
\OpNameSeq, parallel composition for \OpNamePar, and both
$\pincE$, and $\pincF$ in \eqref{equation:PPi-internal-interaction-example-ll} are
represented by corresponding processes $\pincE$, and $\pincF$ in
\eqref{equation:PPi-internal-interaction-example-rr}.
\end{example}
%%%%%%%%
\begin{example}[\textbf{\textit{Modeling external communication inside $\BVT$}}]
\label{example:Modeling external communication inside BVT}
Derivations of $\BVT$ model external communication if we
look at structures of $\BVT$ as they were terms of Milner $\CCS$, as in \cite{Brus:02:A-Purely:wd}.
Let us focus on \eqref{equation:PPi-external-interaction-example-ll} here below.
\par\vspace{\baselineskip}\noindent
{\small
  \fbox{
    \begin{minipage}{.487\textwidth}
      \begin{equation}
      \label{equation:PPi-external-interaction-example-ll}
	\input{./PPi-external-interaction-example-ll}
      \end{equation}
    \end{minipage}
%%%
    \begin{minipage}{.487\textwidth}
%       \vspace{1.3\baselineskip}
      \begin{equation}
      \label{equation:PPi-external-interaction-example-rr}
	\input{./PPi-external-interaction-example-rr}
      \end{equation}
%       \vspace{1.3\baselineskip}
    \end{minipage}
  }%fbox
}%\small
\par\vspace{\baselineskip}\noindent
We look at $\vlsbr[<\atmaRed;\pincE>;\natmaBlu]$ as
containing two sub-structures with different meaning. The structure
$\vlsbr<\atmaRed;\pincE>$ corresponds to the process
$\pincSec{\atma}{\pincE}$. Instead, $\natmaBlu$ can be seen as an
action of the context ``around'' $\vlsbr<\atmaRed;\pincE>$. This means that
\eqref{equation:PPi-internal-interaction-example-ll} formalizes Milner $\CCS$ derivation
\eqref{equation:PPi-internal-interaction-example-rr}.
\end{example}
%%%%%%%%%%%%%%%%%
\begin{remark}[\textbf{\textit{``Processes'', and  ``contexts'' are
first-citizens}}]
\label{remark:Processes and  contexts are first-citizens}
The structure $\vlsbr[<\atmaRed;\pincE>;\natmaBlu]$ is equivalent to $\vlsbr
[<\atmaRed;\pincE> ;<\natmaBlu;\vlone>]$ in \eqref{equation:PPi-external-interaction-example-ll}. This highlights a first difference
between modeling the communication by means of (a sub-system of) $\BVT$,
instead than with Milner $\CCS$.
This latter constantly separates terms from the contexts they interact with.
Instead, the structures of $\BVT$ make no difference, and represent contexts as
first-citizens. Namely, choosing which structures are the ``real
processes'', and which are ``contexts'' is, somewhat, only matter of taste.
Specifically, in our case, we could have said that $ \vlsbr{<\natmaBlu;\vlone>} $ represents the process $ \pincSec{\pincna}{\pincZer}$, instead than the context.
\end{remark}
%%%%%%%%%%%%%%%%%
\begin{example}[\textbf{\textit{Hiding communication}}]
\label{example:Hiding communication}
Derivations in $\BVT$ model
hidden communications of Milner $\CCS$ thanks to \OpNameRen. So, we strictly extend the correspondence between a \DI system and Milner $ \CCS $, as given in \cite{Brus:02:A-Purely:wd}. We build on Example~\ref{example:Modeling external communication inside BVT}, placing an instance of \OpNameRen around every of the two components of
$\vlsbr[<\atmaRed;\pincE>;\natmaBlu]$ in
\eqref{equation:PPi-external-interaction-example-ll}.
\par\vspace{\baselineskip}\noindent
{\small
  \fbox{
    \begin{minipage}{.487\textwidth}
      \begin{equation}
       \label{equation:PPi-external2internal-interaction-example-ll}
	\input{./PPi-external2internal-interaction-example-ll}
      \end{equation}
    \end{minipage}
%%%%
    \begin{minipage}{.487\textwidth}
      \begin{equation}
      \label{equation:PPi-external2internal-interaction-example-rr}
	\input{./PPi-external2internal-interaction-example-rr}
      \end{equation}
    \end{minipage}
  }%fbox
}%\small
\par\vspace{\baselineskip}\noindent
We can look at \OpNameRen, which binds $\atmaRed$, and $\natmaBlu$ as restricting the
visibility of the communication. The derivation in the \lts\ of Milner $\CCS$ that models
\eqref{equation:PPi-external2internal-interaction-example-ll} is
\eqref{equation:PPi-external2internal-interaction-example-rr}.
\end{example}
%%%%%%%%%%%%%%%%%
\begin{example}[\textbf{\textit{More freedom inside $\BVT$}}]
\label{example:Escaping close correspondence with CCS}
Inside  
$\vlsbr
      [<\atmaRed;\pincE>
      ;<\natmaBlu;<\natmbBlu;<\natmcBlu;\pincF>>>
      ;<\atmbRed;<\atmcRed;\vlne\pincF>>]$, of
\eqref{equation:PPi-invertible-interaction-example-ll} among others, we can identify the ``processes''
$\pincG_1\equiv \vlsbr<\atmaRed;\pincE>$,
$\pincG_2\equiv \vlsbr<\natmaBlu;<\natmbBlu;<\natmcBlu;\pincF>>>$,
$\pincG_3\equiv \vlsbr<\natmbBlu;<\natmcBlu;\pincF>>$, and
$\pincG_4\equiv \vlsbr<\atmbRed;<\atmcRed;\vlne\pincF>>$:
\par\vspace{\baselineskip}\noindent
{\small
  \fbox{
   \begin{minipage}{.974\linewidth}
      \begin{equation}
      \label{equation:PPi-invertible-interaction-example-ll}
	\input{./PPi-invertible-interaction-example-ll}
      \end{equation}
    \end{minipage}
  }%fbox
}%\small
\par\vspace{\baselineskip}\noindent
The lowermost instance of $\bvtseqdrulein$ predisposes $\pincG_1$, and $\pincG_2$
to an interaction through $\atmaRed$, and $\natmaBlu$. However, only the instance of
$\bvtatidrulein$ makes the interaction effective. Before that, the instance of
$\bvtintdrulein$ identifies $\pincG_4$ as the negation of $\pincG_3$,
and annihilates them in a whole. So, \eqref{equation:PPi-invertible-interaction-example-ll} suggests that modeling process computations inside $\BVT$ may result more flexible than usual, because it introduces a notion of ``negation of a process'' which sounds as a higher-order ingredient of proof-search-as-computation.
\end{example}

%% file: PPi-internal-interaction-example-ll.tex
\vlderivation                      {
\vlin{\bvtseqdrule}{}
     {\vlsbr
       [<\atmaRed;\pincE>
       ;<\natmaBlu;\pincF>]
     }                             {
\vlin{\bvtatidrule}{}
     {\vlsbr
       <[\atmaRed;\natmaBlu]
       ;[\pincE;\pincF]>
     }                             {
\vliq{\eqref{align:unit-seq}}{}
     {\vlsbr
       <\vlone
       ;[\pincE;\pincF]>
     }                             {
\vlhy{\vlsbr[\pincE;\pincF]}}}}}

%% file: PPi-internal-interaction-example-rr.tex
\vlderivation                                                {
\vliin{}{}
      {\pincLTSJud
      {\pincPar{\pincSec{\atma}{\pincE}}
		{\pincSec{\natma}{\pincF}}
      }
      {\pincPar{\pincE}
	       {\pincF}
      }
      {\pincPreT}}{
\vlhy{\pincLTSJud
	{\pincSec{\atma}{\pincE}}
	{\pincE}
	{\atma}}
      }{
%%%%%%%%%%%
\vlhy{\pincLTSJud
	{\pincSec{\natma}{\pincF}}
	{\pincF}
	{\natma}}
	}}

%% file: PPi-external-interaction-example-ll.tex
\vlderivation                                              {
\vliq{\eqref{align:unit-seq}
     ,\bvtseqdrule}{}
     {\vlsbr[<\atmaRed;\pincE>
             ;\natmaBlu]
     }                                                     {
\vlin{\bvtatidrule}{}
     {\vlsbr
       <[\atmaRed;\natmaBlu]
       ;[\pincE;\vlone]>
     }                                                     {
\vliq{\eqref{align:unit-seq}
     ,\eqref{align:unit-pa}}{}
     {\vlsbr
       <\vlone
       ;[\pincE;\vlone]>
     }                                                     {
\vlhy{\pincE}}}}}

%% file: PPi-external-interaction-example-rr.tex
\vlderivation                       {
\vlin{}{}
      {\pincLTSJud
        {\pincSec{\atma}{\pincE}}
        {\pincE}
        {\atma}}                    {
\vlhy{}                             }}

%% file: PPi-external2internal-interaction-example-ll.tex
\vlderivation                                              {
\vlin{\bvtrdrule}{}
     {\vlsbr
       [\vlfo{\atma}{<\atmaRed;\pincE>}
       ;\vlfo{\atma}{\natmaBlu}]
     }                                                     {
\vliq{\eqref{align:unit-seq}
     ,\bvtseqdrule
     ,\eqref{align:unit-pa}}{}
     {\vlfo{\atma}
      {\vlsbr[<\atmaRed;\pincE>
             ;\natmaBlu]}
     }                                                     {
\vlin{\bvtatidrule}{}
     {\vlfo{\atma}{
      \vlsbr
       <[\atmaRed;\natmaBlu]
       ;\pincE>
      }
     }                                                     {
\vliq{\eqref{align:unit-seq}}{}
     {\vlfo{\atma}{
      \vlsbr
       <\vlone;\pincE>
      }}                                                   {
\vlhy{\vlfo{\atma}{\pincE}}}}}}}

%% file: PPi-external2internal-interaction-example-rr.tex
\vlderivation                                           {
\vlin{}{}
      {\pincLTSJud{\pincNu{\atma}{(\pincSec{\atma}{\pincE})}}
      {\pincNu{\atma}{\pincE}}
      {\pincPreT}}{
\vlhy{\pincLTSJud
      {\pincSec{\atma}{\pincE}}
      {\pincE}
      {\atma}}    }}

%% file: PPi-invertible-interaction-example-ll.tex
\vlderivation                      {
\vlin{\bvtseqdrule}{}
     {\vlsbr
      [<\atmaRed;\pincE>
      ;<\natmaBlu;<\natmbBlu;<\natmcBlu;\pincF>>>
      ;<\atmbRed;<\atmcRed;\vlne{\pincF}>>]
     }                             {
\vliq{\eqref{align:unit-seq}}{}
     {\vlsbr
       [<[\atmaRed;\natmaBlu]
        ;[\pincE;<\natmbBlu;<\natmcBlu;\pincF>>]>
        ;<\atmbRed;<\atmcRed;\vlne{\pincF}>>]
     }                             {
\vlin{\bvtseqdrule}{}
     {\vlsbr
       [<[\atmaRed;\natmaBlu]
        ;[\pincE;<\natmbBlu;<\natmcBlu;\pincF>>]>
        ;<\vlone;<\atmbRed;<\atmcRed;\vlne{\pincF}>>>]
     }                             {
\vliq{\eqref{align:unit-pa}}{}
     {\vlsbr
       <[\atmaRed;\natmaBlu;\vlone]
       ;[\pincE;<\natmbBlu;<\natmcBlu;\pincF>>
        ;<\atmbRed;<\atmcRed;\vlne{\pincF}>>
       ]>
     }                             {
\vlin{\bvtintdrule}{}
     {\vlsbr
      <[\atmaRed;\natmaBlu]
       ;[\pincE;<\natmbBlu;<\natmcBlu;\pincF>>
        ;<\atmbRed;<\atmcRed;\vlne{\pincF}>>
       ]>
     }                             {
\vlin{\bvtatidrule}{}
     {\vlsbr
      <[\atmaRed;\natmaBlu]
       ;\pincE>
     }                             {
\vliq{\eqref{align:unit-seq}}{}
     {\vlsbr
      <\vlone
       ;\pincE>
     }                             {
\vlhy{\pincE}}}}}}}}}

%% file: PPi.tex
\section{Communication, and concurrency with logic restriction}
\label{section:Communicating, and concurrent processes with restriction}
The correspondences Section~\ref{section:Communication core of BVT} highlights, justify the introduction of a calculus of processes which we identify as $\CCSR$.
Specifically, $\CCSR$ is a calculus of communicating, and concurrent processes, with a logic-based restriction, whose operational semantics is driven by the logical behavior of 
$ \bvtrdrulein $ rule.
%%%%%%%%
\begin{remark}[\textbf{\textit{$\CCSR$ vs. Milner $\CCS$}}]
\label{remark:CCSR vs Milners CCS}
It will turn out that $\CCSR$ is not Milner $\CCS$ \cite{Miln:89:Communic:qo}
%, but certainly related to it
. The concluding Section~\ref{section:Final discussion, and future work} will discuss on this.
\end{remark}
%%%
\paragraph{Actions on terms of $\CCSR$.}
Let $\pinca, \pincb, \pincc, \ldots$ denote the elements of a countable set of
\dfn{names}, and let $\pincna, \pincnb, \pincnc, \ldots$ denote the elements of a
countable set of \dfn{co-names}. The set of \dfn{labels}, which we range over by
$\pincLabL$\,, $\pincLabM$, and $\pincLabN$ contains both names, and co-names, and
nothing else. Let $\pincLabT$ be
the \dfn{silent}, or \dfn{perfect action}, different from any name, and
co-name. The (set of) \dfn{sequences of actions} contains equivalence classes
defined on the language that \eqref{equation:PPi-LTS-labels} yields:
%%%%%%%%%%%%%
\par\vspace{\baselineskip}\noindent
{\small
  \fbox{
    \begin{minipage}{.974\linewidth}
%     \vspace{-.3cm}
      \begin{equation}
        \label{equation:PPi-LTS-labels}
	\input{./PPi-LTS-labels}
      \end{equation}
    \end{minipage}
  }%fbox
}%\small
\par\vspace{\baselineskip}\noindent
%%%%%%%%%%%%
By definition, the equivalence relation~\eqref{align:PPi-labels-congruence} here below induces the congruence $\pincCong$ on~\eqref{equation:PPi-LTS-labels}.
%%%%%%%%%%%%
\par\vspace{\baselineskip}\noindent
{\small
  \fbox{
   \begin{minipage}{.974\linewidth}
      \begin{equation}
	\label{align:PPi-labels-congruence}
	\input{./PPi-labels-congruence}
      \end{equation}
    \end{minipage}
  }%fbox
}%\small
\par\vspace{\baselineskip}\noindent
%%%%%%%%%%%
We shall use $\atmalpha, \atmbeta$, and $\pincgamma$ to range over the elements in
the set of actions sequences.
%%%
\paragraph{Processes of $\CCSR$.}
The terms of $\CCSR$, \ie \dfn{processes}, belong to the language of the grammar
\eqref{align:PPi-syntax} here below.
%%%%%%%%%%%%%%%%%%%%%
\par\vspace{\baselineskip}\noindent
{\small
  \fbox{
   \begin{minipage}{.974\linewidth}
      \begin{equation}
	  \label{align:PPi-syntax}
	  \input{./PPi-syntax}
      \end{equation}
    \end{minipage}
  }%fbox
}%\small
\par\vspace{\baselineskip}\noindent
%%%%%%%%%%%%%%%%%%%%%
We use $ \pincE, \pincF, \pincG$, and $ \pincH $ to range over processes.
The \dfn{inactive process} is $\pincZer$,
the \dfn{parallel composition} of $\pincE$, and $\pincF$ is
$\pincPar{\pincE}{\pincF}$. 
The \dfn{sequential composition}
$\pincSec{\pincLabL}{\pincE}$ sets the occurrence of the \dfn{action prefix}
$\pincLabL$ before the occurrence of $\pincE$.
\dfn{Logic restriction} $\pincNu{\pinca}{\pincE}$ hides all, and only, the occurrences of
$\pinca$, and $ \pincna $, inside $ \pincE $, which becomes invisible outside $\pincE$.
%%%
\paragraph{Size of processes.} The size $\Size{\pincE}$ of $\pincE$ is the number of symbols of $\pincE$.
%%%
\paragraph{Congruence on processes of $\CCSR$.}
We partition the processes of $\CCSR$ up to the smallest congruence which, by abusing
notation, we keep calling $\pincCong$, and which we obtain as reflexive, transitive, and contextual closure of the relation \eqref{align:PPi-structural-congruence} here below.
%%%%%%%%%%%%%%%%%%%%
\par\vspace{\baselineskip}\noindent
{\small
  \fbox{
   \vspace{-.2cm}
   \begin{minipage}{.974\linewidth}
      \begin{equation}
	\label{align:PPi-structural-congruence}
	\input{./PPi-structural-congruence}
      \end{equation}
    \end{minipage}
  }%fbox
}%\small
\par\vspace{\baselineskip}\noindent
In~\eqref{align:PPi-structural-congruence}
(i) $\pincE\subst{\pinca}{\pincb}$ denotes a standard clash-free substitution of $\pinca$ for both $\pincb$, and $\pincnb$ in $\pincE$ that we can define as usual, and
%in analogy to \eqref{equation:BV2-structure-substitution}, and
% %%%%%%%%%%%%% NON CANCELLARE: contiene la definizione esplicita di
% %%%%%%%%%%%%% sostituizione sui termini di CCSr.
% \par\vspace{\baselineskip}\noindent
%   \fbox{
%     \begin{minipage}{.974\linewidth}
%      {\small
%       \vspace{-.3cm}
%       \begin{equation}
%         \label{equation:PPI-process-substitution}
%         \input{./PPI-process-substitution}
%       \end{equation}
%      }%\small
%     \end{minipage}
%   }%fbox
% \par\vspace{\baselineskip}\noindent
% %%%%%%%%%%%%%
(ii) $\pincFN{\cdot}$ is the set of free-names of a term in $\CCSR$, whose definition, again, is the obvious one. Namely, neither
$\pinca$, nor $\pincna$ belong to the set $\pincFN{\pincNu{\pinca}{\pincE}}$.
%%%%%%%
\paragraph{\Lts\ of $\CCSR$.}
Its rules are in~\eqref{equation:PPi-LTS-from-BVT}, and they justify why $ \CCSR $ is not Milner $ \CCS $.
%%%%%%%%%%%%%
\par\vspace{\baselineskip}\noindent
{\small
  \fbox{
    \begin{minipage}{.974\linewidth}
%     \vspace{-.1cm}
      \begin{equation}
        \label{equation:PPi-LTS-from-BVT}
   	    \input{./PPi-LTS-from-BVT}
      \end{equation}
    \end{minipage}
  }%fbox
}%\small
\par\vspace{\baselineskip}\noindent
%%%%%%%%%%%%%
In~\eqref{equation:PPi-LTS-from-BVT}, the rule $\pincact$ implements external
communication, by firing the action prefix $\pincLabL$, as usual. The rule $\pinccom$ implements internal
communication, annihilating two complementary actions. 
The rules $\pincpi$, and $\pincpe$ allow processes, one aside the other, to communicate, even when both are inside a logic restriction. This is a consequence of the 
logical nature of \OpNameRen, which binds names, and co-names, up to their renaming, indeed.
The rule $\pinccntxp$ leaves processes, one aside the other, to evolve
independently. Finally, $\pincrefl$ makes the relation reflexive.
%%%%
\begin{example}[\textbf{\textit{Using the \lts}}]
\label{example:Using the labeled transition system}
As a first example, we rewrite
$\pincNu{\pinca}
	{(\pincPar{(\pincSec{\pinca}
	                   {\pincSec{\pincb}
                                    {\pincE})}
	          }
                  {\pincSec{\pincna}
                           {\pincF}
                  })
        }$
to $\pincNu{\pinca}{(\pincPar{\pincE}{\pincF})}$, observing the action $\pincb$, as
follows:
%%%%%%%%%%%
\par\vspace{\baselineskip}\noindent
{\scriptsize
  \fbox{
    \begin{minipage}{.974\linewidth}
    \vspace{-.3cm}
      \begin{equation}
      \label{equation:Using the labeled transition system-01}
      \input{./PPi-example-LTS-01}
      \end{equation}
    \end{minipage}
  }%fbox
}%\scriptsize
\vspace{\baselineskip}\par
%%%%%%%%%%%%%%%
As a second example, we show that the labeled transition system
\eqref{equation:PPi-LTS-from-BVT} allows some interaction which originates from the logical nature of \OpNameRen. In $ \CCSR $ we model that
$\pincPar{\pincNu{\pinca}
	         {(\pincSec{\pinca}{\pincSec{\pincb}
	                                    {\pincE}}
                 )}
	 }
         {\pincNu{\pinca}
                 {(\pincSec{\pincna}
                           {\pincF })}}$
reduces to $\pincNu{\pinca}{(\pincPar{\pincE}{\pincF})}$, observing $\pincb$,
unlike in Milner $\CCS$:
%%%%%%%%%%%%%
\par\vspace{\baselineskip}\noindent
{\small
  \fbox{
    \begin{minipage}{.974\linewidth}
    \vspace{-.4cm}
      \begin{equation}
      \label{equation:Using the labeled transition system-02}
      \input{./PPi-example-LTS-02}
      \end{equation}
    \end{minipage}
  }%fbox
}%\small
\vspace{\baselineskip}\par\noindent
%%%%%%%%%%%%%
\end{example}
%%%%%%%%%%%%%
\paragraph{\Simpleprocess es.} They are the last notion we introduce in this section. 
They are useful for technical reasons which Section~\ref{section:Soundness of BVTCC} will make apparent. A process $\pincE$ is a \textit{\simpleprocess} whenever it satifies two constraints. First, $ \pincE $ must belong to the language of~\eqref{equation:PPi-simple-processes}:
%%%%%%%%%%%%%
\par\vspace{\baselineskip}\noindent
{\small
  \fbox{
    \begin{minipage}{.974\linewidth}
     \vspace{-.1cm}
       \begin{equation}
	 \label{equation:PPi-simple-processes}
	 \input{./PPi-simple-processes}
       \end{equation}
%      \vspace{-.7cm}
    \end{minipage}
  }%fbox
}%\small
\par\vspace{\baselineskip}\noindent
%%%%%%%%%%%%
Second, if $\pincLabL_1,\ldots,\pincLabL_n$ are all, and only, the action prefixes that occur in $ \pincE $, then $ i\neq j $ implies $\pincLabL_i\neq\vlne{\pincLabL_j}$, for every 
$i,j\in\Set{1,\ldots,n}$.
%%%
\begin{example}[\textbf{\textit{\Simpleprocess es}}]
\label{example:Simple processes}
Some are in the following table.
%%%%%%%%%%%%%
\par\vspace{\baselineskip}\noindent
{\small
\fbox{
    \begin{minipage}{.974\linewidth}
%       \begin{equation}
	      \input{SBV2-example-simple-processes}
%       \end{equation}
    \end{minipage}
  }%fbox
}%\small
\par\vspace{\baselineskip}\noindent
Both the second, and the third process are \simple\ because they belong to~\eqref{equation:PPi-simple-processes}, and $\pinca,\pincb,\pincnc$ is the list of their pairwise distinct action prefixes.
\end{example}
%%%%%%%%%%%%%
\begin{remark}[\textbf{\textit{Aim, and nature of \simpleprocess es}}]
In coming Section~\ref{section:How computing in CCSR by means of BVTCC} we shall intuitively show that \simpleprocess es play the role of results of computations when we use derivations of $\BVT$ to compute what the \lts\ in~\eqref{equation:PPi-LTS-from-BVT} can, in fact, compute by itself.
\end{remark}

%% file: PPi-LTS-labels.tex
\begin{aligned}
% \nonumber
\pincAct &
	\grammareq 	\pincLabT
	\ \ \mid\ \ 	\pincLabL
	\ \ \mid\ \ 	\vlne\pincAct
	\ \ \mid\ \  	\pincAct\,;\pincAct
\end{aligned}

%% file: PPi-labels-congruence.tex
\begin{aligned}
 \vlne\pincPreT & \sim \pincPreT
 &\qquad\qquad
 \vlne\pincna & \sim \pinca
 &\qquad\qquad
 \vlne{\pincAct\,;\pincAct'} & \sim
   \vlne{\pincAct}\,;\vlne{\pincAct'}
 &\qquad\qquad
 \pincPreT\,;\pincAct & \sim \pincAct
\end{aligned}

%% file: PPi-syntax.tex
\begin{aligned}
% \nonumber
\pincE &
	\grammareq 	\pincZer
%	\ \ \mid\ \ 	\pincSec{\pincLabT}
%		                    {\pincE}
	\ \ \mid\ \ 	\pincSec{\pincLabL}
		                    {\pincE}
	\ \ \mid\ \  	(\pincPar{\pincE}{\pincE})
	\ \ \mid\ \ 	\pincNu{\pinca}{\pincE}
\end{aligned}

%% file: PPi-structural-congruence.tex
\begin{aligned}
 \vlne\pincna & \sim \pinca
% \\
% \phantom{ \vlne\pincna} & \phantom{\sim \pinca}
% & \quad
% \pincSec{\pincLabL}{\pincSec{\pincZer}{\pincZer}} & \sim
%   \pincSec{\pincLabL}{\pincZer}
 & \quad
 \pincPar{\pincE}{\pincZer} & \sim \pincE
 & \quad
 \pincPar{\pincE}{\pincF} & \sim \pincPar{\pincF}{\pincE}
 & \quad
 \pincPar{\pincE}{(\pincPar{\pincF}{\pincG})}
 & \sim
 \pincPar{(\pincPar{\pincE}{\pincF})}{\pincG}
%  & \quad
%  \pincSec{\pincZer}{\pincE} & \sim \pincE
\\
\phantom{ \vlne\pincna}
%& \phantom{\sim \pinca}
%\phantom{\pincSec{\pincLabL}{\pincSec{\pincZer}{\pincZer}}} &
%\phantom{\sim \pincSec{\pincLabL}{\pincZer}}
& \phantom{}
   &\quad
%  \pincNu{\pincb}{\pincZer} & \sim \pincZer
%    &\quad
 \pincNu{\pinca}{\pincNu{\pincb}{\pincE}}
 & \sim
 \pincNu{\pincb}{\pincNu{\pinca}{\pincE}}
   &\quad
 \pincNu{\pincb}{(\pincE\subst{\pinca}{\pincb})}
 & \sim
 \pincNu{\pinca}{\pincE}
   &\quad
 \pincNu{\pinca}{\pincE}
 & \sim
 \pincE
 \textrm{ if } \pinca\not\in\pincFN{\pincE}
\end{aligned}

%% file: PPi-LTS-from-BVT.tex
\begin{gathered}
% 1st
%\vlinf{\pincact}{}
%      {\pincLTSJud{\pincSec{\pincalpha}
%                           {\pincE}
%		          }
%				  {\pincE}
%				  {\pincLabI{\vlne\pincalpha}{\vlne\pincalpha}}
%      }
%      {
%        \pincalpha\in\Set{\pincLabT,\pincLabL}
%      }
\vlinf{\pincact}{}
      {\pincLTSJud{\pincSec{\pincLabL}
                           {\pincE}
		          }
				  {\pincE}
				  {\pincLabL}
      }
      {
      }
\qquad\qquad
\vliinf{\pinccom}{}
       {\pincLTSJud{\pincPar{\pincE}
                            {\pincF}}
		   {\pincPar{\pincE'}
                            {\pincF'}}
		   {\pincLabT}
       }
       {%1
        \pincLTSJud{\pincE}
                   {\pincE'}
                   {\pincLabL}
       }
       {%2
        \pincLTSJud{\pincF}
                   {\pincF'}
                   {\vlne\pincLabL}
       }
\\
% 3rd line
\vlinf{\pincpi}
      {(\pincalpha\in\Set{\pincb,\pincnb})}
      {\pincLTSJud{\pincPar{\pincNu{\pincb}{\pincE}}
                           {\pincNu{\pincb}{\pincF}}}
		  {\pincPar{\pincNu{\pincb}{\pincE'}}
		           {\pincNu{\pincb}{\pincF'}}}
		  {\pincPreT}}
      {\pincLTSJud{\pincPar{\pincE}{\pincF}}
		  {\pincPar{\pincE'}{\pincF'}}
		  {\pincalpha}}
\qquad
\vlinf{\pincpe}
      {(\pincalpha\not\in\Set{\pincb,\pincnb})}
      {\pincLTSJud{\pincPar{\pincNu{\pincb}{\pincE}}
                           {\pincNu{\pincb}{\pincF}}}
		  {\pincPar{\pincNu{\pincb}{\pincE'}}
		           {\pincNu{\pincb}{\pincF'}}}
		  {\pincalpha}}
      {\pincLTSJud{\pincPar{\pincE}{\pincF}}
		  {\pincPar{\pincE'}{\pincF'}}
		  {\pincalpha}}
\\
% 5th line
\vlinf{\pincrefl}
      {}
      {\pincLTSJud{\pincE}
		  {\pincE}
		  {\pincPreT}}
      {}
\qquad
\vlinf{\pinccntxp}{}
      {\pincLTSJud{\pincPar{\pincE}{\pincG}}
		  {\pincPar{\pincF}{\pincG}}
		  {\pincalpha}
      }
      {\pincLTSJud{\pincE}
                  {\pincF}
                  {\pincalpha}}
\qquad
  \vlinf{\pinctran}{}
    {\pincLTSJud{\pincE}
		{\pincG}
		{\pincalpha;\,\pincbeta}}
    {\pincLTSJud{\pincF}
		{\pincF'}
		{\pincalpha}
    \qquad
    \pincLTSJud{\pincF'}
		{\pincG'}
		{\pincbeta}}
\end{gathered}

%% file: PPi-example-LTS-01.tex
\vlderivation{
 \vliin{\pinctran}{}
	   {\pincLTSJud{\pincNu{\pinca}
	                       {(\pincPar{(\pincSec{\pinca}
	                                           {\pincSec{\pincb}
	                                                    {\pincE}
	                                           })
	                                 }
                                         {\pincSec{\pincna}
                                                  {\pincF}
                                         })
                               }
                   }
		           {\pincNu{\pinca}{(\pincPar{\pincE}{\pincF})}
			       }
		       {\pincPreT;\pincb\,\pincCong\,\pincb}
       }
       {%%A
        \vlin{\pincpe}{\pincLabT\not\equiv\pinca}
	         {\pincLTSJudShort{\pincNu{\pinca}
		                              {(\pincPar{(\pincSec{\pinca}
						                                  {\pincSec{\pincb}
						                                           {\pincE}})
					                            }
					                            {\pincSec{\pincna}
					                                     {\pincF}})
		                              } 
					           \pincCong
					           \pincPar{\pincNu{\pinca}
					             	           {(\pincPar{(\pincSec{\pinca}
					             					               {\pincSec{\pincb}
					             					                        {\pincE}})
					             				         }
					             				         {\pincSec{\pincna}
					             				                  {\pincF}})
					                           }
					                   }
					                   {\pincNu{\pinca}{\pincZer}}
	                           }
			                   {\pincPar{\pincNu{\pinca}
			                                    {(\pincPar{\pincSec{\pincb}{\pincE}}
			                                     		  {\pincF})}
			                            }
			                            {\pincNu{\pinca}{\pincZer}}
			                   }
			                   {\pincPreT}
%		                       {r}{&}
             } {
	            \vliin{\pinccom}{}
    		          {\pincLTSJudShort{\pincPar{\pincPar{(\pincSec{\pinca}
    		          		                                       {\pincSec{\pincb}
    		          		                                                {\pincE}})}
    		          				                     {\pincSec{\pincna}
    		          				                              {\pincF}}}
    		                                    {\pincZer}
    		                             \pincCong
    		                             \pincPar{(\pincSec{\pinca}
		                                                   {\pincSec{\pincb}
		                                                            {\pincE}})}
				                                 {\pincSec{\pincna}
				                                          {\pincF}}
			                           }
				                       {\pincPar{(\pincSec{\pincb}{\pincE})}{\pincF}
				                        \pincCong
				                        \pincPar{\pincPar{(\pincSec{\pincb}{\pincE})}{\pincF}}
				                                {\pincZer}
				                       }
			                           {\pincPreT}
  		              } {%%%A.1-------
		                 \vlin{\pincact}{}
		                      {\pincLTSJudShort{\pincSec{\pinca}
		                                                {\pincSec{\pincb}{\pincE}}
			                                   }
			                                   {\pincSec{\pincb}{\pincE}}
		                                       {\pinca}
		                      } {
		                 \vlhy{}}
		                }
		                {%%%A.2-------
		                 \vlin{\pincact}{}
		                      {\pincLTSJudShort{\pincSec{\pincna}{\pincF}}{\pincF}{\pincna}
		                      }{
		                 \vlhy{}}
		                }
              }
   }%%A-stop
   {%%%B
    \vlin{\pincpe}{\pincb\not\equiv\pinca}
   		 {\pincLTSJudShort{\pincPar{\pincNu{\pinca}
		                                   {(\pincPar{\pincSec{\pincb}{\pincE}}
		                                             {\pincF})}
		                           }
		                           {\pincNu{\pinca}{\pincZer}}
		                  }
   			              {\pincPar{\pincNu{\pinca}{(\pincPar{\pincE}{\pincF})}}
   			                       {\pincNu{\pinca}{\pincZer}}
   			               \pincCong
   			               \pincNu{\pinca}{(\pincPar{\pincE}{\pincF})}
   			              }
   			              {\pincb}
   	      }   {
    \vlin{\pinccntxp}{}
   		 {\pincLTSJudShort{\pincPar{\pincPar{\pincSec{\pincb}{\pincE}}{\pincF}}
   		                           {\pincZer}
   		                   \pincCong
   		                   \pincPar{\pincSec{\pincb}{\pincE}}{\pincF}
   		                  }
   			              {\pincPar{\pincE}{\pincF}
   			               \pincCong
   			               \pincPar{\pincPar{\pincE}{\pincF}}
   			                       {\pincZer}}
   			              {\pincb}
   	      }   {
    \vlin{\pincact}{}
   		 {\pincLTSJudShort{\pincSec{\pincb}{\pincE}}
   			              {\pincE}
   			              {\pincb}
   	      }   {
    \vlhy{}   }}}
   }%%%B-stop
}%%%vlderivation---stop

%% file: PPi-example-LTS-02.tex
\vlderivation{
 \vliin{\pinctran}{}
	   {\pincLTSJud{\pincPar{\pincNu{\pinca}
		   	                        {(\pincSec{\pinca}
		                                      {\pincSec{\pincb}
		                                               {\pincE}
		                                      })}
	                        }
                            {\pincNu{\pinca}
                                    {(\pincSec{\pincna}
                                              {\pincF})}
                            }
                   }
		           {\pincNu{\pinca}{(\pincPar{\pincE}{\pincF})}
			       }
		       {\pincPreT;\pincb\,\pincCong\,\pincb}
       }
       {%%A
        \vlin{\pincpe}{\pincLabT\not\equiv\pinca}
	         {\pincLTSJudShort{\pincPar{\pincNu{\pinca}
		         	   	                       {(\pincSec{\pinca}
		         	                                     {\pincSec{\pincb}
		         	                                              {\pincE}
		         	                                     })}
	       	                           }
	                                   {\pincNu{\pinca}
	                                           {(\pincSec{\pincna}
	                                                     {\pincF})}
	                                   }
		                      }
			                  {\pincPar{\pincNu{\pinca}{(\pincSec{\pincb}{\pincE})}}
			                           {\pincNu{\pinca}{\pincF}}
			                  }
			                  {\pincPreT}
             } {
	            \vliin{\pinccom}{}
    		          {\pincLTSJudShort{\pincPar{(\pincSec{\pinca}
			                                              {\pincSec{\pincb}{\pincE}})}
					                            {\pincSec{\pincna}{\pincF}}
				                       }
				                       {\pincPar{(\pincSec{\pincb}{\pincE})}{\pincF}
				                       }
			                           {\pincPreT}
  		              } {%%%A.1-------
		                 \vlin{\pincact}{}
		                      {\pincLTSJud{\pincSec{\pinca}
		                                           {\pincSec{\pincb}{\pincE}}
			                              }
			                              {\pincSec{\pincb}{\pincE}}
		                                  {\pinca}
		                      } {
		                 \vlhy{}}
		                }
		                {%%%A.2-------
		                 \vlin{\pincact}{}
		                      {\pincLTSJud{\pincSec{\pincna}{\pincF}}{\pincF}{\pincna}
		                      }{
		                 \vlhy{}}
		                }
              }
   }%%A-stop
   {%%%B
    \vlin{\pincpe}{\pincnb\not\equiv\pinca}
   		 {\pincLTSJudShort{\pincPar{\pincNu{\pinca}{(\pincSec{\pincb}{\pincE})}}
	   		 		               {\pincNu{\pinca}{\pincF}}
	   		 		      }
	   			          {\pincPar{\pincNu{\pinca}{(\pincPar{\pincE}{\pincF})}}
	   			                   {\pincNu{\pinca}{\pincZer}}
	   			           \pincCong
	   			           \pincNu{\pinca}{(\pincPar{\pincE}{\pincF})}
	   			          }
	   			          {\pincb}
   	      }   
   	      {
    \vlin{\pinccntxp}{}
   		 {\pincLTSJudShort{\pincPar{\pincSec{\pincb}{\pincE}}{\pincF}
	   		              }
	   			          {\pincPar{\pincE}{\pincF}
	   			           \pincCong
	   			           \pincPar{\pincPar{\pincE}{\pincF}}
	   			                   {\pincZer}}
	   			          {\pincb}
   	      }   {
    \vlin{\pincact}{}
   		 {\pincLTSJudShort{\pincSec{\pincb}{\pincE}}
	   			          {\pincE}
	   			          {\pincb}
   	      }   {
    \vlhy{}   }}}
   }%%%B-stop
}%%%vlderivation---stop

%% file: PPi-simple-processes.tex
\begin{aligned}
\pincE & \grammareq \pincZer
	  \ \ \mid \ \ \pincSec{\pincLabL}{\pincZer}
	  \ \ \mid \ \ \pincPar{\pincE}{\pincE}
	  \ \ \mid \ \ \pincNu{\pinca}{\pincE}
\end{aligned}

%% file: SBV2-example-simple-processes.tex
\begin{center}
{%%%%%%Start scope of \doublespacing
\doublespacing
\begin{tabular}{c}
%\textbf{\Simpleprocess es}
%\\\hline
$\pincPar{(\pincSec{\pinca}{\pincZer})}
         {(\pincSec{\pincnb}{\pincZer})}$
\\\hline
$    \pincPar{(\pincSec{\pinca}{\pincZer})
             }
             {\pincPar{\pincNu{\pincb}
			      {(\pincPar{\pincNu{\pincd}
                                                {(\pincPar{(\pincSec{\pinca}
                                                                    {\pincZer})}

							  {(\pincSec{\pincnc}
								    {\pincZer})})
                                                }
					}
					{(\pincSec{\pincb}{\pincZer})})}
                      }
                      {(\pincSec{\pinca}{\pincZer})}
             }
$
\\\hline
$
\pincPar{\pincNu{\pincb}
		{(\pincPar{\pincNu{\pincc}
				  {(\pincPar{(\pincSec{\pinca}
						      {\pincZer})}

					    {(\pincSec{\pincnc}
						      {\pincZer})})
				  }
			  }
			  {(\pincSec{\pincb}{\pincZer})})}
	}
	{(\pincSec{\pinca}{\pincZer})}
$
\end{tabular}
}%%%%%End scope of \doublespacing
\end{center}

%% file: PPi-reduced-to-BVTCC.tex
\section{How computing in $\CCSR$ by means of $\BVT$}
\label{section:How computing  in CCSR by means of BVTCC}
Given $\BVT$, and $\CCSR$ we illustrate how transforming questions about the existence of computations of $ \CCSR $ into questions about proof-search inside the standard fragment $\BVTL$ of $\BVT$.
Let $\pincE$, and $\pincF$, be two processes of $\CCSR$, with $ \pincF$ simple.
Let us assume we want to check
$\pincLTSJud
   {\pincE}
   {\pincF}
   {\pincLabL_1;\cdots;\pincLabL_n}$.
Next we highlight the main steps to answer such a question by answering a question 
about proof-search inside $\BVT$, without resuming to computations in the \lts\ of $\CCSR$.
\par
To that purpose, this section has two parts. 
The first one formalizes the notions that makes the link between processes of $ \CCSR $, and structures of $ \BVT $ precise. The second part, \ie Subsection~\ref{subsection:Reducing the labeled transitions to proof-search steps}, delineates the steps to transform one question into the other, eventually justifying also the need to prove the Soundness of $\BVTL$ --- not 
$\BVT$ --- \wrt\ $\CCSR$, in Section~\ref{section:Soundness of BVTCC}.
%%%%%%%%%
\subsection{Connecting $\CCSR$, and $ \BVT $}
\label{subsubsection:Formally connecting CCSR, and BVTCC}
%%%%%%
\paragraph{\Processstructure s.}
They belong to the language of the grammar~\eqref{align:BV2-process-structures} here below, and, clearly, they are \OpNameTen-free:
%%%%%%%%%%%%%
\par\vspace{\baselineskip}\noindent
{\small
  \fbox{
    \begin{minipage}{.974\linewidth}
      \begin{equation}
	\label{align:BV2-process-structures}
         \input{./BV2-process-structures}
      \end{equation}
    \end{minipage}
  }%fbox
}%\small
\par\vspace{\baselineskip}\noindent
%%%%%%%%%%%%
Like at page~\ref{fig:BVT-structures}, we range over variable names of \processstructure s by
$\atmLabL, \atmLabM$, and $ \atmLabN $.
%%%%
\begin{fact}[\textbf{\textit{Processes correspond to \processstructure s}}]
\label{fact:From a process to a process structure}
Processes, and \processstructure s isomorphically correspond thanks to the following isomorphism, so extending the correspondence in
\cite{Brus:02:A-Purely:wd} among $\CCS$ terms, and $\BV$ structures.
%%%%%%%%%%%%%
\par\vspace{\baselineskip}\noindent
{\small
  \fbox{
    \begin{minipage}{.974\linewidth}
       \begin{equation}
	 \label{equation:SBV2-to-PPi-map-process}
	\input{./SBV2-to-PPi-map-process}
       \end{equation}
    \end{minipage}
  }%fbox
}%\small
\par\vspace{\baselineskip}\noindent
%%%%%%%%%%%%%
\end{fact}
%%%%%%%%%%
\paragraph{\Environmentstructure s.}
\label{paragraph:Environmentstructure s}
Let us recall Example~\eqref{example:Modeling external communication inside BVT}. It shows that representing an external communication as a derivation of $\BVT$ requires to assign a specific meaning to the structures in the conclusion of the derivation. One structure represents a process. The other one encodes the labels that model the sequence of messages between the process, and an environment. So,  we need to identify the \dfn{\environmentstructure s}, namely the set of structures that can fairly represent the sequence of messages. By definition, we say that every \dfn{\environmentstructure}\ is a \emph{\canonical} structure 
(page~\pageref{paragraph:Structures in canonical form}) that the following grammar \eqref{equation:SBV2-environment-structures} generates:
%%%%%%%%%%%%%
\par\vspace{\baselineskip}\noindent
  \fbox{
    \begin{minipage}{.974\linewidth}
      {\small
        \begin{equation}
	    \label{equation:SBV2-environment-structures}
	    \input{./SBV2-environment-structures}
	\end{equation}
      }%\small
    \end{minipage}
  }%fbox
\par\vspace{\baselineskip}\noindent
%%%%%%%%%%%%%
If different from $\vlone$, we have to think of every \environmentstructure\ as a
list, possibly in the scope of some instance of \OpNameRen, that we can consume
from its leftmost component, onward.
%%%%%%%%%%%%%
\begin{example}[\textit{\textbf{\Environmentstructure s}}]
\label{example:Environment structures}
Let 
$\natma,\atma_1,\natma_1,\atmb_1,\atmb_2 \not\approx \vlone$.
%%%%%%%%%%%%%
\par\vspace{\baselineskip}\noindent
{\small
\fbox{
    \begin{minipage}{.974\linewidth}
    \vspace{-.4cm}
	    \input{./SBV2-example-environment-structures}
    \vspace{-.5cm}
    \end{minipage}
  }%fbox
}%\small
\par\vspace{\baselineskip}\noindent
%%%%%%%%%%%%%
\eqref{eqnarray:example-environment-structure-01}
is not an \environmentstructure\ because $\atmb_4$ does not occur in the structure.
\eqref{eqnarray:example-environment-structure-02}
is not an \environmentstructure\ because $\vlone$ occurs in it.
\end{example}
%%%%%%%%%%%%
\begin{fact}[\textbf{\textit{\Environmentstructure s map to sequences of
actions}}]
\label{fact:From an environment structure to a set of actions}
The map~\eqref{equation:SBV2-to-PPi-map-actions} takes both an 
\environmentstructure, and a set of
atoms as arguments. The map transforms a given \environmentstructure\ to a sequence of actions
that may work as a label of transitions in~\eqref{equation:PPi-LTS-from-BVT}.
%%%%%%%%%%%%%
\par\vspace{\baselineskip}\noindent
{\small
  \fbox{
    \begin{minipage}{.974\linewidth}
     \vspace{-.5cm}
       \begin{equation}
	     \label{equation:SBV2-to-PPi-map-actions}
  	     \input{./SBV2-to-PPi-map-actions}
       \end{equation}
     \vspace{-.7cm}
    \end{minipage}
  }%fbox
}%\small
\par\vspace{\baselineskip}\noindent
%%%%%%%%%%%%%%%%%%
Given an \environmentstructure, the map yields the corresponding sequence, if its second argument is  $\emptyset$.
\end{fact}
%%%%%%%%%%%%%
\begin{example}[\textbf{\textit{From an \environmentstructure\ to actions}}]
\label{example:From an environment structure to a set of actions}
Both $\atmb_1$, and $\atmb_2$ are internal actions of
$\vlstore{\vlsbr<\atma_1
                ;\vlfo{\atmb_2}
                      {<\natma_1
                       ;\vlfo{\atmb_1}{<\atmb_2;\atmb_1>}>}>}
\mapDiToPinc{\vlread}{\emptyset} = 
\atma_1;\natma_1; \pincPreT;\pincPreT
\pincCong
\atma_1 ;\natma_1$ in~\eqref{eqnarray:example-environment-structure-00}. Intuitively, if a variable name $\pincLabL$ that occurs in a structure $\pincE$ belongs to $X$ in $\mapDiToPinc{\pincE}{X}$, then $\pincLabL$ gets mapped to $\pincLabT$. The reason why $\pincLabL$ is in $X$ is that $ \pincLabL $ is not a free name of $ \pincE $.
\end{example}
%%%%%%%%%%%%%5
\paragraph{\Trivialderivation s.}
By definition, a derivation $\bvtDder$ of $\BVT$ is \dfn{\trivial} if (i) $\bvtDder$ only operates on \OpNameTen-free structures, and (ii) $\bvtDder$ does not contain any
occurrence of $\bvtatidrulein$. All the others are \dfn{\nontrivialderivation s}.
%%%%%%%%%
\begin{example}[\textbf{\textit{A \trivialderivation}}]
\label{example:A trivial derivation}
It is in~\eqref{equation:A-trivial-derivation} here below.
%%%%%%%%%%%%%%%%%%%%%%
\par\vspace{\baselineskip}\noindent
{\small
  \fbox{
    \begin{minipage}{.974\linewidth}
      \begin{equation}
       \label{equation:A-trivial-derivation}
      	\input{./SBV2-trivial-derivation}
       \end{equation}
    \end{minipage}
  }%fbox
}%\small
\par\vspace{\baselineskip}\noindent
%%%%%%%%%%%%%%%%%%%%%%
Being \trivial\ does not mean without rules.
``\Trivial'' identifies a derivation where no communication, represented by
instances of $\bvtatidrulein$, occur.
\end{example}
%%%%%%%%%%%%%%
\begin{fact}[\textit{\textbf{\Trivialderivation s on \processstructure s are
quite simple}}]
\label{fact:Trivial derivations preserve process structures}
Let $\strR$, and $\strT$ be \processstructure s, and
$\bvtInfer{\bvtDder}
          {\, \strT
          \bvtJudGen{\BVsub}{}
          \strR}$ 
be \trivial.
Then $\BVsub =\Set{\bvtseqdrulein,\bvtrdrulein} $, and all the instances of
$\bvtseqdrulein$ in $ \bvtDder $ have form
$\vlderivation{
 \vliq{\bvtseqdrule}{}
      {\vlsbr[<\pincLabL;\strR'>;\strR'']}{
 \vlhy{\vlsbr<\pincLabL;[\strR';\strR'']>}}
}
$, or
$\vlderivation{
 \vliq{\bvtseqdrule}{}
      {\vlsbr[\strR';\strR'']}{
 \vlhy{\vlsbr<\strR';\strR''>}}
}$, for some
$\strR', \strR''$, and $\pincLabL\not\approx\vlone$.
\end{fact}
%%%%%%%
\begin{proof}
By definition, no $\bvtatidrulein$ can exist in $\bvtDder$.
Let us assume an instance 
$ \vlinf{\bvtswirule }{}
        {\vlsbr[(\strR;\strT);\strU]}
        {\vlsbr([\strR;\strU];\strT)} $
exists in $\bvtDder$. Since $ \bvtDder $ is \OpNameTen-free, it must be $ \strT \approx \vlone$ and we can eliminate such an $ \bvtswirulein $.
Let us assume one instance of $\bvtseqdrulein$ exists in $\bvtDder$.
In general it would be
$\vlinf{\bvtseqdrule}{(*)}
       {\vlsbr[<\pincLabL;\strR'>;<\pincLabM;\strR''>]}
       {\vlsbr<[\pincLabL;\pincLabM];[\strR';\strR'']>}$, for some
$\pincLabL\,,\pincLabM, \strR'$, and $\strR''$.
So, let us assume such a $(*)$ occurs in $\bvtDder$ with
$\pincLabL\,, \pincLabM\not\approx\vlone$. In absence of $\bvtatidrulein$, even
though we might have $\pincLabL\approx\vlne\pincLabM$, the structure $\vlsbr[\pincLabL;\pincLabM]$ could not disappear from $\bvtDder$, namely from $\strT$.
Consequently, $\strT$ could not be a process structure, against assumption.
\end{proof}
%%%%%%
\paragraph{\Simplestructure s.}
This notion strengthens the idea that ``trivial'' stands for ``no interactions''. A structure $\strR$ is a \emph{\simplestructure} if it satisfies two constraints.
First, it must belong to the language of~\eqref{equation:SBV2-normal-structures}.
%%%%%%%%%%%%%
\par\vspace{\baselineskip}\noindent
{\small
  \fbox{
    \begin{minipage}{.974\linewidth}
     \vspace{-.1cm}
       \begin{equation}
	 \label{equation:SBV2-normal-structures}
	 \input{./SBV2-normal-structures}
       \end{equation}
%      \vspace{-.7cm}
    \end{minipage}
  }%fbox
}%\small
\par\vspace{\baselineskip}\noindent
%%%%%%%
Second, if $ \pincLabL_1,\ldots,\pincLabL_n $ are all, and only, the variable names that occur in $ \strR $, then $ i\neq j $ implies $ \pincLabL_i\neq\pincLabL_j $, for every $ i,j\in\Set{1,\ldots, n} $.
%%%%%%%%%%%%%
\begin{fact}[\textbf{\textit{Basic properties of \simplestructure s}}]
\label{fact:Basic properties of simplestructure}
\begin{itemize}
\item 
Trivially, by definition, \simplestructure s are \coinvertible, because every of them is the negation of an \invertiblestructure\ (Proposition~\ref{proposition:Invertible structures are invertible}.)
\item 
\Simplestructure s are the logical counterpart of \simpleprocess es, thanks to the isomorphism~\eqref{equation:SBV2-to-PPi-map-process}.
\end{itemize}
\end{fact}
%%%
\begin{example}[\textbf{\textit{\Simplestructure s}}]
\label{example:Normal structures stand for simple processes}
The following table shows some instances of \simplestructure s which correspond to the \simpleprocess es\ in Example~\eqref{example:Simple processes}.
%%%%%%%%%%%%%
\par\vspace{\baselineskip}\noindent
{\small
\fbox{
    \begin{minipage}{.974\linewidth}
	      \input{SBV2-example-simple-structures}
    \end{minipage}
  }%fbox
}%\small
\par\vspace{\baselineskip}\noindent
Both the second, and the third structures are \simple\ because belong to~\eqref{equation:SBV2-normal-structures}, and $\pinca,\pincb,\pincnc$ is the list of their pairwise distinct variable names.
All the structures are co\invertible because negation of 
$\vlsbr(\natma;\atmb)$, and
$\vlsbr(\natma;\vlfo{\atmb}{(\vlfo{\atmd}{(\natma;\atmc)};\natmb)};\natma)$, and
$\vlsbr(\vlfo{\atmb}{(\vlfo{\atmc}{(\natma;\atmc)};\natmb)};\natma)$, respectively, which all are \invertible. \qed
\end{example}
%%%%%%%%%%%%
The following fact formalizes that \trivialderivation s operating on \simplestructure s only, represent computations where only instances of $\bvtrdrulein$ occur.
In Section~\ref{section:Soundness of BVTCC} this will allow to see that a \trivialderivation\ on \simplestructure s stands for a process that cannot communicate, neither internally, nor externally.
%%%%%%%%%%%%
\begin{fact}[\textbf{\textit{\Trivialderivation s on \simplestructure s contain almost no rules}}]
\label{fact:Trivial derivations preserve normal structures}
For any \simple\ $\strT$, if
$\bvtInfer{\bvtDder}
          {\, \strT
          \bvtJudGen{\SBVsub}{}
          \strR}$ is \trivial, then $\SBVsub=\Set{\bvtrdrulein}$, and $\strR$ is
\simple\ as well.
\end{fact}
%%%%%%%%%%%%%
\begin{proof}
Fact~\ref{fact:Trivial derivations preserve process structures} implies that the
derivation $\bvtDder$ only contains instances of $\bvtrdrulein$, and of very
specific instances of $\bvtseqdrulein$. Both kinds of rules neither
erase, nor introduce atoms, or new occurrences of \OpNameSeq in between $\strR$,
and $\strT$. Let us assume that $\bvtDder$ effectively contains an instance of
$\bvtseqdrulein$ with reduct
$\vlsbr<\atmLabL;\strR'>$, for some $\atmLabL$, and $\strR'$.
Then, the occurrence of \OpNameSeq would occur in $\strT$, as well,
making it not \simple, against our assumption. So, no occurrence of $\bvtseqdrulein$
exists in $\bvtDder$. This, of course, does not prevent the existence of
$\vlsbr<\atmLabL;\strR'>$ along $\bvtDder$, and, in particular, inside $\strR$. 
However, $\bvtrdrulein$ could not eliminate it, and an occurrence of \OpNameSeq would be inside $\strT$. In that case $\strT$ could not be \simple, against assumption. But if no occurrence of $\vlsbr<\atmLabL;\strR'>$ is inside $\bvtDder$, then our assumptions imply that 
$\strR$ is a \simplestructure. \qed
\end{proof}
%%%%%%%%%%%%%%

%% file: BV2-process-structures.tex
\strR  \grammareq \vlone
	  \ \ \mid \ \ \vlsbr<\pincLabL;\strR >
	  \ \ \mid \ \ \vlsqbrl\strR\vlpa \strR\vlsqbrr
	  \ \ \mid \ \ \vlfo{\pinca}{\strR}

%% file: SBV2-to-PPi-map-process.tex
\begin{gathered}
  \begin{minipage}{.35\textwidth}
      \vlstore{
% 	  \label{align:map-Zer}
	  \mapPincToDi{\pincZer}
	  &\mapsto       
	  \vlone
	  \\
% 	  \label{align:map-name}
	  \mapPincToDi{\pinca}
	  &\mapsto
	  \atma
	  \\
% 	  \label{align:map-coname}
	  \mapPincToDi{\pincna}
	  &\mapsto
	  \natma
	}%\vlstore
	  $\begin{aligned}
	    \vlread
	   \end{aligned}$
  \end{minipage}
  %%%% 2nd column
  \begin{minipage}{.45\textwidth}
      \vlstore{
          \mapPincToDi{\pincSec{\pincLabT}
          			           {\pincE}}
				  &\mapsto
			      \vlsbr<\vlone;\mapPincToDi{\pincE}>
				  \\
			% 	  \label{align:map-Pre-I}
			      \mapPincToDi{\pincSec{\pincLabL}
			      			           {\pincE}}
				  &\mapsto
				  \vlsbr<\mapPincToDi{\pincLabL}
				        ;\mapPincToDi{\pincE}>
				  \\
			% 	  \label{align:map-Par}
				  \mapPincToDi{\pincPar{\pincE}
				                       {\pincF}
				              }
				  &\mapsto
				  \vlsbr[\mapPincToDi{\pincE}
				        ;\mapPincToDi{\pincF}]
				  \\
			% 	  \label{align:map-Nu}
				  \mapPincToDi{\pincNu{\pinca}
				                      {\pincE}}
				  &\mapsto
 				  \vlfo{\atma}{\mapPincToDi{\pincE}}
				}%\vlstore
         $\begin{aligned}
	  \vlread
	 \end{aligned}$
    \end{minipage}
\end{gathered}

%% file: SBV2-environment-structures.tex
\begin{aligned}
%%%% With single action prefix
\strR & \grammareq
                      \vlone
        \ \ \mid  \ \ \atmLabL
        \ \ \mid  \ \ \vlsbr<\atmLabL;\strR>
        \ \ \mid  \ \ \vlstore{\vlsbr<\atmLabL;\strR>}
                      \vlfo{\atma}{\vlread}
%         \ \ \mid  \ \ \vlfo{\atma}{\strR}
\end{aligned}

%% file: SBV2-example-environment-structures.tex
\vlstore{
& \natma
   & \textrm{example}
\\
\label{eqnarray:example-environment-structure-00}
& \vlsbr<\atma_1
        ;\vlfo{\atmb_2}{<\natma_1
                        ;\vlfo{\atmb_1}
                              {<\atmb_2;\atmb_1>}>}>
   & \textrm{example}
\\
\label{eqnarray:example-environment-structure-01}
& \vlsbr<\atma_1
        ;\vlfo{\atmb_4}{<\natma_1
                        ;\vlfo{\atmb_1}
                              {<\atmb_2;\atmb_1>}>}>
   & \textrm{counterexample}
\\
\label{eqnarray:example-environment-structure-02}
& \vlsbr<\atma_1
        ;\vlfo{\atmb_2}
              {<\vlone
               ;\vlfo{\atmb_1}{<\atmb_2;\atmb_1>}>}>
   & \mbox{counterexample}
}
{\setlength{\arraycolsep}{2pt}
\begin{eqnarray}
 \vlread
\end{eqnarray}}

%% file: SBV2-to-PPi-map-actions.tex
%%%%%% From environment structures to actions
\begin{gathered}
    \begin{minipage}{.4\linewidth}
      \vlstore{
	}%\vlstore
	{\setlength{\arraycolsep}{2pt}
	\begin{eqnarray*}
% 	  \vlread
	  \mapDiToPinc{\vlone}{X}
	  &\mapsto&
	  \pincPreT
	  \\
	  \mapDiToPinc{\pincLabL}{X}
	  &\mapsto&
	  \pincPreT
	  \qquad (\pincLabL\in X)
	  \\
	  \mapDiToPinc{\pincLabL}{X}
	  &\mapsto&
	  \pincLabnL
	  \qquad (\pincLabL\not\in X)
	\end{eqnarray*}
	}
    \end{minipage}
    %%%% 2nd column
    \begin{minipage}{.5\linewidth}
      \vlstore{
          \mapDiToPinc{\vlfo{\atma}{\strR}}{X}
	  &\mapsto&
          \mapDiToPinc{\strR}{X\cup\Set{\atma,\natma}}
	  \\
          \mapDiToPinc{\vlsbr<\strP;\strR>}{X}
	  &\mapsto& %% continues after vlread here below
	}%\vlstore
	{\setlength{\arraycolsep}{2pt}
	\begin{eqnarray*}
	  \vlread
          \mapDiToPinc{\strP}{X};\mapDiToPinc{\strR}{X}
	\end{eqnarray*}
	}
    \end{minipage}
\end{gathered}

%% file: SBV2-trivial-derivation.tex
    \vlderivation{
      \vlin{\bvtseqdrule}{}
	  {\vlsbr
	    [\vlfo{\atma}
		  {[<\atmaRed;\atmbBlu;\strR>
		  ;<\natmaRed;\strT>]}
	    ;\natmbBlu]
	  }{
      \vliq{\eqref{align:unit-seq}}{}
	  {\vlsbr
	    [\vlfo{\atma}
		  {<[\atmaRed;\natmaRed]
		  ;[<\atmbBlu;\strR>;\strT]>}
	    ;\natmbBlu]
	  }{
      \vlin{\bvtseqdrule}{}
	  {\vlsbr
	    [\vlfo{\atma}
		  {<[\atmaRed;\natmaRed]
		  ;[<\atmbBlu;\strR>
		    ;<\vlone;\strT>]>}
	    ;\natmbBlu]
	  }{
      \vliq{\eqref{align:unit-pa}
	  ,\eqref{align:alpha-intro}}{}
	  {\vlsbr
	    [\vlfo{\atma}
		  {<[\atmaRed;\natmaRed]
		  ;<[\atmbBlu;\vlone]
		    ;[\strR;\strT]>>}
	    ;\natmbBlu]
	  }{
      \vlin{\bvtrdrule}{}
	  {\vlsbr
	    [\vlfo{\atma}
		  {<[\atmaRed;\natmaRed]
		  ;<\atmbBlu
		    ;[\strR;\strT]>>}
	    ;\vlfo{\atma}{\natmbBlu}]
	  }{
      \vliq{\eqref{align:unit-seq}}{}
	  {\vlfo{\atma}
	    {\vlsbr
	    [<[\atmaRed;\natmaRed]
	      ;<\atmbBlu
		;[\strR;\strT]>>
	    ;\natmbBlu]
	    }
	  }{
      \vlin{\bvtseqdrule}{}
	  {\vlfo{\atma}
	    {\vlsbr
	    [<[\atmaRed;\natmaRed]
	      ;<\atmbBlu
		;[\strR;\strT]>>
	    ;<\vlone;\natmbBlu>]
	    }
	  }{
      \vliq{\eqref{align:unit-pa}
	  ,\eqref{align:unit-seq}}{}
	  {\vlfo{\atma}
	    {\vlsbr
	    <[\atmaRed;\natmaRed;\vlone]
	      ;[<\atmbBlu;[\strR;\strT]>
	      ;\natmbBlu]>
	    }
	  }{
      \vlin{\bvtseqdrule}{}
	  {\vlfo{\atma}
	    {\vlsbr
	    <[\atmaRed;\natmaRed]
	      ;[<\atmbBlu;[\strR;\strT]>
	      ;<\natmbBlu;\vlone>]>
	    }
	  }{
      \vlhy{\vlfo{\atma}
	         {\vlsbr
	          <[\atmaRed;\natmaRed]
	           ;<[\atmbBlu;\natmbBlu]
	          ;[\strR;\strT;\vlone]>>
	         }
	   }}}}}}}}}}
      }

%% file: SBV2-normal-structures.tex
{%do not erase
\vlstore{
\strR & \grammareq \vlone
	  \ \ \mid \ \ \atmLabL
	  \ \ \mid \ \ \vlsbr[\strR;\strR]
	  \ \ \mid \ \ \vlfo{\pinca}{\strR}
	}%\vlstore
\begin{split}
 \vlread
\end{split}
}%do not erase

%% file: SBV2-example-simple-structures.tex
\begin{center}
%   \begin{minipage}{.3\linewidth}
    {%%%%%%Start scope of \doublespacing
     \doublespacing
     \begin{tabular}{c}
      \textbf{\Simplestructure s}
      \\\hline
      $\vlsbr[\atma;\natmb]$
      \\\hline
      $\vlsbr[\atma;\vlfo{\atmb}{[\vlfo{\atmd}{[\atma;\natmc]};\atmb]};\atma]$
      \\\hline
      $\vlsbr[\vlfo{\atmb}{[\vlfo{\atmc}{[\atma;\natmc]};\atmb]};\atma]$
     \end{tabular}
    }%%%%%End scope of \doublespacing
%   \end{minipage}
\end{center}

%% file: PPi-LTS-to-proof-search.tex
\subsection{Recasting labeled transitions to proof-search}
\label{subsection:Reducing the labeled transitions to proof-search steps}
Once connected $\BVT$, and $\CCSR$ as in the previous subsection, we get back to our initial reachability problem. Let us assume we want to check 
$\pincLTSJud
   {\pincE}
   {\pincF}
   {\pincLabL_1;\cdots;\pincLabL_n}$ in $ \CCSR $, where $ \pincF$ is a \simpleprocess.
The following steps recast the problem of $ \CCSR $ into a problem of searching inside $\BVT$:
\begin{enumerate}
\item 
\label{enumerate:how-soundness-works-01}
First we ``compile'' both $\pincE$, and $\pincF$ into \processstructure s $\mapPincToDi{\pincE}$, and $\mapPincToDi{\pincF}$, where
$\mapPincToDi{\pincF}$ is forcefully \simple. Then, we fix an $ \strR $ such that 
   $ \mapDiToPinc{\strR}{\emptyset} = \pincLabL_1;\cdots;\pincLabL_n$.

\item
\label{enumerate:how-soundness-works-02}
Second, it is sufficient to look for 
$\vlstore{\vlsbr[\mapPincToDi{\pincE}
                ;\vlne{\mapPincToDi{\pincF}}
                ;\strR]}
 \bvtInfer{\bvtPder}
          {\ \bvtJudGen{}{}
          \vlread}$ inside $ \BVT $ as the up-fragment of $ \SBVT $ is admissible for $ \BVT $ (Corollary~\ref{theorem:Admissibility of the up fragment} \cite{Roversi:unpub2012-I}.)
. 
\item 
\label{enumerate:how-soundness-works-03}
Finally, if $ \bvtPder $ of point~\eqref{enumerate:how-soundness-works-02} here above exists,
we can conclude
$\pincLTSJud
   {\pincE}
   {\pincF}
   {\pincLabL_1;\cdots;\pincLabL_n}$ in $ \CCSR $.
\end{enumerate}
%%%%%
Point~\ref{enumerate:how-soundness-works-03} rests on some simple observations.
The structure $\vlne{\mapPincToDi{\pincF}}$ is \invertible\ thanks to Fact~\ref{fact:Basic properties of simplestructure}.
So, it exists
$\vlstore{\vlsbr[\mapPincToDi{\pincE};\strR]}
 \bvtInfer{\bvtDder'}
          {\mapPincToDi{\pincF} \bvtJudGen{\BVT}{}
          \vlread}$ where both $ \mapPincToDi{\pincE} $, and $ \mapPincToDi{\pincF} $ are 
\OpNameTen-free because they are \processstructure s. The same holds for $\strR$ which is an \environmentstructure. Consequently, every instance of $\bvtswirulein$ in $\bvtDder$, if any, can only be
$
\vlderivation{
\vlin{\bvtswirule}{}
     {\vlsbr[(\strR;\vlone);\strU]}{
\vlhy{\vlsbr([\strR;\strU];\vlone)}}
}
$, and it can be erased.
This means that $ \bvtDder $ only contains rules that belong to 
$\Set{\bvtatidrulein,\bvtseqdrulein,\bvtrdrulein} $.
Standardization
(Theorem~\ref{theorem:Standardization in bvtatrdrulein...}), which applies to
$\Set{\bvtatrdrulein,\bvtatidrulein,\bvtseqdrulein,\bvtrdrulein}$,
implies we can transform $\bvtDder$ in $\BVT$ to a \standardderivation\ $\bvtEder$ of 
$\BVTL$.
The only missing step is in the coming section. It shows that proof-search in $ \BVTL $ is sound \wrt\ the computations of the \lts\ defined for $ \CCSR $.

%% file: BV2-soundness-wrt-PPi.tex
\section{Soundness of $\BVTL$ \wrt\ $ \CCSR $}
\label{section:Soundness of BVTCC}
The goal is proving Soundness whose formal statement is in Theorem~\eqref{equation:PPi-soundness-example-00} below.
We remark that our statement generalizes the one in \cite{Brus:02:A-Purely:wd}, and our  proof pinpoints many of the details missing in \cite{Brus:02:A-Purely:wd}.
\par
Soundness relies on the notions ``reduction of a \nontrivialderivation'', and ``\environmentstructure s that are consumed'', and needs some technical lemma.
%%%%%%%%%%%%%%
\paragraph{Reduction of \nontrivial, and \standardderivation s of $\BVTL$.}
Let $\strR$, and $\strT$ be \processstructure s.
Let $\bvtDder$ be a \nontrivial, and \standard\ derivation
${\small
%  \vlupsmash{
  \vlderivation{
  \vlde{\bvtDder'}{\BVTL}
       {\strR}{
  \vlin{\bvtatrdrule}{(*)}
       {\vlholer{\strS\vlsbr[\atma;\natma]}}{
  \vlde{\bvtDder''}{\BVTL}
       {\vlholer{\strS\vlsbr[\vlone]}}{
  \vlhy{\strT}}}}}
%   }
%\small
}$, where $(*)$ is the lowermost occurrence of $\bvtatrdrulein$ in $\bvtDder$.
The \dfn{reduction of $\bvtDder$} is the derivation $\bvtEder$ of rules of $\BVTL$
that we get from $\bvtDder$ by (i) replacing $\vlone$ for all occurrences of
$\atma$, and $\natma$ in $\bvtDder'$ that,
eventually, form the redex of $(*)$, and by (ii) eliminating all the fake
instances of rules that the previous step may have created.
%%%%%%%%%%%%%%%
\begin{fact}[\textit{\textbf{Reduction preserves \processstructure s}}]
\label{fact:Reduction preserve process structures}
Let $\strR$, and $\strT$ be \processstructure s. For every \nontrivial, and
\standard\ derivation
$\bvtInfer{\bvtDder}
          {\, \strT
          \bvtJudGen{\BVTL}{}
          \strR}$,
its reduction
$\bvtInfer{\bvtEder}
          {\, \strT'
          \bvtJudGen{\BVTL}{}
          \strR'}$
is such that both $\strR'$, and $\strT'$ are \processstructure s. Moreover,
$\bvtEder$ may not be \nontrivial, namely, no $\bvtatrdrulein$ may remain in $\bvtEder$. However, if $\bvtEder$ is \nontrivial, then it is \standard.
\end{fact}
%%%%%%%%%%%
\begin{proof}
The first statement follows from the definition of \processstructure s. If we erase
any sub-structure from a given \processstructure, we still get a 
\processstructure\ which, at least, is $\vlone$. Moreover, the lowermost instance of
$\bvtatrdrulein$ disappears, after a reduction. So, if it was the only one, none
remains. Finally, reduction does not alter the order of rules in $\bvtDder$.
\end{proof}
%%%%%%%%%%%%%%%%%%
% \newpage
\begin{fact}[\textit{\textbf{Preserving \rightcontext s}}]
\label{fact:Preserving rightcontext s}
Let $\bvtDder$ be a trivial derivation
$\vlstore{
  \vlsbr[\pincE;\strR]
 }
 \bvtInfer{\bvtDder}
          {\, \strS'\,\vlscn{\atma}
          \bvtJudGen{\Set{\bvtseqdrule,\bvtrdrule}}{}
          \strS\vlscn{\atma}}$, for some $\strS\vlhole, \strS'\vlhole$, and $ \atma $.
\begin{enumerate}
\item
\label{enumerate:Preserving rightcontext s-00}
If $\strS\vlscn{\atma}$ is not a \rightcontext, then
$\strS'\,\vlscn{\atma}$ cannot be a \rightcontext\ as well.

\item
\label{enumerate:Preserving rightcontext s-01}
If $\strS'\,\vlscn{\atma}$ is a \rightcontext, then
$\strS\,\vlscn{\atma}$ is a \rightcontext\ as well.
\end{enumerate}
\end{fact}
%%%%%%%%%%%%%%
\begin{proof}
\begin{enumerate}
\item
If $\strS\vlscn{\atma}$ is not a \rightcontext, then it has form
$\strS\vlscn{\atma}\approx\strS_0\,\vlsbr<\strR;\strS_1\,\vlscn{\atma}>$, with
$\strR\not\approx\vlone$, for some $\strS_0\vlhole$, and $\strS_1\vlhole$.
\OpNameSeq is non commutative. So, going upward in $\bvtDder$, 
there is no hope to transform
$\strS_0\,\vlsbr<\strR;\strS_1\,\vlscn{\atma}>$ into some
$\vlstore{\strS'_0\,\vlsbr<\vlholer{\strS'_1\,\vlscn{\atma}};\strR'>}
\vlholer{\vlread}$ where the occurrence of $ \atma $ in the first structure is the same occurrence as $ \atma $ in the second one.
Moreover,
$\vlinf{}{}{\vlsbr<\strR;\strT>}{\vlsbr[\strR;\strT]}$ is not derivable in $\Set{\bvtseqdrulein,\bvtrdrulein}\subset \BVT $. 
So, $\strS_0\vlsbr<\strR;\strS_1\,\vlscn{\atma}>$ cannot transform  into
some $\vlstore{\strS'_0\,\vlsbr[\strR';\vlholer{\strS'_1\,\vlscn{\atma}}]}
\vlholer{\vlread}$, going upward in $\bvtDder$.

\item
By contraposition of the previous point~\eqref{enumerate:Preserving
rightcontext s-00}.
\end{enumerate}
\end{proof}
%%%%%%%%%%%%
\begin{proposition}[\textbf{\textit{\Processstructure s, \trivialderivation s, and
\rightcontext s}}]
\label{proposition:Rightcontext s preserve communication}
\label{fact:Trivial derivations and rightcontext s}
Let $\strR$ be a \processstructure, and $\bvtDder$
be a \trivialderivation\
$\vlstore{
  \vlholer{\strS\vlsbr[\atmb;\natmb]}
 }
 \bvtInfer{\bvtDder}
          {\, \vlread
          \bvtJudGen{\Set{\bvtseqdrule,\bvtrdrule}}{}
          \strR}$, for some $\vlholer{\strS\vlhole}, \atmb$, and $\natmb$. Then:

\begin{enumerate}
\item
\label{enumerate:Trivial derivations and rightcontext s-10}
$\strR\not\approx\vlone$, and both $\atmb, \natmb$ occur in it.

\item
\label{enumerate:Trivial derivations and rightcontext s-00}
The structure $\strR$ is a
\rightcontext\ for both $\atmb$, and $\natmb$.
Namely, $\strR\approx\vlholer{\strS'\,\vlscn{\atmb }}$, and
$\strR\approx\vlholer{\strS''\,\vlscn{\natmb}}$
for some
$\vlholer{\strS' \vlhole}$, and
$\vlholer{\strS''\vlhole}$.

\item
\label{enumerate:Trivial derivations and rightcontext s-01}
$\strR\not\approx\vlsbr\strSc'<\pincalpha;\strSb' \,\vlscn{\atmb }>$, and
$\strR\not\approx\vlsbr\strSc''<\pincalpha;\strSb''\,\vlscn{\natmb}>$, for any
$\strSc'\vlhole, \strSc''\vlhole, \strSb'\vlhole$, and $\strSb''\vlhole$.

\item
\label{enumerate:Trivial derivations and rightcontext s-02}
$\vlstore{
 \vlsbr[\vlholer{\strS' \,\vlscn{\atmb}}
       ;\vlfo{\atmb}{\vlholer{\strS''\,\vlscn{\natmb}}}
       ;\strT]}
 \strR\not\approx\vlread$,
with $\atmb\in\strFN{\vlholer{\strS' \,\vlscn{\atmb}}}$, and
$\vlstore{
 \vlsbr[\vlfo{\atmb}{\vlholer{\strS' \,\vlscn{\atmb}}}
       ;\vlholer{\strS''\,\vlscn{\natmb}}
       ;\strT]}
 \strR\not\approx\vlread$,
with $\natmb\in\strFN{\vlholer{\strS'' \,\vlscn{\natmb}}}$, for any
$\vlholer{\strS'\vlhole}, \vlholer{\strS''\vlhole}$, and \processstructure\ $\strT$.

\item
\label{enumerate:Trivial derivations and rightcontext s-02'}
Let $\vec{\atma}$ be a, possibly empty, sequence of names.
Let $\strT$ be a \processstructure, possibly such that $\strT\approx\vlone$.
Then 
$\vlstore{
\vlsbr[\vlholer{\strS' \,\vlscn{\atmb }}
      ;\vlholer{\strS''\,\vlscn{\natmb}}
      ;\strT]
}
\strR\approx\vlfo{\vec{\atma}}{\vlread}$ such that
either (i)
$\atmb\in\strFN{\vlholer{\strS' \,\vlscn{\atmb }}}$, and
$\natmb\in\strFN{\vlholer{\strS''\,\vlscn{\natmb}}}$,
or (ii)
$\atmb\in\strBN{\vlholer{\strS' \,\vlscn{\atmb }}}$, and
$\natmb\in\strBN{\vlholer{\strS''\,\vlscn{\natmb}}}$.

\item
\label{enumerate:Trivial derivations and rightcontext s-03}
Let $\vlholer{\strS'\,\vlscn{\atmb}}$ be the one in Point~\eqref{enumerate:Trivial derivations and rightcontext s-02'} here above.
If $ \pincE $, and $ \pincF $ are processes \ST\
$\mapPincToDi{\pincE}=\vlholer{\strS'\,\vlscn{\atmb}}$, and
$\mapPincToDi{\pincF}=\vlholer{\strS'\,\vlscn{\vlone}}$,
then
$\pincLTSJud{\pincE}
            {\pincF}
            {\pincLabL}$,  where $\pincLabL$ is $\pincLabT$, if
$\atmb\in\strBN{\vlholer{\strS'\,\vlscn{\atmb}}}$, and $\pincLabL$ is $\atmb$, if
$\atmb\in\strFN{\vlholer{\strS'\,\vlscn{\atmb}}}$. The same holds by replacing
$\vlholer{\strS''\vlhole}$ for $\vlholer{\strS'\vlhole}$, and
$\natmb$ for $\atmb$.

\item
\label{enumerate:Trivial derivations and rightcontext s-04}
Let $\vlholer{\strS'\,\vlscn{\atmb}}$, and $\vlholer{\strS''\,\vlscn{\atmb}}$ be
the ones in Point~\eqref{enumerate:Trivial derivations and rightcontext s-02'} here
above.
If $ \pincE, \pincF, \pincE'$, and $\pincF'$ are processes \ST\
$\mapPincToDi{\pincE} =\vlholer{\strS'\,\vlscn{\atmb}},
 \mapPincToDi{\pincF} =\vlholer{\strS''\,\vlscn{\natmb}},
 \mapPincToDi{\pincE'}=\vlholer{\strS'\,\vlscn{\vlone}}$, and
$\mapPincToDi{\pincF'}=\vlholer{\strS''\,\vlscn{\vlone}}$,
then
$\pincLTSJud{\pincPar{\pincE}
                     {\pincF}}
            {\pincPar{\pincE'}
                     {\pincF'}}
            {\pincPreT}$.
\end{enumerate}
\end{proposition}
%%%%%%%%%
\begin{proof}
Concerning point~\eqref{enumerate:Trivial derivations and rightcontext s-10},
since no rule of $\bvtDder$ generates atoms both $\atmb$, and $\natmb$
must already occur in $\strR$.

\par%%-00----------------
Concerning point~\eqref{enumerate:Trivial derivations and rightcontext s-00},
we start from point~\eqref{enumerate:Trivial derivations and rightcontext s-10}, and
we look at $\vlstore{\vlsbr[\atmb;\natmb]}\vlholer{\strS\vlread}$ by
first ``hiding'' $\atmb$, which gives
$\vlholer{\strS_0\,\vlscn{\atmb}}
 \equiv
 \vlstore{\vlsbr[\atmb;\natmb]}
 \vlholer{\strS\vlread}$, for some $\vlholer{\strS_0\,\vlhole}$, and then
 ``hiding''  $\natmb$ yielding
$\vlholer{\strS_1\,\vlscn{\natmb}}
 \equiv
 \vlstore{\vlsbr[\atmb;\natmb]}
 \vlholer{\strS\vlread}$, for some
$\vlholer{\strS_1\,\vlhole}$.
Then, we apply point~\eqref{enumerate:Preserving rightcontext s-01} of
Fact~\ref{fact:Preserving rightcontext s} to
$\vlholer{\strS_0\,\vlscn{\atmb}}$. It implies that
$\strR\approx\strS'\,\vlscn{\atmb }$ is a \rightcontext, for some
$\strS'\vlhole$. Analogously, point~\eqref{enumerate:Preserving rightcontext s-01}
on Fact~\ref{fact:Preserving rightcontext s} to
$\vlholer{\strS_1\,\vlscn{\natmb}}$ implies that
$\strR\approx\strS''\,\vlscn{\natmb}$ is a \rightcontext, for some
$\strS''\vlhole$.

\par%%-01 --------
Point~\eqref{enumerate:Trivial derivations and rightcontext s-01},
directly follows from point~\eqref{enumerate:Trivial derivations and rightcontext s-00}.

\par%%-02--------
Point~\eqref{enumerate:Trivial derivations and rightcontext s-02} holds because, for example, $\atmb$ cannot enter the scope of $\vlfo{\atmb}{\vlholer{\strS''\,\vlscn{\natmb}}}$.

\par%%-02'-----------
Point~\eqref{enumerate:Trivial derivations and rightcontext s-02'} follows
from \eqref{enumerate:Trivial derivations and rightcontext s-02}.

\par%%-03-----------
Point~\eqref{enumerate:Trivial derivations and rightcontext s-03}
holds by proceeding inductively on $\Size{\pincE}$, and by cases on the form of
$\vlholer{\strS'\vlhole}$, or $\vlholer{\strS''\vlhole}$, respectively. (Details,
relative to $\vlholer{\strS'\vlhole}$, in Appendix~\ref{section:Proof of
Rightcontext s can preserve external communication}.)

\par%%-04-------------
Point~\eqref{enumerate:Trivial derivations and rightcontext s-04}
holds thanks to points~\eqref{enumerate:Trivial derivations and rightcontext s-02},
and~\eqref{enumerate:Trivial derivations and rightcontext s-03}, by proceeding inductively on $\Size{\pincPar{\pincE}{\pincF}}$, and by cases on the form of $\vlholer{\strS'\vlhole}$, and $\vlholer{\strS''\vlhole}$. (Details in Appendix~\ref{section:Proof of Rightcontext s can preserve internal communication}.)
\end{proof}
%%%%%%%%%%%
\par
The coming theorem says that the absence of interactions, as in a \trivialderivation, models non interacting transitions inside the \lts\ of $ \CCSR $. We include proof details here, and not in an Appendix, because this proof supplies tha simplest technical account of what we shall do for proving soundness.
%%%%%%%%%%%%%%%
\begin{theorem}[\textit{\textbf{\Trivialderivation s model empty computations in \lts}}]
\label{theorem:Soundness w.r.t. trivial deductions}
Let $\pincE$, and $\pincF$ be processes, with $\pincF$ \simple. If
$\vlstore{\mapPincToDi{\pincE}}
 \bvtInfer{\bvtDder}
          {\, \mapPincToDi{\pincF}
          \bvtJudGen{\BVT}{}
          \vlread}$ is trivial --- beware, not necessarily in $\BVTL$ ---,
then
$\pincLTSJud{\pincE}
            {\pincF}
            {\pincPreT}$.
\end{theorem}
%%%%%%%%%%%
\begin{proof}%%%%theorem:Soundness w.r.t. trivial deductions
Fact~\ref{fact:Trivial derivations preserve normal structures} implies that
$\mapPincToDi{\pincE}$ is \simple, like $\mapPincToDi{\pincF}$ is, and that $\bvtDder$ can only contain instances of $\bvtrdrulein$, if any rule occurs.
We proceed by induction on the number $n$ of instances of $\bvtrdrulein$ in
$\bvtDder$.
\par
If $n=0$, forcefully $\mapPincToDi{\pincE}\equiv\mapPincToDi{\pincF}$. We conclude by $\pincrefl$, \ie
$\pincLTSJud{\pincE}
            {\pincE}
            {\pincLabT}$. Otherwise, the last rule of $\bvtDder$ is:
{\small
$$
\vlinf{\bvtrdrule}{}
      {\strS\vlsbr[\vlfo{\atma}{\mapPincToDi{\pincE'}}
                  ;\vlfo{\atma}{\mapPincToDi{\pincE''}}]}
      {\strS\vlfo{\atma}{\vlsbr[\mapPincToDi{\pincE'}
                               ;\mapPincToDi{\pincE''}]}}
$$
}%\small
for some context $\strS\vlhole$, and processes $\pincE'$, and $\pincE''$, such that
$\mapPincToDi{\pincE}
 \approx
 \strS\vlsbr[\vlfo{\atma}{\mapPincToDi{\pincE'}}
            ;\vlfo{\atma}{\mapPincToDi{\pincE''}}]$.
We can proceed by cases on the form of $\strS\vlhole$.

\begin{itemize}
\item
%% 1-------------
Let $\strS\vlhole\approx\vlhole$. 
So, $\pincE$ must be $\pincPar{\pincNu{\pinca}{\pincE'}}
                             {\pincNu{\pinca}{\pincE''}}$, and we can write:
{\small
$$
\vlderivation{
\vliin{\pinctran}{}
      {\pincLTSJud{\pincPar{\pincNu{\pinca}
                                   {\pincE'}}
                           {\pincNu{\pinca}
                                   {\pincE''}}
                  }
	          {\pincF}
	          {\pincLabT}
      }
      {%%1---
	\vlin{\pincpi}{}
	     {\pincLTSJud{\pincPar{\pincNu{\pinca}
                                          {\pincE'}}
                                  {\pincNu{\pinca}
                                          {\pincE''}}
                         }
	                 {\pincPar{\pincNu{\pinca}
                                          {(\pincPar{\pincE'}
				 	            {\pincE''})}}
				  {\pincNu{\pinca}
                                          {\pincZer}
                                  }
                          \pincCong
	                  \pincNu{\pinca}
                                 {(\pincPar{\pincE'}
                                                       {\pincE''})
                                 }
                         }
	                 {\pincLabT}
             }{
	\vlin{\pincrefl}{}
	     {\pincLTSJud{\pincPar{\pincE'}
                                  {\pincE''}
                         }
	                 {\pincPar{\pincE'}
                                  {\pincE''}
                          \pincCong
			  \pincPar{\pincPar{\pincE'}
					   {\pincE''}}
				  {\pincZer}
                         }
	                 {\pincLabT}
             }{
             }{}}
      }
      {%%2---
	\vlhy{\pincLTSJud{\pincNu{\pinca}
                                 {(\pincPar{\pincE'}
                                                       {\pincE''})
                                 }
                         }
	                 {\pincF}
	                 {\pincLabT}
	     }
      }
}$$
}%\small
where
$\pincLTSJud{\pincNu{\pinca}
                    {(\pincPar{\pincE'}
                              {\pincE''})}} 
            {\pincF}
            {\pincLabT}$
holds by induction because
$\vlstore{\vlfo{\atma}{\vlsbr[\mapPincToDi{\pincE'}
                             ;\mapPincToDi{\pincE''}]}}
  \mapPincToDi{\pincF}\bvtJudGen{\Set{\bvtrdrule}}{} \vlread$ is shorter than $\bvtDder$.

\item
%% 2-------------
Let $\strS\vlhole\approx\vlsbr[\vlhole;\strT]$. 
So, $\pincE$ must be
$\pincPar{\pincPar{\pincNu{\pinca}{\pincE'}}
                  {\pincNu{\pinca}{\pincE''}}}
         {\pincF'}
         $, with $\mapPincToDi{\pincF'}=\strT$.
The case is analogous to the
previous one, with the proviso that an instance of $\pinccntxp$ must precede the
instance of $\pincpi$.
In particular,
$\pincLTSJud{\pincPar{\pincNu{\pinca}
			                 {(\pincPar{\pincE'}
				                       {\pincE''})
			                 }
		             }
		    {\pincF'}
	    }
	    {\pincF}
	    {\pincLabT}$
holds by induction because
$\vlstore{\vlsbr[\vlfo{\atma}{\vlsbr[\mapPincToDi{\pincE'}
                                    ;\mapPincToDi{\pincE''}]}
                ;\mapPincToDi{\pincF'}]}
  \mapPincToDi{\pincF}\bvtJudGen{\Set{\bvtrdrule}}{} \vlread$ 
is shorter than $\bvtDder$.
\end{itemize}

%% 3-------------
\par
The third case 
$\vlstore{\vlsbr<\pincLabL;\vlhole>}
 \strS\vlhole\approx\vlread$ that we could obtain by assuming 
$\pincE =
\pincSec{\pincLabL}
        {\pincE'}$ cannot occur because $\pincE$ would not be \simple, against assumptions.

\end{proof}%%%%theorem:Soundness w.r.t. trivial deductions
%%%%%%%%%%%
\begin{remark}[\textbf{\textit{Why do we define \simplestructure s as such?}}]
\label{remark:Why process structures include normal ones}
Theorem~\ref{theorem:Soundness w.r.t. trivial deductions} would not hold if we used ``\processstructure s'' in place of ``\simplestructure s''. Let us pretend, for a moment, that $ \pincF $ be any \processstructure, and not only a \simple\ one, indeed. The bottommost rule in $ \bvtDder $ might well be:
{\small
$$
\vlinf{\bvtseqdrule}{}
      {\vlsbr[\mapPincToDi{\pincE'}
             ;<\mapPincToDi{\pincLabL}
              ;\mapPincToDi{\pincE''}>]}
      {\vlsbr<\mapPincToDi{\pincLabL}
             ;[\mapPincToDi{\pincE'}
              ;\mapPincToDi{\pincE''}]>}
$$
}%\small
for some $\pincE'$, and $\pincE''$, such that 
$\pincE = \pincPar{\pincE'}
                     {(\pincSec{\pincLabL}
                               {\pincE''})}$.
By induction,
$\pincLTSJud{\pincSec{\pincLabL}
                     {(\pincPar{\pincE'}
                               {\pincE''})}
            }
            {\pincF}
            {\pincLabT}$. However, in the \lts~\eqref{equation:PPi-LTS-from-BVT} of
$\CCSR$ we cannot deduce 
$\pincLTSJud{\pincPar{\pincE'}
                     {(\pincSec{\pincLabL}
                               {\pincE''})}
            }
            {\pincSec{\pincLabL}
                     {(\pincPar{\pincE'}
                               {\pincE''})}}
            {\pincLabT}$ whenever $ \pincLabL $ occurs free in $ \pincE' $.
So, as we did in the definition of \simpleprocess es, we must eliminate any occurrence of \OpNameSeq structure.
\end{remark}
%%%%%%%%%%%%%
% \newpage
\begin{theorem}[\textbf{\textit{Soundness \wrt\ internal communication}}]
\label{theorem:Soundness w.r.t. internal communication}
Let $\pincE$, and $\pincF$ be processes, with $\pincF$
\simple, and $\pincE\not\approx\vlone$. 
Let $\bvtDder$ be the derivation
{\small
$\vlderivation{
  \vlde{\bvtDder'}{\BVTL}
       {\mapPincToDi{\pincE}}
       {
  \vlin{\bvtatrdrule}{(*)}
       {\vlholer{
        \strS\vlsbr[\atmb;\natmb]
        }
       }
       {
  \vlde{\bvtDder''}{\BVTL}
       {\vlholer{
        \strS\vlscn{\vlone}
        }
       }
       {
  \vlhy{\mapPincToDi{\pincF}}}}}
}
$\,}%\small
which, besides being \standard, we assume to be \nontrivial, and
\ST\ $(*)$ is its lowermost instance of $\bvtatrdrulein$.
If, for some process $\pincG$, the derivation
$\bvtInfer{\bvtEder}
          {\,\mapPincToDi{\pincF}
          \bvtJudGen{\BVTL}{}
          \mapPincToDi{\pincG}}$ is the reduction of $\bvtDder$,
then
$\pincLTSJud{\pincE}
            {\pincG}
            {\pincLabT}$.
\end{theorem}
%%%%%%%%%%%
\begin{proof}%%%Soundness w.r.t. internal communication
The derivation $\bvtDder'$ satisfies the assumptions of
Point~\eqref{enumerate:Trivial derivations and rightcontext s-00} in
Proposition~\ref{fact:Trivial derivations and rightcontext s} which implies
$\mapPincToDi{\pincE}\approx\vlholer{\strS' \,\vlscn{\atmb}}$, and $\mapPincToDi{\pincE}\approx\vlholer{\strS''\,\vlscn{\natmb}}$, for some 
$ \vlholer{\strS'\,\vlscn{\natmb}} $, and $ \vlholer{\strS''\,\vlscn{\natmb}} $,
which must be \processstructure s. We proceed on the possible distinct forms that $\mapPincToDi{\pincE}$ can assume. Point~\eqref{enumerate:Trivial derivations and rightcontext s-04} of Proposition~\ref{proposition:Rightcontext s preserve communication} will
help concluding. (Details in Appendix~\ref{section:Proof of theorem:Soundness w.r.t.
internal communication}.)
\end{proof}%%%%Soundness w.r.t. internal communication
%%%%%%%%%%%
\paragraph{\Environmentstructure s that get consumed.}
Let $\strT$, and $\strU$ be \processstructure s, and $\strR$ be an
\environmentstructure. Let 
$\vlstore{
  \vlsbr[\strT;\strR]
 }
 \bvtInfer{\bvtDder}
          {\, \strU
          \bvtJudGen{\BVTL}{}
          \vlread}$ which, since belongs to $\BVTL$, is \standard.
We say that $\bvtDder$ \dfn{consumes $\strR$} if every atom of $\strR$ eventually
annihilates with an atom of $\strT$ thanks to an instance of $\bvtatrdrulein$, so
that none of them occurs in $\strU$.
%%%%%%%%%%%
\begin{example}[\textbf{\textit{Consuming environment structures}}]
Derivations that consume the environment structure $\vlsbr<\natmaRed;\atmbBlu>$ that
occurs in their conclusion are~\eqref{equation:tracing-sequential-interactions-01},
and \eqref{equation:tracing-sequential-interactions-00}.
If we consider only a part of \eqref{equation:tracing-sequential-interactions-01},
as here below, we get a standard derivation that does not consume
$\vlsbr<\natmaRed;\atmbBlu>$:
%%%%%%%%%%%%%
\par\vspace{\baselineskip}\noindent
{\small
  \fbox{
    \begin{minipage}{.974\linewidth}
%      \vspace{-.5cm}
       \begin{equation}
	 \label{equation:BV2-consuming-environment-structures}
 	 \input{./BV2-consuming-environment-structures}
       \end{equation}
%      \vspace{-.7cm}
    \end{minipage}
  }%fbox
}%\small
\end{example}
%%%%%%%%%%%
\begin{theorem}[\textbf{\textit{Soundness w.r.t. external communication}}]
\label{theorem:Soundness w.r.t. external communication}
Let $\pincE$, and $\pincF$ be processes, and $\strR$ be an
\environmentstructure.
Let $\pincF$ be \simple, and $\pincE\not\approx\vlone$.
Let $\bvtDder$ be a \nontrivial, and \standardderivation\ that assumes one of the
two following forms:
{\small
$$\vlderivation{
  \vlde{\bvtDder'}{\BVTL}
       {\vlsbr[\mapPincToDi{\pincE};<\natmb;\strR>]}
       {
  \vlin{\bvtatrdrule}{(*)}
       {\vlholer{
        \strS\vlsbr[\atmb;\natmb]
        }
       }
       {
  \vlde{\bvtDder''}{\BVTL}
       {\vlholer{
        \strS\vlscn{\vlone}
        }
       }
       {
  \vlhy{\mapPincToDi{\pincF}}}}}}
\qquad\qquad\qquad\textrm{or}\qquad\qquad\qquad
\vlderivation{
  \vlde{\bvtDder'}{\BVTL}
       {\vlsbr[\mapPincToDi{\pincE};\vlfo{\atmb}{<\natmb;\strR>}]}
       {
  \vlin{\bvtatrdrule}{(*)}
       {\vlholer{
        \strS\vlsbr[\atmb;\natmb]
        }
       }
       {
  \vlde{\bvtDder''}{\BVTL}
       {\vlholer{
        \strS\vlscn{\vlone}
        }
       }
       {
  \vlhy{\mapPincToDi{\pincF}}}}}}
$$}%\small
\ST\
$(*)$ is its lowermost instance of $\bvtatrdrulein$, 
and $\natmb$ in $\strS\vlsbr[\atmb;\natmb]$ is the same occurrence of $\natmb$ as the one in
$\vlsbr<\natmb;\strR>$.
If
$\vlstore{\vlsbr[\mapPincToDi{\pincG};\strR]}
 \bvtInfer{\bvtEder}
          {\, \mapPincToDi{\pincF}
          \bvtJudGen{\BVTL}{}
          \vlread}$ is the reduction of $\bvtDder$,
then
$\pincLTSJud{\pincE}
            {\pincG}
            {\pincLabT}$ if
$\atmb\in\strBN{\pincE}$. Otherwise, if $\atmb\in\strFN{\pincE}$, then
$\pincLTSJud{\pincE}
            {\pincG}
            {\pincb}$.
\end{theorem}
%%%%%%%%%%
\begin{proof}%%%%theorem:Soundness w.r.t. external communication nuovo
First, $\bvtDder$ necessarily consumes $\vlsbr<\natmb;\strR>$, or
$\vlstore{\vlsbr<\natmb;\strR>} \vlfo{\atmb}{\vlread}$ in either cases. 
The reason is twofold. Being $\mapPincToDi{\pincF}$ a 
\simplestructure\ implies it cannot contain any \OpNameSeq\ structure which, instead, is one of the operators that can compose $\vlsbr<\natmb;\strR>$, or
$\vlstore{\vlsbr<\natmb;\strR>} \vlfo{\atmb}{\vlread}$. Moreover, no occurrence of $\atmb$
inside $\strR$ can annihilate with the first occurrence of $\natmb$ inside $\vlsbr<\natmb;\strR>$, or $\vlstore{\vlsbr<\natmb;\strR>} \vlfo{\atmb}{\vlread}$.
\par
Second, $\bvtDder'$ satisfies the assumptions of Proposition~\ref{proposition:Rightcontext s preserve communication}. So, its
Point~\eqref{enumerate:Trivial derivations and rightcontext s-00} applies to
$\vlstore{\vlsbr<\natmb;\strR>}
 \vlsbr[\mapPincToDi{\pincE};\vlread]$, and
$\vlstore{\vlsbr<\natmb;\strR>}
 \vlsbr[\mapPincToDi{\pincE};\vlfo{\atmb}{\vlread}]$. 
Since $\natmb$ occurs in $\vlsbr<\natmb;\strR>$, for some $\vlholer{\strS'\vlhole}$, it must be $\mapPincToDi{\pincE}\approx\vlholer{\strS'\,\vlscn{\atmb }}$ in which the occurrence of $ \atmb $ we outline is the one that annihilates the given $ \natmb $.
We proceed on the possible forms that
%$\vlholer{\strS'\,\vlscn{\atmb}}$, which must be a \processstructure, 
$\mapPincToDi{\pincE}$ can assume, in relation with the form of $\strR$.
Point~\eqref{enumerate:Trivial derivations and rightcontext s-03} of
Proposition~\ref{proposition:Rightcontext s preserve communication} will help
concluding. (Details in Appendix~\ref{section:Proof of theorem:Soundness w.r.t.
external communication}.)
\end{proof}%%%%theorem:Soundness w.r.t. external communication nuovo
%%%%%%%%%%%
\begin{theorem}[\textbf{\textit{Soundness}}]
\label{equation:PPi-soundness-example-00}
Let $\pincE$, and $\pincF  $ be processes with $\pincF $ \simple.
For every \standardderivation\ $\bvtDder $, and 
every \environmentstructure\ $ \strR $,
if
$ \vlderivation{
  \vlde{\bvtDder}{\BVTL}
       {\vlsbr[\mapPincToDi{\pincE};\strR]}{
  \vlhy{\mapPincToDi{\pincF}}}} $, and
$\bvtDder$ consumes $\strR$,
then
$\pincLTSJud{\pincE}{\pincF}
            {\mapDiToPinc{\strR}{\emptyset}}$.
\end{theorem}
%%%%
\begin{proof}
As a basic case we assume $\mapPincToDi{\pincE}\approx\vlone$. This means that 
$ \pincE $  is $ \pincZer $. Moreover, since $\bvtDder$ consumes $\strR$, and no atom exists in $\mapPincToDi{\pincE}$ to annihilate atoms of $\strR$, we must have $\mapPincToDi{\pincF}\approx\vlone$, \ie $ \pincF \equiv \pincZer$, 
and $\strR\approx\vlone$. Since $\pincLTSJud{\pincZer}{\pincZer}{\pincPreT}$, thanks to $\pincrefl$, we are done.
\par
Instead, if $\mapPincToDi{\pincE}\not\approx\vlone$, in analogy with \cite{Brus:02:A-Purely:wd},
we proceed by induction on the number of rules in $\bvtDder$, in relation with the
two cases where $\strR\approx\vlone$, or $\strR\not\approx\vlone$.
\par
Since $\bvtDder$ is \nontrivial, and \standard,
we can focus on its lowermost occurrence $ (*) $ of $\bvtatrdrulein$. Let us assume the redex of $ (*) $ be  $\vlsbr[\atmb;\natmb]$. We can have the following cases.
\begin{itemize}
 %%1
 \item
 Let $\strR\approx\vlone$,
 and $\bvtInfer{\bvtEder}
               {\, \mapPincToDi{\pincF}
                \bvtJudGen{\BVT}{}
                \mapPincToDi{\pincG}}$ be the reduction of $\bvtDder$.
 \begin{enumerate}
  %%% 1.1
  \item
  The first case is with $\bvtEder$ \nontrivial. The inductive hypothesis holds on
$\bvtEder$, and we get
$\pincLTSEXTJud{\pincG}
               {\pincF}
               {\pincPreT=\mapDiToPinc{\vlone}{\emptyset}}
               {r}{&}$.

  %%% 1.2
  \item
  The second case is with $\bvtEder$ trivial, so we cannot apply the inductive
hypothesis on $\bvtEder$. However, Theorem~\ref{theorem:Soundness w.r.t.
trivial deductions} holds on $\bvtEder$, and we get
$\pincLTSJud{\pincG}
            {\pincF}
            {\pincPreT}$.
 \end{enumerate}
Finally, both $\bvtDder$, and $ \bvtEder $ satisfy the assumptions of Theorem~\ref{theorem:Soundness w.r.t. internal communication}, so it
implies
$\pincLTSJud{\pincE}
            {\pincG}
            {\pincPreT}$,
and the statement we are proving holds thanks to $\pinctran$.

%%2
 \item
 Let $\vlone
      \not\approx\strR
      \approx\vlstore{\vlsbr<\natmb;\strT>}
             \vlfo{\atmb}{\vlread}$, for some \environmentstructure\ $ \strT $.
 Let $\vlstore{\vlsbr[\mapPincToDi{\pincG};\vlfo{\atmb}{\vlsbr<\vlone;\strT>}]}
      \bvtInfer{\bvtEder}
               {\, \mapPincToDi{\pincF}
               \bvtJudGen{\BVT}{}
               \vlread}$ be the reduction of $\bvtDder$.
Since $\vlstore{\vlsbr<\vlone;\strT>}
       \vlfo{\atmb}{\vlread}$ is an \environmentstructure, it is \canonical, so,
necessarily
$\vlstore{\vlsbr<\vlone;\strT>}
 \vlfo{\atmb}{\vlread} \approx \vlfo{\atmb}{\strT} \approx\strT$ because $ \natmb\not\in\strFN{\strT} $.
Hence, $\vlstore{\vlsbr[\mapPincToDi{\pincG};\strT]}
       \bvtInfer{\bvtEder}
                {\, \mapPincToDi{\pincF}
                \bvtJudGen{\BVT}{}
                \vlread}$.
Moreover, since $\natmb$ disappears along $ \bvtDder $, we forcefully have 
$\atmb\in\strBN{\mapPincToDi{\pincE}}$.
%So, $\Set{\atmb, \natmb}\cap\strFN{\strT}=\emptyset$, and we can apply
%    $\vlfo{\atmb}{\strT}\approx\strT$. Moreover, it must be.

 \begin{enumerate}
  %%% 2.1
  \item
  Let $\bvtEder$ be \nontrivial. The inductive hypothesis holds on
$\bvtEder$, implying
$\pincLTSJud{\pincG}
               {\pincF}
               {\mapDiToPinc{\strT}{\emptyset}}$.
Moreover, $ \bvtDder $ satisfies the assumptions of Theorem~\ref{theorem:Soundness w.r.t. external communication} which implies
$\pincLTSJud{\pincE}
            {\pincG}
            {\pincLabT}$ also because, as we said, $\atmb\in\strBN{\mapPincToDi{\pincE}}$.
So, the statement holds because
$\mapDiToPinc{\strT}{\Set{\atmb,\natmb}}
 \pincCong
 \pincLabT;\mapDiToPinc{\strT}{\Set{\atmb,\natmb}}
 =\mapDiToPinc{\natmb}{\Set{\atmb,\natmb}}
  ;\mapDiToPinc{\strT}{\Set{\atmb,\natmb}}
 =\vlstore{\vlsbr<\natmb;\strT>}
  \mapDiToPinc{\vlfo{\atmb}{\vlread}}{\emptyset}$, and by $\pinctran$ we get
$\vlstore{\vlsbr<\natmb;\strT>}
 \pincLTSEXTJud{\pincE}
               {\pincF}
               {\mapDiToPinc{\vlfo{\atmb}{\vlread}}{\emptyset}}
               {r}{&}$.

  %%% 2.2--------
  \item
  The second case is with $\bvtEder$ \trivial, so we cannot apply the inductive
hypothesis on $\bvtEder$. However, Theorem~\ref{theorem:Soundness w.r.t.
trivial deductions} holds on $\bvtEder$, and we get
$\pincLTSJud{\pincG}
            {\pincF}
            {\pincPreT}$, which implies $\strT\approx\vlone$. Indeed, if $\strT\not\approx\vlone$, then $ \bvtDder' $ could not consume $\strT$.
The reason is that being $ \bvtEder $ a \trivialderivation, it cannot contain any instance of $\bvtatidrulein$. But a $ \bvtDder' $ not consuming $ \strT $, would mean $ \bvtDder $ not consuming $ \strR $, against assumption.
Finally, Theorem~\ref{theorem:Soundness w.r.t. external communication} holds on
$\bvtDder$, and implies
$\pincLTSJud{\pincE}
            {\pincG}
            {\pincLabT}$, because, as we said, $\atmb\in\strBN{\mapPincToDi{\pincE}}$.
So, the statement holds because
$\mapDiToPinc{\vlone}{\Set{\atmb,\natmb}}
 \pincCong
 \pincLabT;\mapDiToPinc{\vlone}{\Set{\atmb,\natmb}}
 =\mapDiToPinc{\natmb}{\Set{\atmb,\natmb}}
  ;\mapDiToPinc{\vlone}{\Set{\atmb,\natmb}}
 =\vlstore{\vlsbr<\natmb;\vlone>}
  \mapDiToPinc{\vlfo{\atmb}{\vlread}}{\emptyset}$, and by $\pinctran$ we get
$\vlstore{\vlsbr<\natmb;\vlone>}
 \pincLTSEXTJud{\pincE}
               {\pincF}
               {\mapDiToPinc{\vlfo{\atmb}{\vlread}}{\emptyset}}
               {r}{&}$.
 \end{enumerate}
We could proceed in the same way when
$\vlone \not\approx\strR
        \approx\vlstore{\vlsbr<\atmb;\strT>}
             \vlfo{\atmb}{\vlread}$.

%%3
 \item
 Let $\vlone\not
      \approx\strR
      \approx\vlstore{\vlsbr<\natmb;\strT>}\vlread$.
 Then, both
      $\vlstore{\vlsbr[\mapPincToDi{\pincG};\strT]}
      \bvtInfer{\bvtEder}
               {\, \mapPincToDi{\pincF}
               \bvtJudGen{\BVT}{}
               \vlread}$, and
$\atmb\in\strFN{\mapPincToDi{\pincE}}$ for the reasons analogous to the ones given in the previous case.

 \begin{enumerate}
  %%% 3.1
  \item
  The first case is with $\bvtEder$ \nontrivial. The inductive hypothesis holds on
$\bvtEder$, and we get
$\pincLTSJud{\pincG}
               {\pincF}
               {\mapDiToPinc{\strT}{\emptyset}}
               $.
Moreover, Theorem~\ref{theorem:Soundness w.r.t. external communication} holds on
$\bvtDder$, and implies
$\pincLTSJud{\pincE}
            {\pincG}
            {\pincb}$, because, as we said, $\atmb\in\strFN{\mapPincToDi{\pincE}}$.
So, the statement holds because
$\mapDiToPinc{\natmb}{\emptyset}
  ;\mapDiToPinc{\strT}{\emptyset}
 =\vlstore{\vlsbr<\natmb;\strT>}
  \mapDiToPinc{\vlread}{\emptyset}$, and by $\pinctran$ we get
$\vlstore{\vlsbr<\natmb;\strT>}
 \pincLTSEXTJud{\pincE}
               {\pincF}
               {\mapDiToPinc{\vlread}{\emptyset}}
               {r}{&}$.

  %%% 3.2--------
  \item
  The second case is with $\bvtEder$ trivial, so we cannot apply the inductive
hypothesis on $\bvtEder$. However, Theorem~\ref{theorem:Soundness w.r.t.
trivial deductions} holds on $\bvtEder$, and we get
$\pincLTSJud{\pincG}
            {\pincF}
            {\pincPreT}$, which implies $\strT\approx\vlone$ for reasons analogous to the
ones given in the previous case. 
Moreover, Theorem~\ref{theorem:Soundness w.r.t. external communication} holds on
$\bvtDder$, and implies
$\pincLTSJud{\pincE}
            {\pincG}
            {\pincb}$, because, as we said, $\atmb\in\strFN{\mapPincToDi{\pincE}}$.
So, the statement holds because
$\mapDiToPinc{\natmb}{\emptyset}
 \pincCong
  \mapDiToPinc{\natmb}{\emptyset}
  ;\pincLabT
 =\mapDiToPinc{\natmb}{\emptyset}
  ;\mapDiToPinc{\vlone}{\emptyset}
 =\vlstore{\vlsbr<\natmb;\vlone>}
  \mapDiToPinc{\vlread}{\emptyset}$, and by $\pinctran$ we get
$\vlstore{\vlsbr<\natmb;\vlone>}
 \pincLTSEXTJud{\pincE}
               {\pincF}
               {\mapDiToPinc{\vlread}{\emptyset}}
               {r}{&}$.
 \end{enumerate}
We could proceed in the same way when
$\vlone \not\approx\strR
        \approx\vlstore{\vlsbr<\atmb;\strT>}
        \vlread$. 
\end{itemize}
\end{proof}
%%%%%%%%%%%
\subsection{An instance of the proof of Soundness}
\label{subsection:An instance of Soundness proof}
The derivation~\eqref{equation:example-reduction-10} is \standard.
%%%%%%%%%%%%%%%%%
\par\vspace{\baselineskip}\noindent
{\small
  \fbox{
    \begin{minipage}{.974\linewidth}
      \begin{equation}
       \label{equation:example-reduction-10}
		  \begin{minipage}{.47\textwidth}
		    \input{./SBV2-example-reduction-10}
		  \end{minipage}
      \end{equation}
    \end{minipage}
  }%fbox
}%\small
\par\vspace{\baselineskip}\noindent
%%%%%%%%%%%
Hence, \eqref{equation:example-reduction-10} is an instance of the assumption 
$\vlstore{
\vlsbr[\mapPincToDi{\pincE};\strR]
}
\bvtInfer{\bvtDder}
         {\, \mapPincToDi{\pincF}
         \bvtJudGen{\BVTL}{}
         \vlread}$ in Theorem~\ref{equation:PPi-soundness-example-00} above.
The structure 
$\vlstore{\vlsbr[<\atmaRed;\atmbBlu;\mapPincToDi{\pincE'}>
                ;<\natmaRed;\mapPincToDi{\pincF'}>]} 
 \vlfo{\atma}{\vlread}$ in~\eqref{equation:example-reduction-10}
plays the role of $ \mapPincToDi{\pincE} $,
while $ \natmbBlu $ corresponds to $ \strR $. Finally 
$\vlstore{\vlsbr[\mapPincToDi{\pincE'}
                ;\mapPincToDi{\pincF'}]}
 \vlfo{\atma}{\vlread}$ plays the role of $ \mapPincToDi{\pincF}$, for some process $\pincE'$, and $\pincF'$.
By definition,
$ \pincE = 
  \pincNu{\pinca}
         {(\pincPar{(\pincSec{\pincaRed}
                             {\pincSec{\pincbBlu}{\pincE'}})}
                   {(\pincSec{\pincnaRed}{\pincF'})})}$, and
$ \pincF = 
  \pincNu{\pinca}
         {(\pincPar{\pincE'}{\pincF'})}$.
%%%%%%%%%%%%%
Once identified the lowermost instance $(*)$ of $\bvtatrdrulein$, we replace $\vloneRed$ for all those occurrences of atoms that, eventually, annihilate in $(*)$.
So, \eqref{equation:example-reduction-10} becomes the structure~\eqref{equation:example-reduction-00} which is not a derivation because it contains fake instances of rules.
%%%%%%%%%%%%%%%%%
\par\vspace{\baselineskip}\noindent
{\small
  \fbox{
    \begin{minipage}{.974\linewidth}
      \begin{equation}
       \begin{minipage}{.45\textwidth}
         \label{equation:example-reduction-00}
	     $$\input{./SBV2-example-reduction-00}$$
       \end{minipage}
      \end{equation}
    \end{minipage}
  }%fbox
}%\small
\par\vspace{\baselineskip}\noindent
%%%%%%%%%%%
Removing all the fake rules, we get to $ \bvtEder $ in~\eqref{equation:example-reduction-20}:
%%%%%%%%%%%%%%%%%
\par\vspace{\baselineskip}\noindent
{\small
  \fbox{
    \begin{minipage}{.974\linewidth}
      \begin{equation}
       \label{equation:example-reduction-20}
		%\begin{gathered}
		  \begin{minipage}{.47\textwidth}
	       $$\input{./SBV2-example-reduction-20}$$
		  \end{minipage}
		%\end{gathered}
      \end{equation}
    \end{minipage}
  }%fbox
}%\small
\par\vspace{\baselineskip}\noindent
%%%%%%%%%%%
The lowermost instance $(*)$ of $ \bvtatrdrulein $ in~\eqref{equation:example-reduction-10} has disappeared from~\eqref{equation:example-reduction-20}.
The inductive argument on~\eqref{equation:example-reduction-20} implies
$\pincLTSJud{\pincNu{\pinca}
                    {(\pincPar{(\pincSec{\pincbBlu}
                                        {\pincE'})}
                              {\pincF'})}
            }
            {\pincNu{\pinca}
                    {(\pincPar{\pincE'}{\pincF'})}
            }
            {\pincbBlu}$.
Since we can prove:
%%%%%%%%%%%%%
\par\vspace{\baselineskip}\noindent
{\scriptsize
  \fbox{
    \begin{minipage}{.974\linewidth}
    \vspace{-.1cm}
      \begin{equation}
      \label{equation:PPi-example-LTS-02-reconstruction}
      \input{./PPi-example-LTS-02-reconstruction}
      \end{equation}
    \end{minipage}
  }%fbox
}%\small
\vspace{\baselineskip}\par\noindent
by transitivity, we conclude
$\pincLTSJud{\pincNu{\pinca}
                    {(\pincPar{(\pincSec{\pincaRed}
                                        {\pincSec{\pincbBlu}
                                                 {\pincE'}})
                              }
                              {(\pincSec{\pincnaRed}
                                        {\pincF'})}
                     )}
            }
            {\pincNu{\pinca}
                    {(\pincPar{\pincE'}{\pincF'})}
            }
            {\pincbBlu}$.
%%%%%%%%%%%%

%% file: BV2-consuming-environment-structures.tex
\vlderivation                   {
  \vlin{\bvtseqdrule}{}
       {\vlsbr
        [<\atmaRed;\strT>
        ;<\natmbBlu;\strU>
        ;<\natmaRed;\atmbBlu>]}{
  \vliq{\bvtatrdrule}{}
       {\vlsbr
        [<[\atmaRed;\natmaRed]
         ;[\strT;\atmbBlu]>
        ;<\natmbBlu;\strU>]}{
  \vlhy{\vlsbr
        [\strT
        ;<\natmbBlu;\strU>
        ;\atmbBlu]}}}}

%% file: SBV2-example-reduction-10.tex
   $$
    \vlderivation{
      \vlin{\bvtseqdrule}{}
	  {\vlsbr
	    [\vlfo{\atma}
		  {[<\atmaRed;\atmbBlu;\mapPincToDi{\pincE'}>
		  ;<\natmaRed;\mapPincToDi{\pincF'}>]}
	    ;\natmbBlu]
	  }{
      \vliq{\eqref{align:unit-seq}}{}
	  {\vlsbr
	    [\vlfo{\atma}
		  {<[\atmaRed;\natmaRed]
		  ;[<\atmbBlu;\mapPincToDi{\pincE'}>;\mapPincToDi{\pincF'}]>}
	    ;\natmbBlu]
	  }{
      \vlin{\bvtseqdrule}{}
	  {\vlsbr
	    [\vlfo{\atma}
		  {<[\atmaRed;\natmaRed]
		  ;[<\atmbBlu;\mapPincToDi{\pincE'}>
		    ;<\vlone;\mapPincToDi{\pincF'}>]>}
	    ;\natmbBlu]
	  }{
      \vliq{\eqref{align:unit-pa}
	  ,\eqref{align:alpha-intro}}{}
	  {\vlsbr
	    [\vlfo{\atma}
		  {<[\atmaRed;\natmaRed]
		  ;<[\atmbBlu;\vlone]
		    ;[\mapPincToDi{\pincE'};\mapPincToDi{\pincF'}]>>}
	    ;\natmbBlu]
	  }{
      \vlin{\bvtrdrule}{}
	  {\vlsbr
	    [\vlfo{\atma}
		  {<[\atmaRed;\natmaRed]
		  ;<\atmbBlu
		    ;[\mapPincToDi{\pincE'};\mapPincToDi{\pincF'}]>>}
	    ;\vlfo{\atma}{\natmbBlu}]
	  }{
      \vliq{\eqref{align:unit-seq}}{}
	  {\vlfo{\atma}
	    {\vlsbr
	    [<[\atmaRed;\natmaRed]
	      ;<\atmbBlu
		;[\mapPincToDi{\pincE'};\mapPincToDi{\pincF'}]>>
	    ;\natmbBlu]
	    }
	  }{
      \vlin{\bvtseqdrule}{}
	  {\vlfo{\atma}
	    {\vlsbr
	    [<[\atmaRed;\natmaRed]
	      ;<\atmbBlu
		;[\mapPincToDi{\pincE'};\mapPincToDi{\pincF'}]>>
	    ;<\vlone;\natmbBlu>]
	    }
	  }{
      \vliq{\eqref{align:unit-pa}
	  ,\eqref{align:unit-seq}}{}
	  {\vlfo{\atma}
	    {\vlsbr
	    <[\atmaRed;\natmaRed;\vlone]
	      ;[<\atmbBlu;[\mapPincToDi{\pincE'};\mapPincToDi{\pincF'}]>
	      ;\natmbBlu]>
	    }
	  }{
      \vlin{\bvtseqdrule}{}
	  {\vlfo{\atma}
	    {\vlsbr
	    <[\atmaRed;\natmaRed]
	      ;[<\atmbBlu;[\mapPincToDi{\pincE'};\mapPincToDi{\pincF'}]>
	      ;<\natmbBlu;\vlone>]>
	    }
	  }{
      \vliq{\bvtatrdrule
	  ,\eqref{align:unit-seq}
	  ,\eqref{align:unit-pa}^2}{(*)}
	  {\vlfo{\atma}
	    {\vlsbr
	    <[\atmaRed;\natmaRed]
	      ;<[\atmbBlu;\natmbBlu]
	      ;[\mapPincToDi{\pincE'};\mapPincToDi{\pincF'};\vlone]>>
	    }
	  }{
      \vliq{\bvtatrdrule
	  ,\eqref{align:unit-seq}}{}
	  {\vlfo{\atma}
	    {\vlsbr
	    <[\atmbBlu;\natmbBlu]
	      ;[\mapPincToDi{\pincE'};\mapPincToDi{\pincF'}]>
	    }
	  }{
      \vlhy{\vlfo{\atma}
 	             {\vlsbr[\mapPincToDi{\pincE'};\mapPincToDi{\pincF'}]}
      }}}}}}}}}}}}
      }
   $$

%% file: SBV2-example-reduction-00.tex
    \vlderivation{
      \vlin{}{}
	  {\vlsbr
	    [\vlfo{\atma}
		  {[<\vloneRed;\atmbBlu;\mapPincToDi{\pincE'}>
		  ;<\vloneRed;\mapPincToDi{\pincF'}>]}
	    ;\natmbBlu]
	  }{
      \vliq{}{}
	  {\vlsbr
	    [\vlfo{\atma}
		  {<[\vloneRed;\vloneRed]
		  ;[<\atmbBlu;\mapPincToDi{\pincE'}>;\mapPincToDi{\pincF'}]>}
	    ;\natmbBlu]
	  }{
      \vlin{}{}
	  {\vlsbr
	    [\vlfo{\atma}
		  {<[\vloneRed;\vloneRed]
		  ;[<\atmbBlu;\mapPincToDi{\pincE'}>
		    ;<\vlone;\mapPincToDi{\pincF'}>]>}
	    ;\natmbBlu]
	  }{
      \vliq{}{}
	  {\vlsbr
	    [\vlfo{\atma}
		  {<[\vloneRed;\vloneRed]
		  ;<[\atmbBlu;\vlone]
		    ;[\mapPincToDi{\pincE'};\mapPincToDi{\pincF'}]>>}
	    ;\natmbBlu]
	  }{
      \vlin{}{}
	  {\vlsbr
	    [\vlfo{\atma}
		  {<[\vloneRed;\vloneRed]
		  ;<\atmbBlu
		    ;[\mapPincToDi{\pincE'};\mapPincToDi{\pincF'}]>>}
	    ;\vlfo{\atma}{\natmbBlu}]
	  }{
      \vliq{}{}
	  {\vlfo{\atma}
	    {\vlsbr
	    [<[\vloneRed;\vloneRed]
	      ;<\atmbBlu
		;[\mapPincToDi{\pincE'};\mapPincToDi{\pincF'}]>>
	    ;\natmbBlu]
	    }
	  }{
      \vlin{}{}
	  {\vlfo{\atma}
	    {\vlsbr
	    [<[\vloneRed;\vloneRed]
	      ;<\atmbBlu
		;[\mapPincToDi{\pincE'};\mapPincToDi{\pincF'}]>>
	    ;<\vlone;\natmbBlu>]
	    }
	  }{
      \vliq{}{}
	  {\vlfo{\atma}
	    {\vlsbr
	    <[\vloneRed;\vloneRed;\vlone]
	      ;[<\atmbBlu;[\mapPincToDi{\pincE'};\mapPincToDi{\pincF'}]>
	      ;\natmbBlu]>
	    }
	  }{
      \vlin{}{}
	  {\vlfo{\atma}
	    {\vlsbr
	    <[\vloneRed;\vloneRed]
	      ;[<\atmbBlu;[\mapPincToDi{\pincE'};\mapPincToDi{\pincF'}]>
	      ;<\natmbBlu;\vlone>]>
	    }
	  }{
      \vliq{}{}
	  {\vlfo{\atma}
	    {\vlsbr
	    <[\vloneRed;\vloneRed]
	      ;<[\atmbBlu;\natmbBlu]
	      ;[\mapPincToDi{\pincE'};\mapPincToDi{\pincF'};\vlone]>>
	    }
	  }{
      \vliq{}{}
	  {\vlfo{\atma}
	    {\vlsbr
	    <[\atmbBlu;\natmbBlu]
	      ;[\mapPincToDi{\pincE'};\mapPincToDi{\pincF'}]>
	    }
	  }{
      \vlhy{\vlfo{\atma}
	    {\vlsbr[\mapPincToDi{\pincE'};\mapPincToDi{\pincF'}]
	    }
      }}}}}}}}}}}}
      }

%% file: SBV2-example-reduction-20.tex
   \vlderivation{
      \vliq{\eqref{align:unit-seq}}{}
	  {\vlsbr
	    [\vlfo{\atma}
		  {[<\atmbBlu;\mapPincToDi{\pincE'}>;\mapPincToDi{\pincF'}]}
	    ;\natmbBlu]
	  }{
      \vlin{\bvtseqdrule}{}
	  {\vlsbr
	    [\vlfo{\atma}
		  {[<\atmbBlu;\mapPincToDi{\pincE'}>
		    ;<\vlone;\mapPincToDi{\pincF'}>]}
	    ;\natmbBlu]
	  }{
      \vliq{\eqref{align:unit-pa}
           ,\eqref{align:alpha-intro}}{}
	  {\vlsbr
	    [\vlfo{\atma}
		  {<[\atmbBlu;\vlone]
		    ;[\mapPincToDi{\pincE'};\mapPincToDi{\pincF'}]>}
	    ;\natmbBlu]
	  }{
      \vlin{\bvtrdrule}{}
	  {\vlsbr
	    [\vlfo{\atma}
		  {<\atmbBlu
		    ;[\mapPincToDi{\pincE'};\mapPincToDi{\pincF'}]>}
	    ;\vlfo{\atma}{\natmbBlu}]
	  }{
      \vliq{\eqref{align:unit-seq}}{}
	  {\vlfo{\atma}
	    {\vlsbr
	    [<\atmbBlu
		;[\mapPincToDi{\pincE'};\mapPincToDi{\pincF'}]>
	    ;\natmbBlu]
	    }
	  }{
      \vlin{\bvtseqdrule}{}
	  {\vlfo{\atma}
	    {\vlsbr
	    [<\atmbBlu
		;[\mapPincToDi{\pincE'};\mapPincToDi{\pincF'}]>
	    ;<\vlone;\natmbBlu>]
	    }
	  }{
      \vliq{\eqref{align:unit-seq}^2}{}
	  {\vlfo{\atma}
	    {\vlsbr
	    <\vlone
	      ;[<\atmbBlu;[\mapPincToDi{\pincE'};\mapPincToDi{\pincF'}]>
	      ;\natmbBlu]>
	    }
	  }{
      \vlin{\bvtseqdrule}{}
	  {\vlfo{\atma}
	    {\vlsbr
	      [<\atmbBlu;[\mapPincToDi{\pincE'};\mapPincToDi{\pincF'}]>
	      ;<\natmbBlu;\vlone>]
	    }
	  }{
      \vliq{\eqref{align:unit-pa}}{}
	  {\vlfo{\atma}
	    {\vlsbr
             <[\atmbBlu;\natmbBlu]
	      ;[\mapPincToDi{\pincE'};\mapPincToDi{\pincF'};\vlone]>
	    }
	  }{
      \vliq{\bvtatrdrule
           ,\eqref{align:unit-seq}}{}
	  {\vlfo{\atma}
	    {\vlsbr
	    <[\atmbBlu;\natmbBlu]
	      ;[\mapPincToDi{\pincE'};\mapPincToDi{\pincF'}]>
	    }
	  }{
      \vlhy{\vlfo{\atma}
	    {\vlsbr[\mapPincToDi{\pincE'};\mapPincToDi{\pincF'}]
	    }
      }}}}}}}}}}}
      }

%% file: PPi-example-LTS-02-reconstruction.tex
\vlderivation{
           \vlin{\pincpe}{\pincLabT\not\equiv\pinca}
		        {\pincLTSJud{\pincNu{\pinca}
		                            {(\pincPar{(\pincSec{\pincaRed}
		                                                {\pincSec{\pincbBlu}
						                                          {\pincE'}})}
					                          {(\pincSec{\pincnaRed}{\pincF'})})
					                }
					         \pincCong
					         \pincPar{\pincNu{\pinca}
		                                     {(\pincPar{(\pincSec{\pincaRed}
		                                                         {\pincSec{\pincbBlu}
						                                                  {\pincE'}})}
					                                   {(\pincSec{\pincnaRed}
					                                             {\pincF'})})
					                         }
					                 }
					                 {\pincNu{\pinca}
					                         {\pincZer}}
							}
			                {\pincPar{\pincNu{\pinca}{(\pincPar{(\pincSec{\pincbBlu}
							                                            {\pincE'})}
							                          {\pincF'})}
			                         }
			                         {\pincNu{\pinca}
			                         		 {\pincZer}}
			                 \pincCong
			                 \pincNu{\pinca}{(\pincPar{(\pincSec{\pincbBlu}
							                                    {\pincE'})}
							                          {\pincF'})}
							}
 			                {\pincPreT}
		        }{%%%\vliin
	    \vliin{\pinccom}{}
		      {\pincLTSJud{
				           \pincPar{\pincPar{(\pincSec{\pincaRed}
		                                              {\pincSec{\pincbBlu}
		                                                       {\pincE'}})}
				                            {(\pincSec{\pincnaRed}
				                                      {\pincF'})}
				                   }
				           		   {\pincZer}
				           \pincCong
		                   \pincPar{(\pincSec{\pincaRed}
		                                     {\pincSec{\pincbBlu}
		                                              {\pincE'}})}
				                    {(\pincSec{\pincnaRed}{\pincF'})}
				          }
			              {
			               \pincPar{(\pincSec{\pincbBlu}
			               			         {\pincE'})}
			               			                      {\pincF'}
			               \pincCong
			               \pincPar{\pincPar{(\pincSec{\pincbBlu}
 			               					          {\pincE'})}
			               				    {\pincF'}
			                       }
			                       {\pincZer}
						  }
						  {\pincPreT}
		     }
		     {%%%1.1-------
		      \vlin{\pincact}{}
		           {\pincLTSJud{\pincSec{\pincaRed}
		                                {\pincSec{\pincbBlu}
		                                         {\pincE'}}}
 			                   {\pincSec{\pincbBlu}
 			                            {\pincE'}}
		                       {\pinca}
		           }{\vlhy{}}
		     }%%%1.1-------
		     {%%%1.2-------
		      \vlin{\pincact}{}
		           {\pincLTSJud{\pincSec{\pincnaRed}
		                                {\pincF'}}
		                       {\pincF'}
		                       {\pincna}
		           }{\vlhy{}}
		     }
		 }%%%\vliin
}

%% file: conclusions.tex
\section{Final discussion, and future work}
\label{section:Final discussion, and future work}
This work shows that $\BVT$ \cite{Roversi:2010-LLCexDI,Roversi:TLCA11,Roversi:unpub2012-I}, which we can consider as a minimal extension of $\BV$ \cite{Gugl:06:A-System:kl}, is expressive enough to model concurrent and communicating computations, as expressed by the language  $\CCSR$, whose logic-based restriction con hide actions to the environment in an unusual flexible way, as compared to the restriction of Milner $ \CCS $.
The reason why, in various points, we have kept relating $\CCSR$ with a fragment of Milner $\CCS$ is twofold.
First, we start from the programme of~\cite{Brus:02:A-Purely:wd}, that shows the connections between $\BV$ and the smallest meaningful fragment of Milner $\CCS$.
Second, it is evident we can define $\BVTMin$ as follows. We take $\BVT\setminus\Set{\bvtrdrulein}$ and
we forbid clauses \eqref{align:alpha-intro}, and~\eqref{align:alpha-varsub} on its structures. So defined, $\BVTMin$ would be very close to the fragment of Milner $ \CCS $, which we have called $\CCSRM$, and which only contains restriction, and both sequential, and parallel composition. The reason is that $\BVTMin$ could simulate the two standard rules for restriction:
{\small
\[
\vlinf{}
      {\pincLabL\not\in\Set{\pinca,\pincna}}
      {\pincLTSJud{\pincNu{\pinca}{\pincE}}
		          {\pincNu{\pinca}{\pincE'}}
  		          {\pincLabL}}
      {\pincLTSJud{\pincE}
                  {\pincE'}
                  {\pincLabL}}
\qquad\qquad
\vlinf{}
      {\pincLabL\in\Set{\pinca,\pincna}}
      {\pincLTSJud{\pincNu{\pinca}{\pincE}}
		          {\pincNu{\pinca}{\pincE'}}
  		          {\pincLabT}}
      {\pincLTSJud{\pincE}
                  {\pincE'}
                  {\pincLabL}}
\]
} %\small
but not the rules $\pincpi$, and $\pincpe$ in~\eqref{equation:PPi-LTS-from-BVT}. However, in fact, \OpNameRen looks much closer to the hiding operator $\pincNuPi{\pinca}{\pincE}$ of $\pi$-calculus \cite{SangiorgiWalker01}.
Clause~\eqref{align:alpha-symm} ``is''
$\pincNuPi{\pinca}{\pincNuPi{\pincb}{\pincE}}\approx\pincNuPi{\pincb}{\pincNuPi{\pinca}{\pincE}}$.
Clause~\eqref{align:alpha-intro} generalizes
$\pincNuPi{\pinca}{\pincZer}\approx\pincZer$. The instance:
{\small
\vlstore{
\vlderivation{
 \vliq{\eqref{align:alpha-intro},\bvtrdrule}{}
      {\vlsbr[\vlfo{\pinca}{\pincE};\pincF]}{
 \vlhy{\vlfo{\pinca}{\vlsbr[\pincE;\pincF]}}}}
}
\begin{equation}
 \label{equation:monodirectional-scope-extrusion}
 \vlread
\end{equation}
} %\small
weakly corresponds to scope extrusion $\pincNuPi{\pinca}{(\pincPar{\pincE}{\pincF})}\approx\pincPar{\pincNuPi{\pinca}{\pincE}}{\pincF}$ which holds, in both directions, whenever $\pinca$ is not free in $\pincF$. We postpone the study of semantics and of the relation between $\CCSR$, and the corresponding fragment of $\pi$-calculus, to future work.
\par
Further future work we see as interesting, is about the generalization of Soundness. We believe that a version of Soundness where no restriction to \simpleprocess es holds. The reason is twofold.
First, thanks to the Splitting theorem of $\BVT$ \cite{Roversi:2010-LLCexDI,Roversi:TLCA11,Roversi:unpub2012-I} it is possible to prove that every \emph{proof} of $\BVT$ can be transformed in a \standard\ proof of $\BVT$. So, no need to restrict to \OpNameTen-free derivations of $ \BVT $ exists to have \standard\ proofs. Second, the reduction process looks working on \standard\ proofs as well, and no obstacle seems to exist to the application of inductive arguments analogous to those ones we have used to prove our current Soundness.
\par
We conclude with a remark on the ``missing'' Completeness. Our readers may have noticed the lack of any reference to a Completeness of $\BVT$, \wrt $\CCSR$. Completeness would say that $\BVT$ has enough derivations to represent any computation in the \lts\ of $ \CCSR $. Formally, it would amount to:
%%%%%%%%%%%%%%%%%
\begin{theorem}[\textbf{\textit{Completeness of $\BVT$}}]
\label{theorem:Completeness of BVT}
For every \processstructure\ $\pincE$, and $ \pincF $, if 
$\pincLTSJud{\mapPincToDi{\pincE}}
            {\mapPincToDi{\pincF}}
            {\mapDiToPinc{\strR}{\emptyset}}$, then
$\vlstore{
 \vlsbr[\pincE;\strR]
}\bvtInfer{\bvtDder}
          { \pincF
            \bvtJudGen{\BVT}{}
            \vlread
 	  }$.
\end{theorem}
Ideally, we leave the proof of Theorem~\eqref{theorem:Completeness of BVT} as an exercise. The system $\BVT$ is so flexible that, proving it complete, amounts to show that every rule of  $\CCSR$ is derivable in $\BVT$.

%% file: Appendix-BV2-commuting-conversions.tex
\section{Proof of \textit{commuting conversions in} $\Set{\bvtatrdrulein,\bvtatidrulein,\bvtseqdrulein,\bvtrdrulein}$\\
(Lemma~\ref{lemma:bvtatidrulein commuting conversions},
page~\pageref{lemma:bvtatidrulein commuting conversions})}
\label{section:Proof of lemma:bvtatidrulein commuting conversions}
The proof is, first, by cases on $\bvtrhorulein$, and, then, by cases on $\vlstore{\vlholer{\strS\vlsbr[\atma;\natma]}}\vlread$.
Fixed $\vlstore{\vlholer{\strS\vlsbr[\atma;\natma]}}\vlread$, the proof is by cases on $\strR$ which must contain a redex of
$\bvtatidrulein, \bvtseqdrulein$, or $\bvtrdrulein$, that, after
$\bvtatidrulein^{\bullet}$, leads to the chosen $\vlstore{\vlholer{\strS\vlsbr[\atma;\natma]}}\vlread$.
\par
We start with $\bvtrhorulein \equiv \bvtatidrulein$.
\begin{itemize}
% Case  U ~ { }
\item Let $\vlstore{\vlholer{\strS\vlsbr[\atma;\natma]}}
           \vlread
           \approx\vlsbr[\atma;\natma]$.
So,
$\vlstore{\vlsbr<\natma;[\atmb;\natmb]>}
 \vlsbr[\atma;\vlfo{\atmb}{\vlread}]$, and
$\vlstore{\vlsbr<\natma;[\atmb;\natmb]>}
 \vlsbr[\atma;\vlread]$
are the most relevant forms of $\strR$.
Others can be
$\vlstore{[\atmb;\natmb]}
 \vlsbr[\atma;<\natma;\vlfo{\atmb}{\vlread}>]$, and
$\vlstore{[\atmb;\natmb]}
 \vlsbr[[\atma;\natma];\vlfo{\atmb}{\vlread}]$, and
$\vlsbr<[\atma;\natma];[\atmb;\natmb]>$, and
$\vlstore{[\atmb;\natmb]}
 \vlsbr<[\atma;\natma];\vlfo{\atmb}{\vlread}>$.
\par
We fully develop only the first case
with $\strR\approx
     \vlstore{\vlsbr<\natma;[\atmb;\natmb]>}
     \vlsbr[\atma;\vlfo{\atmb}{\vlread}]$.
In it the derivation\\
  $\vlderivation{
   \vliq{\bvtatidrule
	,\eqref{align:unit-seq}
	,\eqref{align:alpha-intro}}{}
	{\vlsbr[\atma;\vlfo{\atmb}{<\natma;[\atmb;\natmb]>}]}{
   \vlin{\bvtatrdrule}{}
	{\vlsbr[\atma;\natma]}{
   \vlhy{\vlone}}}}$ transforms to
  $
%    \vlupsmash{
   \vlderivation{
   \vliq{\eqref{align:alpha-intro}
        ,\bvtrdrule}{}
	{\vlsbr[\atma;\vlfo{\atmb}{<\natma;[\atmb;\natmb]>}]}{
   \vliq{\eqref{align:unit-pa}
        ,\bvtseqdrule
        ,\eqref{align:unit-pa}}{}
	{\vlfo{\atmb}{\vlsbr[\atma;<\natma;[\atmb;\natmb]>]}}{
   \vlin{\bvtatrdrule}{}
	{\vlfo{\atmb}{\vlsbr<[\atma;\natma];[\atmb;\natmb]>}}{
   \vliq{\bvtatidrule
        ,\eqref{align:alpha-intro}}{}
	{\vlfo{\atmb}{\vlsbr[\atmb;\natmb]}}{
   \vlhy{\vlone}}}}}}
%    }%\vlupsmash
   $.
\par
If, instead,
$\vlstore{\vlholer{\strS\vlsbr[\atma;\natma]}}\vlread
 \approx
 \vlstore{\vlsbr<\natma;[\atmb;\natmb]>}
 \vlsbr[\atma;\vlread]$, then no instances of $\bvtrdrulein$ are required, but only
one of $\bvtseqdrulein$.

% Case  U ~ \vlsbr[\strS\vlhole;\strU']
\item Let $\strS\vlhole\approx\vlsbr[\vlholer{\strS'\vlhole};\strU']$. 
  \begin{itemize}
  %% 1st case
  \item If $\strR\approx
            \vlsbr[\vlholer{\strS'[\atma;\natma]}
                   ;\strS''[\atmb;\natmb]]$, with
           $\strU' \approx \vlsbr{\strS''[\atmb;\natmb]}$, then
  $\vlupsmash{
    \vlderivation{
    \vliq{\bvtatidrule}{}
	 {\vlsbr[\vlholer{\strS'[\atma;\natma]}
	        ;\strS''[\atmb;\natmb]]}{
    \vliq{\bvtatrdrule}{}
	 {\vlsbr[\vlholer{\strS'[\atma;\natma]}
	        ;\strU'']}{
    \vlhy{\vlsbr[\strR';\strU'']}}}}}$ transforms to
  $\vlupsmash{
   \vlderivation{
   \vliq{\bvtatrdrule}{}
	{\vlsbr[\vlholer{\strS'[\atma;\natma]}
	       ;\strS''[\atmb;\natmb]]}{
   \vliq{\bvtatidrule}{}
	{\vlsbr[\strR'
	       ;\strS''[\atmb;\natmb]]
	}{
   \vlhy{\vlsbr[\strR';\strU'']}}}}}$, for some $\strR'$, and $\strU''$. ===
  % 2nd case
  \item If $\strR\approx
            \vlsbr[\vlholer{\strS'[\atma;\natma]}
                   ;\strU']
            \equiv
            \vlsbr[\strS''[\atmb;\natmb]
                   ;\strU']$,
  then
  $\vlupsmash{
    \vlderivation{
    \vliq{\bvtatidrule}{}
	 {\vlsbr[\strS''[\atmb;\natmb];\strU']}{
    \vliq{\bvtatrdrule}{}
	 {\vlsbr[\vlholer{\strS'''[\atma;\natma]};\strU']}{
    \vlhy{\vlsbr[\strR';\strU']}}}}}$, 
for some $\vlholer{\strS'''\vlhole}$, which is $\strS''\vlsbr[\atmb;\natmb]$ with
$\vlsbr[\atmb;\natmb]$ replaced by $\vlone$, and $\strR'$,
transforms to
  $
%   \vlupsmash{
   \vlderivation{
   \vliq{\bvtatrdrule}{}
	{\vlsbr[\vlholer{\strS'[\atma;\natma]}
	       ;\strU']}{
   \vliq{\bvtatidrule}{}
	{\vlsbr[\strS''''\vlsbr[\atmb;\natmb]
	       ;\strU']
	}{
   \vlhy{\vlsbr[\strR';\strU']}}}}
%   }
   $
for some $\strS''''\vlhole$ which is $\strS'\vlsbr[\atma;\natma]$, with
$\vlsbr[\atma;\natma]$ replaced by $\vlone$.
  \end{itemize}

% Case  U ~ \vlfo{\atma}{\strS'\vlhole}
\item Let $\strS\vlhole\approx\vlfo{\atmc}{\vlholer{\strS'\vlhole}}$ where $\atmc$ may also
coincide to $\atma$, or $\atmb$. This case is analogous to the last point of the previous case, because
$\vlstore{\vlsbr[\atma;\natma]}
 \vlholer{\strS'\vlread} \equiv
 \strS''\vlsbr[\atmb;\natmb]$, for some $\strS''\vlhole$.

% Case  U ~ \vlsbr<\strS'\vlhole;\strU'>
\item Let $\strS\vlhole\approx\vlsbr<\vlholer{\strS'\vlhole};\strU'>$.

\begin{itemize}
 \item
If
$\strR\approx
      \vlstore{\vlholer{\strS'[\atma;\natma]}}
      \vlsbr<\vlread
             ;\strS''[\atmb;\natmb]>$, with
           $\strU' \approx \vlsbr{\strS''[\atmb;\natmb]}$, then
  $\vlupsmash{
    \vlderivation{
    \vliq{\bvtatidrule}{}
 	 {\vlsbr<\vlholer{\strS'[\atma;\natma]}
	        ;\strS''[\atmb;\natmb]>}{
    \vliq{\bvtatrdrule}{}
	 {\vlsbr<\vlholer{\strS'[\atma;\natma]}
	       ;\strU''>}{
    \vlhy{\vlsbr<\strR';\strU''>}}}}}$ transforms to
  $\vlupsmash{
   \vlderivation{
   \vliq{\bvtatrdrule}{}
	{\vlsbr<\vlholer{\strS'[\atma;\natma]}
	       ;\strS''[\atmb;\natmb]>}{
   \vliq{\bvtatidrule}{}
	{\vlsbr<\strR'
	       ;\strS''[\atmb;\natmb]>
	}{
   \vlhy{\vlsbr<\strR';\strU''>}}}}}$, for some $\strR'$, and $\strU''$.

  \item If $\vlstore{
	    \vlsbr<\vlholer{\strS'\vlsbr[\atma;\natma]};\strU'>
            }
            \strR\approx\vlread
            \equiv
            \vlsbr<\strS''[\atmb;\natmb];\strU'>$,
  then
  $\vlupsmash{
   \vlderivation{
   \vliq{\bvtatidrule}{}
	{\vlsbr<\strS''[\atmb;\natmb];\strU'>}{
   \vliq{\bvtatrdrule}{}
	{\vlsbr<\vlholer{\strS'''[\atma;\natma]};\strU'>}{
   \vlhy{\vlsbr<\strR';\strU'>}}}}}$, 
for some $\vlholer{\strS'''\vlhole}$, which is $\strS''\vlsbr[\atmb;\natmb]$, with
$\vlsbr[\atmb;\natmb]$ replaced by $\vlone$, and $\strR'$,
transforms to
  $
%   \vlupsmash{
   \vlderivation{
   \vliq{\bvtatrdrule}{}
	{\vlsbr<\vlholer{\strS'[\atma;\natma]}
	       ;\strU'>}{
   \vliq{\bvtatidrule}{}
	{\vlsbr<\strS''''\vlsbr[\atmb;\natmb]
	       ;\strU'>
	}{
   \vlhy{\vlsbr<\strR';\strU'>}}}}
%   }
   $
for some $\strS''''\vlhole$ which is $\strS'\vlsbr[\atma;\natma]$, with
$\vlsbr[\atma;\natma]$ replaced by $\vlone$.

\end{itemize}

\end{itemize}
%%%%%%%%%%%%%
Now we focus on the case with $\bvtrhorulein \equiv \bvtseqdrulein$.
\begin{itemize}
\item
Let 
$\vlholer{\strS\vlhole} 
 \approx 
 \strS'\vlsbr[<\strU';\strS''\vlhole>
             ;<\strU'';\strU'''>]$.
Then 
$\strR\approx
 \strS'\vlsbr[<\strU';\strS''[\atma;\natma]>
             ;<\strU'';\strU'''>]$, and\\
  $\vlderivation{
   \vlin{\bvtatidrule}{}
	    {\strS'\vlsbr[<\strU';\strS''[\atma;\natma]>
                     ;<\strU'';\strU'''>]}{
   \vlin{\bvtseqdrule}{}
	    {\strS'\vlsbr[<\strU';\strS''\,\vlscn{\vlone}>
	                 ;<\strU'';\strU'''>]}{
   \vlhy{\strS'\vlsbr<[\strU';\strU'']
                     ;[\strS''\,\vlscn{\vlone};\strU''']>}}}}$
                     transforms to
  $\vlupsmash{
   \vlderivation{
   \vlin{\bvtseqdrule}{}
	{\strS'\vlsbr[<\strU';\strS''[\atma;\natma]>
                    ;<\strU'';\strU'''>]}{
   \vlin{\bvtatidrule}{}
	{\strS'\vlsbr<[\strU';\strU'']
                    ;[\strS''[\atma;\natma];\strU''']>}{
   \vlhy{\strS'\vlsbr<[\strU';\strU'']
                    ;[\strS''\,\vlscn{\vlone};\strU''']>}}}}}$.
\item
Let $\vlholer{\strS\vlhole} 
     \approx \strS'\vlsbr[<\strS''\vlhole;\strU'>
                         ;<\strU'';\strU'''>]$.
This case is analogous to the previous one.
\end{itemize}
%%%%%%%%%%%%
Finally, let $\bvtrhorulein \equiv \bvtrdrulein$.
Then $\bvtrdrulein$ involves the redex of $\bvtatidrulein$ whenever 
$\vlholer{\strS\vlhole}$ is
$\vlstore{
  \vlholer{
   \strS'
   \vlsbr[\vlstore{\strS''\vlhole}
          \vlfo{\atma}{\vlread}
         ;\vlfo{\atma}{\strU'}]
  }
 }\vlread$.
So,
  $\strR\approx
  \strS'
  \vlsbr[\vlstore{\strS''\vlsbr[\atma;\natma]}
         \vlfo{\atma}{\vlread}
        ;\vlfo{\atma}{\strU'}]$, and
  $\vlderivation{
   \vlin{\bvtatidrule}{}
	{\strS'
         \vlsbr[\vlstore{\strS''\vlsbr[\atma;\natma]}
                \vlfo{\atma}{\vlread}
               ;\vlfo{\atma}{\strU'}]}{
   \vlin{\bvtrdrule}{}
	{\strS'
         \vlsbr[\vlstore{\strS''\,\vlscn{\vlone}}
                \vlfo{\atma}{\vlread}
               ;\vlfo{\atma}{\strU'}]}{
   \vlhy{\vlstore{
         \vlsbr[\strS''\,\vlscn{\vlone};\strU']}
         \strS'\vlfo{\atma}{\vlread}}}}}$ transforms to\\
$\vlderivation{
   \vlin{\bvtrdrule}{}
	{\strS'
         \vlsbr[\vlstore{\strS''\vlsbr[\atma;\natma]}
                \vlfo{\atma}{\vlread}
               ;\vlfo{\atma}{\strU'}]}{
   \vlin{\bvtatidrule}{}
	{\vlstore{
         \vlsbr[\strS''[\atma;\natma];\strU']}
         \strS'\vlfo{\atma}{\vlread}}{
   \vlhy{\vlstore{
         \vlsbr[\strS''\,\vlscn{\vlone};\strU']}
         \strS'\vlfo{\atma}{\vlread}}}}}$.

%% file: Appendix-invertible-structures-are-invertible.tex
\section{Proof of \textit{\textbf{A language of \invertiblestructure s}}
(proposition~\ref{proposition:Invertible structures are invertible},
page~\pageref{proposition:Invertible structures are invertible})}
\label{section:Proof of proposition:Invertible structures are
invertible}
This proof rests on Shallow splitting of \cite{Roversi:TLCA11,Roversi:unpub2012-I} whose statement we recall here.
\begin{proposition}[\textit{\textbf{Shallow Splitting}}]
\label{proposition:Shallow Splitting}
Let $\strR, \strT$, and $\strP$ be structures, and $\atma$ be a name, and  $ \bvtPder$ be a proof of $\BVT$.
\begin{enumerate}
\item\label{enum:Shallow-Splitting-seq}
If $\vlstore{\vlsbr[<\strR;\strT>;\strP]}
    \bvtInfer{\bvtPder}
             {\ \bvtJudGen{\BVT}{} \vlread}$, then there are
$\vlstore{\vlsbr<\strP_1;\strP_2>\bvtJudGen{\BVT}{} \strP}
\bvtInfer{\bvtDder}{\vlread}$, and
$\vlstore{\bvtJudGen{\BVT}{} {\vlsbr[\strR;\strP_1]}}
\bvtInfer{\bvtPder_1}{\ \vlread}$, and
$\vlstore{\bvtJudGen{\BVT}{} {\vlsbr[\strT;\strP_2]}}
 \bvtInfer{\bvtPder_2}{\ \vlread}$, for some
$\strP_1$, and $\strP_2$.

\item\label{enum:Shallow-Splitting-copar}
If $\vlstore{\vlsbr[(\strR;\strT);\strP]}
     \bvtInfer{\bvtPder}{\ \bvtJudGen{\BVT}{} \vlread}$,
then there are
$\vlstore{\vlsbr[\strP_1;\strP_2] \bvtJudGen{\BVT}{} \strP}
 \bvtInfer{\bvtDder}{\vlread}$, and
$\vlstore{\bvtJudGen{\BVT}{}{\vlsbr[\strR;\strP_1]}}
 \bvtInfer{\bvtPder_1}{\ \vlread}$, and
$\vlstore{\bvtJudGen{\BVT}{}{\vlsbr[\strT;\strP_2]}}
 \bvtInfer{\bvtPder_2}{\ \vlread}$, for some $\strP_1$, and $\strP_2$.

\item\label{enum:Shallow-Splitting-atom}
Let
$\vlstore{\vlsbr[\strR;\strP]}
 \bvtInfer{\bvtPder}{\ \bvtJudGen{\BVT}{} \vlread}$
with $\strR\approx\vlsbr[\atmLabL_1;\vldots;\atmLabL_m]$, such that
$ i\neq j $ implies $ \atmLabL_i \neq \vlne{\atmLabL_j} $, 
for every $ i,j \in\Set{1,\ldots,m}$, and $ m>0 $.
Then, for every structure $\strR_0$, and $\strR_1$, if
$\strR\approx\vlsbr[\strR_0;\strR_1]$,
there exists
$\vlstore{\vlne{\strR_1}
          \bvtJudGen{\BVT}{}
          \vlsbr[\strR_0;\strP]}
 \bvtInfer{\bvtDder}
          {\ \vlread}$.

\item\label{enum:Shallow-Splitting-fo}
If $\vlstore{\vlsbr[\vlfo{\atma}{\strR};\strP]}
    \bvtInfer{\bvtPder}{\ \bvtJudGen{}{} \vlread}$,
then there are
$\vlstore{\vlfo{\atma}{\strT} \bvtJudGen{\BVT}{} \strP}
 \bvtInfer{\bvtDder}{\vlread}$, and
$\vlstore{\bvtJudGen{\BVT}{} \vlsbr[\strR;\strT]}
 \bvtInfer{\bvtPder'}{\ \vlread}$, for some $\strT$.
\end{enumerate}
\end{proposition}
%%%%%%%%%%
Now, we reason by induction on
$\vlstore{\vlsbr[\vlne\strT;\strP]}
 \Size{\vlread}$,
proceeding by cases on the form of $\vlne\strT$.
%%% first case
\par
As a \emph{first case} we assume $\vlne\strT\approx\vlne{\vlone}$, and we cope with a
base case. The assumption becomes
$\vlstore{\vlsbr[\vlne{\vlone};\strP]}
 \bvtInfer{\bvtPder}{\ \bvtJudGen{}{} \vlread}$ which is exactly:
\[
\vlderivation                 {
\vlde{\bvtPder}{}
     {\vlsbr[\vlne{\vlone};\strP]
      \approx
      \strP}                 {
\vlhy{\vlne{\vlone}
      \approx
      \vlone
      }}}
\]
%%% second case
\par
As a \emph{second case} we assume
$\vlne\strT\approx
 \vlsbr[\natma_1;\vldots;\natma_m]$, and we cope with another
base case. The assumption becomes
$\vlstore{\vlsbr[[\natma_1;\vldots;\natma_m];\strP]}
 \bvtInfer{\bvtPder}{\ \bvtJudGen{}{} \vlread}$.
We conclude by Point~\ref{enum:Shallow-Splitting-atom} of Shallow Splitting
(Proposition~\ref{proposition:Shallow Splitting}) which implies
$\vlstore{\vlsbr(\atma_1;\vldots;\atma_m)
          \bvtJudGen{\BVT}{}
          \strP}
\vlread$.
%%% third case
\par
As a \emph{third case} we assume
$\vlne\strT\approx \vlsbr(\strR_1;\strR_2)$.
So, the assumption is
$\vlstore{\vlsbr[(\strR_1;\strR_2);\strP]}
 \bvtInfer{\bvtPder}{\ \bvtJudGen{}{} \vlread}$.
\par
Point~\ref{enum:Shallow-Splitting-copar} of Shallow Splitting
(Proposition~\ref{proposition:Shallow Splitting}) implies
$\vlstore{\vlsbr[\strP_1;\strP_2] \bvtJudGen{}{} \strP}
 \bvtInfer{\bvtDder}{\  \vlread}$, and
$\vlstore{\bvtJudGen{}{} \vlsbr[\strR_1;\strP_1]}
 \bvtInfer{\bvtQder_1}{\  \vlread}$, and
$\vlstore{\bvtJudGen{}{} \vlsbr[\strR_2;\strP_2]}
 \bvtInfer{\bvtQder_2}{\  \vlread}$,
for some $\strP_1, \strP_2$.
\par
Both $\strR_1$, and $\strR_2$ are invertible, and
$\vlstore{\Size{\vlsbr[\strR_1;\strP_1]}}
 \vlread$ $<
 \vlstore{\vlsbr[(\strR_1;\strR_2);\strP]}
 \Size{\vlread}$, and
$\vlstore{\Size{\vlsbr[\strR_2;\strP_2]}}
 \vlread$
 $<
 \vlstore{\vlsbr[(\strR_1;\strR_2);\strP]}
 \Size{\vlread}$.
So, the inductive hypothesis holds on
$\bvtQder_1$, and $\bvtQder_2$.
We get
$\vlstore{\vlne{\strR_1} \bvtJudGen{}{} \strP_1}
 \bvtInfer{\bvtEder_1}{\  \vlread}$, and
$\vlstore{\vlne{\strR_2} \bvtJudGen{}{} \strP_2}
 \bvtInfer{\bvtEder_2}{\  \vlread}$.
We conclude by:
\[
\vlderivation                           {
\vlde{\bvtDder}{}
     {\strP}                            {
\vlde{\bvtEder_1}{}
     {\vlsbr[\strP_1;\strP_2]}          {
\vlde{\bvtEder_2}{}
     {\vlsbr[\vlne{\strR_1};\strP_2]}   {
\vliq{\eqref{align:negation-pa}}{}
     {\vlsbr[\vlne\strR_1;\vlne\strR_2]}{
\vlhy{\vlsbr\vlne{(\strR_1;\strR_2)}}}}}}}
\]
%%% fourth case
\par
As a \emph{fourth case} we assume $\vlne\strT\approx\vlfo{\atma}{\strR}$ \ST, without loss of generality, $\atma\in\strBN{\vlfo{\atma}{\strR}}$. So, the assumption is
$\vlstore{\vlsbr[\vlfo{\atma}{\strR};\strP]}
 \bvtInfer{\bvtPder}{\ \bvtJudGen{}{} \vlread}$.
\par
Point~\ref{enum:Shallow-Splitting-fo} of Shallow Splitting
(Proposition~\ref{proposition:Shallow Splitting}) implies
$\vlstore{\strP}
 \bvtInfer{\bvtDder}{\ \vlfo{\atma}{\strT} \bvtJudGen{}{} \vlread}$, and
$\vlstore{\vlsbr[\strR;\strT]}
 \bvtInfer{\bvtQder}{\ \bvtJudGen{}{} \vlread}$, for some $\strT$.
\par
Both $\strR$ invertible, and
$\vlstore{\vlsbr[\strR;\strT]}
 \Size{\vlread}$
 $<
 \vlstore{\vlsbr[\vlfo{\atma}{\strR};\strP]}
 \Size{\vlread}$, imply the
induction holds on $\bvtQder$. We get
$\vlstore{\vlne\strR \bvtJudGen{}{} \strT}
 \bvtInfer{\bvtEder}{\  \vlread}$.
\par
So, we conclude that:
\[
\vlderivation                      {
\vlde{\bvtDder}{}
     {\strP}                       {
\vlde{\bvtEder}{}
     {\vlfo{\atma}{\strT}}         {
\vliq{\eqref{align:negation-fo}}{}
     {\vlfo{\atma}{\vlne\strR}}    {
\vlhy{\vlne{\vlfo{\atma}{\strR}}}} }}}
\]

%% file: Appendix-rightcontext-preserve-external-communication.tex
\section{Proving point~\eqref{enumerate:Trivial derivations and rightcontext s-03}
of \textit{\Processstructure s, \trivialderivation s and \rightcontext s}
(Proposition~\ref{proposition:Rightcontext s preserve communication},
page~\pageref{proposition:Rightcontext s preserve communication})}
\label{section:Proof of Rightcontext s can preserve external communication}
The proof is by induction on the size of $ \pincE $, proceeding by
cases on the form of $\vlholer{\strS'\vlhole}$, which, by assumption, is a \processstructure, so it can assume only specific forms.

\begin{itemize}

\item
%% Base case-----------------
The base case is
   $\vlholer{\strS'\vlhole}
    \approx
    \vlstore{\vlsbr<\vlhole;\strU>}
    \vlread$, for some $\strU$. So,
$\vlholer{\strS'\,\vlscn{\vlone}}
    \approx
    \vlstore{\vlsbr<\vlone;\strU>}
    \vlread
    \approx\strU$. 
Moreover, 
$\mapPincToDi{\pincE}=\vlsbr<\atmb;\strU>$ implies that
$ \pincE $ is $ \pincSec{\pincb}{\pincE'} $ for some $\pincE'$ \ST\
$\mapPincToDi{\pincE'}=\strU$.
Since we can prove:
{\small
$$
\vlinf{\pincact}{}
      {\pincLTSJud{\pincSec{\atmb} {\pincE'}
                  }
                  {\pincE'
                  }
                  {\pincb}}
     {}
$$
}%\small
we are done because $\mapPincToDi{\pincF}=\vlsbr<\vlone;\strU>\approx\strU=\mapPincToDi{\pincE'}$.
\par
A first remark is that we cannot have $\vlholer{\strS'\vlhole} \approx
\vlsbr<\vlholer{\strSb'\vlhole};\pincF>$ with
$\vlholer{\strSb'\vlhole}\not\approx\vlhole$. Otherwise
$\vlholer{\strS'\vlhole}$ would not be a \processstructure.
\par
A second remark is that $ \strU\approx \vlone $ does not pose any problem. In such a case 
$ \pincE $ is $ \pincSec{\atmb}{\pincZer} $, and we can write
$\pincLTSJud{\pincSec{\atmb} {\pincZer}}{\pincZer}{\pincb}$.

\item
%%1-------
Let $\vlholer{\strS'\vlhole} \approx
     \vlsbr[\vlholer{\strSb'\vlhole};\strU]$. 
The assumptions
$\mapPincToDi{\pincE}=\vlsbr[\vlholer{\strSb'\,\vlscn{\atmb}};\strU]$, and
$\mapPincToDi{\pincF}=\vlsbr[\vlholer{\strSb'\,\vlscn{\vlone}};\strU]$ imply that
$\pincE$ is $\pincPar{\pincE'}{\pincE''}$, and 
$\pincF$ is $\pincPar{\pincF'}{\pincE''}$, for some $\pincE' ,\pincE''$, and $\pincF'$ \ST\
$\mapPincToDi{\pincE'}=\vlholer{\strSb'\,\vlscn{\atmb}}$, and
$\mapPincToDi{\pincF'}=\vlholer{\strSb'\,\vlscn{\vlone}}$, and
$\mapPincToDi{\pincE''}=\strU$.
We can prove:
{\small
$$
\vlstore{
}
\vlinf{\pinccntxp}{}
      {\pincLTSJud{\pincPar{\pincE'}
                           {\pincE''}
                  }
                  {\pincPar{\pincF'}
                           {\pincE''}
                  }
                  {\pincLabL}
      }
      {\pincLTSJud{\pincE'}
                  {\pincF'}
                  {\pincLabL}
      }
$$
}%\small
because the premise holds thanks to the inductive hypotheses, also assuring the desired constraints on $\pincLabL$.

\item
%%2-------
Let $\vlholer{\strS'\vlhole} \approx \vlfo{\atma}{\vlholer{\strSb'\vlhole}}$. 
The assumptions
$\mapPincToDi{\pincE}=\vlfo{\atma}{\vlholer{\strSb'\,\vlscn{\atmb}}}$, and
$\mapPincToDi{\pincF}=\vlfo{\atma}{\vlholer{\strSb'\,\vlscn{\vlone}}}$ imply that
$\pincE$ is $\pincNu{\pinca}{\pincE'}$, and 
$\pincF$ is $\pincNu{\pinca}{\pincF'}$, for some $\pincE'$, and $\pincF'$ \ST\
$\mapPincToDi{\pincE'}=\vlholer{\strSb'\,\vlscn{\atmb}}$, and
$\mapPincToDi{\pincF'}=\vlholer{\strSb'\,\vlscn{\vlone}}$.
We can prove:
{\small
$$
\vlderivation{
\vlin{\bvtrhorule}{}
      {\pincLTSJud{\pincNu{\pinca}
                          {\pincE'}
                  }
                  {\pincNu{\pinca}
                          {\pincF'}
                  }
                  {\pincLabL'}
     }{
\vlhy{\pincLTSJud{\pincE'}
                 {\pincF'}
                 {\pincLabL}
      }}
}
$$
}%\small
\par\noindent
because the premise holds thanks to the inductive argument. Of course we
choose $\bvtrhorule$, depending on $\atma$. If $\atma\equiv\atmb$, then
$\bvtrhorule$ must be $\pincpi$, and $\pincLabL'\equiv\pincLabT$. Otherwise, if
$\atma\not\equiv\atmb$, then $\bvtrhorule$ must be $\pincpe$, and
$\pincLabL'\equiv\pincLabL$.

\end{itemize}

Point~\eqref{enumerate:Trivial derivations and rightcontext s-01} of this
Proposition excludes any further case.

%% file: Appendix-rightcontext-preserve-internal-communication.tex
\section{Proving point~\eqref{enumerate:Trivial derivations and rightcontext s-04}
of \textit{\Processstructure s, \trivialderivation s and \rightcontext s}
(Proposition~\ref{proposition:Rightcontext s preserve communication},
page~\pageref{proposition:Rightcontext s preserve communication})}
\label{section:Proof of Rightcontext s can preserve internal communication}
%%%%
The proof is by induction on the size of $ \pincPar{\pincE}{\pincF}$, proceeding by cases on
the forms of $\vlholer{\strS' \vlhole}$, and $\vlholer{\strS''\vlhole}$, which, by assumption, are \processstructure s, so they can assume only specific forms.
\begin{itemize}

\item
%% Base case-----------------
The base case has
$\vlholer{\strS'\vlhole} \approx \vlsbr<\vlhole;\strU'>$, and
$\vlholer{\strS''\vlhole} \approx \vlsbr<\vlhole;\strU''>$,
for some $\strU'$, and $\strU''$ every of which may well be $ \pincZer $.
So,
$\vlholer{\strS'\,\vlscn{\vlone}}
 \approx
 \vlsbr<\vlone;\strU'>
 \approx \strU'$, and
$\vlholer{\strS''\,\vlscn{\vlone}}
 \approx
 \vlsbr<\vlone;\strU''>
  \approx \strU''$.
The assumptions
$\mapPincToDi{\pincE} =\vlsbr<\atmb ;\strU' >$, and
$\mapPincToDi{\pincF} =\vlsbr<\natmb;\strU''>$, and
$\mapPincToDi{\pincE'}=\vlsbr<\vlone;\strU' >\approx\strU'$, and
$\mapPincToDi{\pincF'}=\vlsbr<\vlone;\strU''>\approx\strU''$ imply that
$\pincE =\pincSec{\pincb }{\pincE' }$, and
$\pincF =\pincSec{\pincnb}{\pincE'}$.
We can write:
{\small
$$
\vlderivation{
\vliin{\pinccom}{}
      {\pincLTSJud{\pincPar{(\pincSec{\pincb}
				      {\pincE'})
			    }
			    {(\pincSec{\pincnb}
				      {\pincF'})
			    }
		  }
		  {\pincPar{\pincE'}
			    {\pincF'}
		  }
		  {\pincLabT}
      }{%1
	\vlin{\pincact}{}
	     {\pincLTSJud{\pincSec{\pincb}
				  {\pincE'}}
			  {\pincE'}
			  {\pincb}
	     }{\vlhy{}}
	}
	{%2
	\vlin{\pincact}{}
	     {\pincLTSJud{\pincSec{\pincnb}
				  {\pincE'}}
			  {\pincE'}
			  {\pincnb}
	     }{\vlhy{}}
	}
}
$$
}%\small
We remark that neither
$\vlholer{\strS'\vlhole}
 \approx
 \vlsbr<\vlholer{\strSb'\vlhole};\strU'>$
with $\vlholer{\strSb'\vlhole}\not\approx\vlhole$,
nor
$\vlholer{\strS'\vlhole}
 \approx
 \vlsbr<\vlholer{\strSb''\vlhole};\strU''>$
with $\vlholer{\strSb''\vlhole}\not\approx\vlhole$, can hold. Otherwise
neither $\vlholer{\strS'\vlhole}$, nor
neither $\vlholer{\strS''\vlhole}$ could be \processstructure s.

\item
%%1-------
Let $\vlholer{\strS'\vlhole} \approx
     \vlsbr[\vlholer{\strSb'\vlhole};\strU']$. So,
$\vlholer{\strS'\,\vlscn{\vlone}}
 \approx
 \vlsbr[\strSb'\,\vlscn{\vlone};\strU']$.
The assumptions
$\mapPincToDi{\pincE } =\vlsbr[\vlholer{\strSb'\vlscn{\atmb}};\strU']$, and
$\mapPincToDi{\pincE'} =\vlsbr[\vlholer{\strSb'\vlscn{\vlone}};\strU']$
imply that
$\pincE  =\pincPar{\pincG_1 }{\pincG_2}$, and
$\pincE' =\pincPar{\pincG'_1}{\pincG_2}$
\ST\
$\mapPincToDi{\pincG_1 } =\vlholer{\strSb'\vlscn{\atmb}}$, and
$\mapPincToDi{\pincG'_1} =\vlholer{\strSb'\vlscn{\vlone}}$, and
$\mapPincToDi{\pincG_2}  =\strU'$.
\begin{itemize}
%% 1.1----------
  \item
  Let $\vlholer{\strS''\vlhole} \approx
       \vlsbr[\vlholer{\strSb''\vlhole};\strU'']$. So,
$\vlholer{\strS''\,\vlscn{\vlone}}
 \approx
 \vlsbr[\strSb''\,\vlscn{\vlone};\strU'']$.
The assumptions
$\mapPincToDi{\pincF } =\vlsbr[\vlholer{\strSb''\vlscn{\natmb}};\strU'']$, and
$\mapPincToDi{\pincF'} =\vlsbr[\vlholer{\strSb''\vlscn{\vlone}};\strU'']$
imply that
$\pincF  =\pincPar{\pincH_1 }{\pincH_2}$, and
$\pincF' =\pincPar{\pincH'_1}{\pincH_2}$
\ST\
$\mapPincToDi{\pincH_1 } =\vlholer{\strSb''\vlscn{\natmb}}$, and
$\mapPincToDi{\pincH'_1} =\vlholer{\strSb''\vlscn{\vlone}}$, and
$\mapPincToDi{\pincH_2}  =\strU''$.
We can prove:
{\small
$$
\vlderivation{
\vlin{\pinccntxp}{}
     {
      \pincLTSJud{\pincPar{\pincG_1}
                          {\pincPar{\pincG_2}
                                   {\pincPar{\pincH_1}
                                            {\pincH_2}}
                          }
                 }
                 {\pincPar{\pincG'_1}
                          {\pincPar{\pincG_2}
                                   {\pincPar{\pincH'_1}
                                            {\pincH_2}}
                          }
                 }
                 {\pincPreT}
     }{
\vlin{\pinccntxp}{}
     {
      \pincLTSJud{\pincPar{\pincG_1}
                          {\pincPar{\pincH_1}
                                   {\pincH_2}}
                 }
                 {\pincPar{\pincG'_1}
                          {\pincPar{\pincH'_1}
                                   {\pincH_2}}
                 }
                 {\pincPreT}
     }{
\vlhy{\pincLTSJud{\pincPar{\pincG_1}
                          {\pincH_1}
                 }
                 {\pincPar{\pincG'_1}
                          {\pincH'_1}
                 }
                 {\pincPreT}}}}
}
$$
}%\small
The premise holds thanks to the inductive hypothesis because both
$ \pincPar{\pincG_1}
          {\pincH_1}$ is smaller than
$ \pincPar{\pincG_1}
          {\pincPar{\pincG_2}
                   {\pincPar{\pincH_1}
                            {\pincH_2}}}$.
                            
%% 1.2----------
  \item
  Let $\vlholer{\strS''\vlhole} \approx
\vlsbr<\vlholer{\strSb''\vlhole};\strU''>$ with
$\vlholer{\strSb''\vlhole}\approx\vlhole$. Otherwise
$\vlholer{\strS''\vlhole}$ could not be a \processstructure.
So,
$\vlholer{\strS''\,\vlscn{\vlone}}
 \approx
 \vlsbr<\vlone;\strU''>\approx\strU''$.
The assumptions
$\mapPincToDi{\pincF } =\vlsbr<\natmb;\strU''>$, and
$\mapPincToDi{\pincF'} =\vlsbr<\vlone;\strU''> \approx \strU''$
imply that
$\pincF  =\pincSec{\pincnb}{\pincF'}$.
We can prove:
{\small
$$\vlderivation{
\vlin{\pinccntxp}{}
     {
      \pincLTSJud{\pincPar{\pincPar{\pincG_1}
                                   {\pincG_2}}
                          {(\pincSec{\natmb}
                                    {\pincF'})
                          }
                 }
                 {\pincPar{\pincPar{\pincG'_1}
                                   {\pincG_2}}
                          {\pincF'}
                 }
                 {\pincPreT}
      }{
\vlhy{\pincLTSJud{\pincPar{\pincG_1}
                          {(\pincSec{\natmb}
                                    {\pincF'})
                          }
                 }
                 {\pincPar{\pincG'_1}
                          {\pincF'}
                 }
                 {\pincPreT}}}
}
$$
}%\small
The premise holds thanks to the inductive hypothesis because
$\pincPar{\pincG_1}
         {(\pincSec{\natmb}
                   {\pincF'})}$ is smaller than
$\pincPar{\pincPar{\pincG_1}
                  {\pincG_2}}
         {(\pincSec{\natmb}
                   {\pincF'})}$.

  %%1.3----------
\item
Let $\vlholer{\strS''\vlhole} \approx \vlfo{\atma}{\vlholer{\strSb''\vlhole}}$,
for any $\atma$.
So,
$\vlholer{\strS''\,\vlscn{\vlone}}
 \approx
 \vlfo{\atma}{\vlholer{\strSb''\,\vlscn{\vlone}}}$.
The assumptions
$\mapPincToDi{\pincF } =\vlfo{\atma}{\vlholer{\strSb''\vlscn{\natmb}}}$, and
$\mapPincToDi{\pincF'} =\vlfo{\atma}{\vlholer{\strSb''\vlscn{\vlone}}}$
imply that
$\pincF  =\pincNu{\pincb}{\pincH }$, and
$\pincF' =\pincNu{\pincb}{\pincH'}$, for some 
$\pincH$, and $\pincH'$ such that 
$\mapPincToDi{\pincH } =\vlholer{\strSb''\vlscn{\natmb}}$, and
$\mapPincToDi{\pincH'} =\vlholer{\strSb''\vlscn{\vlone}}$.
We can prove:
{\small
$$
\vlderivation{
\vlin{\pinccntxp}{}
      {\pincLTSJud{\pincPar{\pincPar{\pincG_1}
                                    {\pincG_2}}
                           {\pincNu{\pincb}{(\pincH)}}
                  }
                  {\pincPar{\pincPar{\pincG'_1}
                                    {\pincG_2}}
                           {\pincNu{\pincb}{(\pincH')}}
                  }
                  {\pincPreT}
      }{
\vlhy{\pincLTSJud{\pincPar{\pincG_1}
                           {\pincNu{\pincb}{(\pincH)}}
                  }
                  {\pincPar{\pincG'_1}
                           {\pincNu{\pincb}{(\pincH')}}
                  }
                  {\pincPreT}
      }}
}
$$
}%\small
\par\noindent
The premise holds thanks to the inductive hypothesis because
$\pincPar{\pincG_1}
         {\pincNu{\pincb}{(\pincH)}}$ is smaller than 
$\pincPar{\pincPar{\pincG_1}
                  {\pincG_2}}
         {\pincNu{\pincb}{(\pincH)}}$.
\end{itemize}

% \newpage
\item
%%2--------------
Let $\vlholer{\strS'\vlhole} \approx
\vlsbr<\vlholer{\strSb'\vlhole};\strU'>$ with
$\strSb'\vlhole\approx\vlhole$. Otherwise
$\vlholer{\strS'\vlhole}$ could not be a \processstructure.
So, $\vlholer{\strS'\,\vlscn{\vlone}} \approx
\vlsbr<\vlone;\strU'>\approx\strU'$.
The assumptions
$\mapPincToDi{\pincE } =\vlsbr<\atmb;\strU'>$, and
$\mapPincToDi{\pincE'} =\vlsbr<\vlone;\strU'> \approx \strU''$
imply that
$\pincE  =\pincSec{\pincb}{\pincE'}$.

\begin{itemize}
%% 2.1----------
  \item
  We already considered the case with
  $\vlholer{\strS''\vlhole} \approx
       \vlsbr[\vlholer{\strSb''\vlhole};\strU'']$. It is enough to switch
  $\vlholer{\strS'\vlhole}$ and $\vlholer{\strS''\vlhole}$.

%% 2.2----------
  \item
  Letting $\vlholer{\strS''\vlhole} \approx
\vlsbr<\vlholer{\strSb''\vlhole};\strU''>$, with
$\strSb''\vlhole\approx\vlhole$, otherwise $\vlholer{\strS''\vlhole}$ could not be a
\processstructure, becomes the base case, we started with.

  %%2.3----------
  \item
  Let $\vlholer{\strS''\vlhole} \approx \vlfo{\atma}{\vlholer{\strSb''\vlhole}}$,
for any $\atma$. So,
$\vlholer{\strS''\,\vlscn{\vlone}}
 \approx
 \vlfo{\atma}{\vlholer{\strSb''\,\vlscn{\vlone}}}$ where, thanks to \eqref{align:PPi-structural-congruence}, we can always
be in a situation such that $\atma$ is different from every element
in $\strFN{\vlholer{\strS'\,\vlscn{\atmb}}}$.
The assumptions
$\mapPincToDi{\pincF } =\vlfo{\atma}{\vlholer{\strSb''\vlscn{\natmb}}}$, and
$\mapPincToDi{\pincF'} =\vlfo{\atma}{\vlholer{\strSb''\vlscn{\vlone}}}$
imply that
$\pincF  =\pincNu{\pincb}{\pincH }$, and
$\pincF' =\pincNu{\pincb}{\pincH'}$, for some 
$\pincH$, and $\pincH'$ such that 
$\mapPincToDi{\pincH } =\vlholer{\strSb''\vlscn{\natmb}}$, and
$\mapPincToDi{\pincH'} =\vlholer{\strSb''\vlscn{\vlone}}$.
We can prove:
{\small
$$
\vlderivation{
\vlin{\bvtrhorule}{}
      {\pincLTSJud{\pincPar{\pincNu{\atma}
                                   {(\pincSec{\atmb}
                                             {\pincE'})}
                           }
                           {\pincNu{\pinca}
				                   {\pincH}
                           }
                  }
                  {\pincPar{\pincNu{\atma}
                                   {\pincE'}}
			               {\pincNu{\pinca}{\pincH'}}
                  }
                  {\pincPreT}
      }{
\vlhy{\pincLTSJud{\pincPar{\pincSec{\atmb}
                                    {\pincE'}
                          }
                          {\pincH}
                 }
                 {\pincPar{\pincE'}
      			          {\pincH'}
                 }
                 {\pincPreT}
     }}
}
$$
}%\small
\par\noindent
where $\bvtrhorule$ can be any between $\pincpi$, and $\pincpe$.
The premise holds thanks to the inductive hypothesis because
$\pincPar{\pincSec{\atmb}{\pincE'}}{\pincH}$ is smaller than 
$\pincPar{\pincNu{\atma}{(\pincSec{\atmb}{\pincE'})}}{\pincNu{\pinca}{\pincH}}$.
\end{itemize}

\item
%%3--------------
Let $\vlholer{\strS'\vlhole} \approx \vlfo{\atma}{\vlholer{\strSb'\vlhole}}$ for a
given $\atma$.
So, $\vlholer{\strS'\,\vlscn{\vlone}} \approx
\vlfo{\atma}{\vlholer{\strSb'\,\vlscn{\vlone}}}$.
The assumptions
$\mapPincToDi{\pincE } =\vlfo{\atma}{\vlholer{\strSb'\vlscn{\atmb}}}$, and
$\mapPincToDi{\pincE'} =\vlfo{\atma}{\vlholer{\strSb'\vlscn{\vlone}}}$
imply that
$\pincE  =\pincNu{\pinca}{\pincG }$, and
$\pincE' =\pincNu{\pinca}{\pincG'}$, for some 
$\pincG$, and $\pincG'$ such that 
$\mapPincToDi{\pincG } =\vlholer{\strSb'\vlscn{\atmb}}$, and
$\mapPincToDi{\pincG'} =\vlholer{\strSb'\vlscn{\vlone}}$.

\begin{itemize}
  %% 3.1----------
  \item
  We already considered the case with
  $\vlholer{\strS''\vlhole} \approx
       \vlsbr[\vlholer{\strSb''\vlhole};\strU'']$. It is enough to switch
  $\vlholer{\strS'\vlhole}$ and $\vlholer{\strS''\vlhole}$.

  %% 3.2----------
  \item
  We already considered the case with $\vlholer{\strS''\vlhole} \approx
\vlsbr<\vlholer{\strSb''\vlhole};\strU''>$. It is enough to switch
  $\vlholer{\strS'\vlhole}$ and $\vlholer{\strS''\vlhole}$.

  %%3.3----------
  \item
  Let $\vlholer{\strS''\vlhole}
       \approx
       \vlfo{\atmc}{\vlholer{\strSb''\vlhole}}$,
for any $\atmc$. So,
$\vlholer{\strS''\,\vlscn{\vlone}}
 \approx
 \vlfo{\atmc}{\vlholer{\strSb''\,\vlscn{\vlone}}}$. 
The assumptions
$\mapPincToDi{\pincF } =\vlfo{\atmc}{\vlholer{\strSb''\vlscn{\natmb}}}$, and
$\mapPincToDi{\pincF'} =\vlfo{\atmc}{\vlholer{\strSb''\vlscn{\vlone}}}$
imply that
$\pincF  =\pincNu{\pincc}{\pincH }$, and
$\pincF' =\pincNu{\pincc}{\pincH'}$, for some 
$\pincH$, and $\pincH'$ such that 
$\mapPincToDi{\pincH } =\vlholer{\strSb''\vlscn{\natmb}}$, and
$\mapPincToDi{\pincH'} =\vlholer{\strSb''\vlscn{\vlone}}$. 
We need to consider the
following cases where
(i) $\bvtrhorule$ can be $\pincpi$, or $\pincpe$, and
(ii) the premise of all the given derivations exists thanks to the inductive arguments we have used so far in this proof.

      \begin{itemize}
      %%3.3.A----------
      \item
      As a first case let $\atma\equiv\atmc$, and $\atma,\atmc\not\equiv\atmb$.
      We can prove:
      {\small
      $$
      \vlderivation{
      \vlin{\bvtrhorule}{}
			    {
			     \pincLTSJud{\pincPar{\pincNu{\pinca}
							                 {\pincG}
						             }
									 {\pincNu{\pinca}
									 	     {\pincH}
									 }
					        }
							{\pincPar{\pincNu{\pinca}
							                 {\pincG'}
						             }
									 {\pincNu{\pinca}
									 	     {\pincH'}
									 }							
							}
							{\pincLabT}
			  }{
      \vlhy{
			     \pincLTSJud{\pincPar{\pincG
						             }
									 {\pincH
									 }
					        }
							{\pincPar{\pincG'
						             }
									 {\pincH'
									 }							
							}
							{\pincLabT}
      }}}
      $$
      }%\small
      
      We can proceed in the same way also when $\atma,\atmc\equiv\atmb$, the derivation becoming:
      {\small
      $$
      \vlderivation{
      \vlin{\bvtrhorule}{}
			    {
			     \pincLTSJud{\pincPar{\pincNu{\pincb}
							                 {\pincG}
						             }
									 {\pincNu{\pincb}
									 	     {\pincH}
									 }
					        }
							{\pincPar{\pincNu{\pincb}
							                 {\pincG'}
						             }
									 {\pincNu{\pincb}
									 	     {\pincH'}
									 }							
							}
							{\pincLabT}
			  }{
      \vlhy{
			     \pincLTSJud{\pincPar{\pincG
						             }
									 {\pincH
									 }
					        }
							{\pincPar{\pincG'
						             }
									 {\pincH'
									 }							
							}
							{\pincLabT}
      }}}
      $$
      }%\small

      %%3.3.C----------
      \item
      As a third case let $\atma\equiv\atmb$, and $\atmc\not\equiv\atmb$.
      we can prove:
      {\small
      $$
      \vlderivation{
      \vlin{\bvtrhorule}{}
	   {\pincLTSJud{\pincPar{\pincNu{\pincb}
      	                            {\pincG}}
	                        {\pincNu{\pincc}
	                              	{\pincH}}
	                \pincCong
	                \pincPar{\pincNu{\pincd}
	                      	        {\pincG\subst{\pincd}{\pincb}}}
	                	    {\pincNu{\pincd}
	                	            {\pincNu{\pincc}
	                	            	    {\pincH}}}
			       }
			       {\pincPar{\pincNu{\pincd}
			             	        {\pincG'\subst{\pincd}{\pincb}}}
			       	        {\pincNu{\pincd}
			       	                {\pincNu{\pincc}
			       	                        {\pincH'}}}
			        \pincCong
			        \pincPar{\pincNu{\pincb}
			             	        {\pincG'}}
			       	        {\pincNu{\pincc}
			       	                {\pincH'}}
			       }
			       {\pincLabT}
	   }{
       \vlhy
	   {\pincLTSJud{\pincPar{\pincG\subst{\pincd}{\pincb}}
	                	    {\pincNu{\pincc}{\pincH}}
			       }
			       {\pincPar{\pincG'\subst{\pincd}{\pincb}}
			       	        {\pincNu{\pincc}
			       	                {\pincH'}}
			       }
			       {\pincLabT}
	   }}
      }
      $$
      }%\small
      \par\noindent
      where $ \pincd $ neither occurs in $ \pincG $, nor it occurs in $\pincNu{\pincc}{\pincH} $ so that we can apply \eqref{align:PPi-structural-congruence}.
      \end{itemize}
\end{itemize}

\end{itemize}

%% file: Appendix-soundness-wrt-internal-communication.tex
\section{Proof of \textit{\textbf{Soundness \wrt\ internal communication}}
(Theorem~\ref{theorem:Soundness w.r.t. internal communication},
page~\pageref{theorem:Soundness w.r.t. internal communication})}
\label{section:Proof of theorem:Soundness w.r.t. internal communication}
\begin{itemize}
%%1---------
\item
As a base case, let $\mapPincToDi{\pincE}\approx
     \vlsbr[<\atmb;\mapPincToDi{\pincE}'>
           ;<\natmb;\mapPincToDi{\pincE}''>]$, for some process
$\pincE'$, and $\pincE''$. So,
$\pincE$ is 
$ \pincPar{(\pincSec{\pincb}
                    {\pincE'})}
          {(\pincSec{\pincnb}
                    {\pincE''})
          }$, and
$\vlholer{\strS'\vlhole}
 \approx\vlsbr<\vlhole;\mapPincToDi{\pincE'}>$, and
$\vlholer{\strS''\vlhole}
 \approx\vlsbr<\vlhole;\mapPincToDi{\pincE''}>$.
We can take
$\pincG$ to be
$ \pincPar{\pincE'}{\pincE''}$ because
$ \vlsbr[<\vlone;\mapPincToDi{\pincE'}>
        ;<\vlone;\mapPincToDi{\pincE''}>]
        \approx
  \vlsbr[\mapPincToDi{\pincE}'
        ;\mapPincToDi{\pincE}'']$.
We can write:
{\small
$$
\vlderivation{
\vliin{\pinccom}{}
      {\pincLTSJud{\pincPar{(\pincSec{\pincb}
				                     {\pincE'})
			               }
			               {(\pincSec{\pincnb}
				                     {\pincE''})
			               }
				  }
				  {\pincPar{\pincE'}
				    	   {\pincE''}
				  }
				  {\pincLabT}
      }{%1
	\vlin{\pincact}{}
	     {\pincLTSJud{\pincSec{\pincb}
				              {\pincE'}
				     }
					 {\pincE'}
					 {\pincb}
	     }{\vlhy{}}
	}
	{%2
	\vlin{\pincact}{}
	     {\pincLTSJud{\pincSec{\pincnb}
				              {\pincE'}}
					 {\pincE'}
					 {\pincnb}
	     }{\vlhy{}}
	}
}
$$
}%\small

%%2-------------
\item
Let $\vlstore{
     \vlsbr[\vlfo{\atmc}
                 {\vlholer{\strS' \,\vlscn{\atmb}}}
           ;\vlfo{\atmc}
                 {\vlholer{\strS''\,\vlscn{\natmb}}}
           ;\mapPincToDi{\pincE'''}]
     }
     \mapPincToDi{\pincE}\approx
     \vlread$, for some $\pincE'''$, and $\atmc$.
We remark that $\atmc$ is either different from $\atmb$ in both 
$\vlfo{\atmc}{\vlholer{\strS' \,\vlscn{\atmb}}}$, and
$\vlfo{\atmc}{\vlholer{\strS''\,\vlscn{\natmb}}}$, or it is equal to $ \atmb $ in both of them. Otherwise, we could not get to the premise of $\bvtatrdrulein$ in $\bvtDder'$.
So, 
$\pincE$ is 
$ \pincPar{\pincNu{\pincc}{\pincE'}
          }
	      {\pincPar{\pincNu{\pincc}{\pincE''}}
	               {\pincE'''}
	      }$, 
where $ \mapPincToDi{\pincE'} \approx \vlholer{\strS'\,\vlscn{\atmb}} $, and
$ \mapPincToDi{\pincE''} \approx \vlholer{\strS''\,\vlscn{\natmb}}$.
We can take 
$ \pincG $ as
$ \pincPar{\pincNu{\pincc}{\pincG'}
          }
	      {\pincPar{\pincNu{\pincc}{\pincG''}}
	               {\pincE'''}
	      }$,
because
$\vlstore{
 \vlsbr[\vlfo{\atmc}{\vlholer{\strS'\,\vlscn{\vlone}}}
       ;\vlfo{\atmc}{\vlholer{\strS''\,\vlscn{\vlone}}}
        ;\mapPincToDi{\pincE'''}]}
 \mapPincToDi{\pincG}\approx \vlread$,  with
 $ \mapPincToDi{\pincG'} \approx \vlholer{\strS'\,\vlscn{\vlone}} $, and
 $ \mapPincToDi{\pincG''} \approx \vlholer{\strS''\,\vlscn{\vlone}}$.
We can write:
{\small
$$
\vlderivation{
\vlin{\pinccntxp}{}
     {\pincLTSJud{\pincPar{\pincNu{\pincc}{\pincE'}
                           }
                           {\pincPar{\pincNu{\pincc}
                                            {\pincE''}}
                                    {\pincE'''}
                           }
                 }
                 {\pincPar{\pincNu{\pincc}{\pincG'}}
		                  {\pincPar{\pincNu{\pincc}{\pincG''}}
				                   {\pincE'''}
			              }
                 }
                 {\pincLabT}
     }{
\vlin{\bvtrhorule}{}
     {\pincLTSJud{\pincPar{\pincNu{\pincc}{\pincE'}
                          }
                          {\pincNu{\pincc}
                                  {\pincE''}
                          }
                 }
                 {\pincPar{\pincNu{\pincc}{\pincG'}}
     		              {\pincNu{\pincc}{\pincG''}}
                 }
                 {\pincLabT}
     }{
\vlhy{\pincLTSJud{\pincPar{\pincE'
                          }
                          {\pincE''
                          }
                 }
                 {\pincPar{\pincG'}
     		              {\pincG''}
                 }
                 {\pincLabT}
	 }}}
}
$$
}%\small
where $\bvtrhorule$ can be $\pincpe$, or $\pincpi$.
The premise follows from Point~\eqref{enumerate:Trivial derivations and rightcontext
s-04} of Proposition~\ref{proposition:Rightcontext s preserve communication}.

%%3-------------
\item
Let $\vlstore{
     \vlsbr[\vlholer{\strS' \,\vlscn{\atmb}}
           ;\vlholer{\strS''\,\vlscn{\natmb}}
           ;\mapPincToDi{\pincE'''}]
     }
     \mapPincToDi{\pincE}\approx
     \vlfo{\atmc}{\vlread}$, for some $\pincE'''$, and $\atmc$.
So, 
$\pincE$ is 
$\pincNu{\pincc}{(\pincPar{\pincE'}
	                      {\pincPar{\pincE''}
	                               {\pincE'''}
	                      })}$, 
where $ \mapPincToDi{\pincE'} \approx \vlholer{\strS'\,\vlscn{\atmb}} $, and
$ \mapPincToDi{\pincE''} \approx \vlholer{\strS''\,\vlscn{\natmb}}$.
We can take $\pincG $ as
$\pincNu{\pincc}{
 (\pincPar{\pincG'}
	      {\pincPar{\pincG''}
	               {\pincE'''}
	      })}$, 
because
$\vlstore{
 \vlsbr[\vlholer{\strS'\,\vlscn{\vlone}}
       ;\vlholer{\strS''\,\vlscn{\vlone}}
        ;\mapPincToDi{\pincE'''}]}
 \mapPincToDi{\pincG}\approx \vlfo{\atmc}{\vlread}$, with
$ \mapPincToDi{\pincG'} \approx \vlholer{\strS'\,\vlscn{\vlone}} $, and
$ \mapPincToDi{\pincG''} \approx \vlholer{\strS''\,\vlscn{\vlone}}$.
We can write:
{\small
$$
\vlderivation{
\vlin{\bvtrhorule}{}
     {\pincLTSJud{\pincNu{\pincc}{(\pincPar{\pincE'}
				                           {\pincPar{\pincE''}
				                                    {\pincE'''}
				                           })
                                 }
                  \approx
                  \pincPar{\pincNu{\pincc}{(\pincPar{\pincE'}
		  				                            {\pincPar{\pincE''}
		  				                                     {\pincE'''}
		  				                            })
		                                  }
                          }
                          {\pincNu{\pincc}
                                  {\pincZer}}
                 }
                 {\pincPar{\pincNu{\pincc}{(\pincPar{\pincG'}
		  				                            {\pincPar{\pincG''}
		  				                                     {\pincE'''}
		  				                            })
		                                  }
                          }
                          {\pincNu{\pincc}
                                  {\pincZer}}
                  \approx
                  \pincNu{\pincc}{(\pincPar{\pincG'}
 		                                   {\pincPar{\pincG''}
				                                    {\pincE'''}
			                               })
			                     }
                 }
                 {\pincLabT}
     }{
\vlin{\pinccntxp}{}
     {\pincLTSJud{\pincPar{\pincPar{\pincE'}
                                        {\pincE''}}
                               {\pincZer}
                      }
                      {\pincPar{\pincG'}
          		              {\pincPar{\pincG'}
          		                   	   {\pincZer}}
                      }
                      {\pincLabT}
     }{
\vlhy{\pincLTSJud{\pincPar{\pincE'}
                          {\pincE''}
                 }
                 {\pincPar{\pincG'}
     		              {\pincG''}
                 }
                 {\pincLabT}
	 }}}
}
$$
}%\small
where $\bvtrhorule$ can be $\pincpe$, or $\pincpi$.
The premise follows from Point~\eqref{enumerate:Trivial derivations and rightcontext
s-04} of Proposition~\ref{proposition:Rightcontext s preserve communication}.
\par
Of course, if
$\mapPincToDi{\pincE}
 \approx
     \vlsbr[\vlholer{\strS' \,\vlscn{\atmb }}
           ;\vlholer{\strS''\,\vlscn{\natmb}}
           ;\mapPincToDi{\pincE'''}]$, for some $\pincE'''$,
we can proceed as here above, dropping $\bvtrhorule$.
\end{itemize}
Assuming that $ (*) $ is the lowermost instance of $\bvtatrdrulein $ of $ \bvtDder $ excludes other cases that would impede getting to the premise of $(*)$ itself in a \trivialderivation\ like $ \bvtDder'$ has to be.

%% file: Appendix-soundness-wrt-external-communication.tex
\section{Proof of \textit{\textbf{Soundness w.r.t. external communication}}
(Theorem~\ref{theorem:Soundness w.r.t. external communication},
page~\pageref{theorem:Soundness w.r.t. external communication})}
\label{section:Proof of theorem:Soundness w.r.t. external communication}
We proceed on the possible forms that
$\mapPincToDi{\pincE}$ can assume, in relation with the form of $\strR$.
Point~\eqref{enumerate:Trivial derivations and rightcontext s-03} of
Proposition~\ref{proposition:Rightcontext s preserve communication} will help
concluding.
%%%%%%%%%
\begin{itemize}
\item[\textbf{First case.}]
We focus on $\bvtDder$ concluding with
$\vlstore{\vlsbr<\natmb;\strR>}
 \vlsbr[\mapPincToDi{\pincE};\vlfo{\atmb}{\vlread}]$.
In the simplest case,
Points~\eqref{enumerate:Trivial derivations and rightcontext s-01}, and
\eqref{enumerate:Trivial derivations and rightcontext s-02} of
Proposition~\ref{proposition:Rightcontext s preserve communication}
imply that either
$\vlstore{\vlholer{\strS'\,\vlscn{\atmb}}}
 \mapPincToDi{\pincE}\approx
 \vlsbr[\vlfo{\atmb}{\vlread};\mapPincToDi{\pincE''}]$, or
$\vlstore{\vlsbr<\atmb;\mapPincToDi{\pincE''}>}
 \mapPincToDi{\pincE}\approx
 \vlfo{\atmb}{\vlread}$,
for some $\pincE''$, and $\vlholer{\strS'\vlhole}$,
such that $\atmb\in\strFN{\vlholer{\strS'\,\vlscn{\atmb}}}$.

\begin{enumerate}
%1.1--------
\item
Let
$\vlstore{\vlsbr<\atmb;\mapPincToDi{\pincE''}>} 
    \mapPincToDi{\pincE}\approx\vlfo{\atmb}{\vlread}$. So,
$\pincE$ is $\pincNu{\pincb}{(\pincSec{\pincb}{\pincE''})}$. We can take
$ \pincG $ coinciding to $\pincE''$, because 
$\vlstore{\vlsbr<\vlone;\mapPincToDi{\pincE''}>}
 \vlfo{\atmb}{\vlread}
 \approx \vlfo{\atmb}{\mapPincToDi{\pincE''}}$. We can prove:
{\small
$$
\vlderivation                          {
\vlin{\pincpi}{}
     {\pincLTSJud{\pincNu{\atmb}
		         {(\pincSec{\atmb}
			               {\pincE''})}
                 }
                 {\pincNu{\pincb}
                         {\pincE''}
                 }
                 {\pincLabT}}             {
\vlin{\pincact}{}
     {\pincLTSJud{\pincSec{\atmb}{\pincE''}}
                 {\pincE''}
                 {\pincb}}{
\vlhy{}}}}
$$
}%\small

   %1.2--------
   \item
    Let
$\vlstore{\vlholer{\strS'\,\vlscn{\atmb}}}
 \mapPincToDi{\pincE}\approx
 \vlsbr[\vlfo{\atmb}{\vlread};\mapPincToDi{\pincE''}]$. 
So, $ \pincE $ is 
$ \pincPar{\pincNu{\atmb}{\pincE'}}
          {\pincE''} $ where
$ \mapPincToDi{\pincE'} \approx \vlholer{\strS'\,\vlscn{\atmb}}$.
We can take $\pincG$ as $\pincPar{\pincNu{\atmb}{\pincG'}}{\pincE''}$ where
$\vlstore{\vlholer{\strS'\,\vlscn{\vlone}}}
 \mapPincToDi{\pincG'} \approx \vlfo{\atmb}{\vlread}$.
We can prove:
{\small
$$
\vlderivation{
\vlin{\pinccntxp}{}
     {\pincLTSJud{\pincPar{\pincNu{\pincb}
                                  {\pincE'}}
	                      {\pincE''}
                 }
                 {\pincPar{\pincNu{\pincb}
                                  {\pincG'}}
                 	      {\pincE''}
                 }
                 {\pincPreT}}     {
\vlin{\pincpi}{}
     {\pincLTSJud{\pincNu{\pincb}{\pincE'}
                 }
                 {\pincNu{\pincb}{\pincG'}
                 }
                 {\pincPreT}}         {
\vlhy{\pincLTSJud{\pincE'
                 }
                 {\pincG'
                 }
                 {\pincb}}}}}
$$
}%\small
Point~\eqref{enumerate:Trivial derivations and rightcontext
s-03} of Proposition~\ref{proposition:Rightcontext s preserve communication} implies that the premise holds.
\end{enumerate}
%%%%%%%%
% \newpage
In fact, the most general situations that Points~\eqref{enumerate:Trivial
derivations and rightcontext s-01}, and \eqref{enumerate:Trivial derivations and
rightcontext s-02} of Proposition~\ref{proposition:Rightcontext s preserve
communication} imply are:
{\small
$$
\vlstore{
\vlsbr[\vlfo{\atma_1}
            {\vldots
             \vlfo{\atma_m}
                  {\vlholer{\strS' \,\vlscn{\atmb }}}
             \vldots}
      ;\mapPincToDi{\pincE'}]
}
\mapPincToDi{\pincE}\approx\vlread
\qquad\qquad\qquad\qquad
\vlstore{
\vlsbr<\atmb ;\mapPincToDi{\pincE'}>
}
\mapPincToDi{\pincE}\approx
\vlfo{\atma_1} {\vldots \vlfo{\atma_m} {\vlread} \vldots}
$$
}%\small
where $a_i \not\equiv a_j$, for every $1\leq i, j\leq m$, and
$\atmb\equiv a_i$, for some $1\leq i \leq m$.
We can resume to the situation we have just developed in
detail, by rearranging the occurrences of \OpNameRen, thanks to
congruence \eqref{align:PPi-structural-congruence}.

\item[\textbf{Second case.}] 
Let us assume that $\bvtDder$ concludes with
$\vlstore{\vlsbr<\natmb;\strR'>}\strR\approx\vlread$.
Points~\eqref{enumerate:Trivial derivations and rightcontext s-01}, and
\eqref{enumerate:Trivial derivations and rightcontext s-02} of
Proposition~\ref{proposition:Rightcontext s preserve communication}
imply either
$\vlstore{\vlsbr<\atmb;\mapPincToDi{\pincE'}>}
 \mapPincToDi{\pincE}\approx\vlread$, or
$\vlstore{\vlholer{\strS'\,\vlscn{\atmb}}}
 \mapPincToDi{\pincE}\approx\vlsbr[\vlread;\mapPincToDi{\pincE'}]$, where
$\atmb\in\strFN{\vlholer{\strS'\,\vlscn{\atmb}}}$.
Both combinations are simple sub-cases of the previous ones, just developed
in detail.
\end{itemize}